\documentclass[%
 amsmath,amssymb,
 aps, %physrev,
reprint,
pra,
]{revtex4-2}{}

\usepackage{graphicx}% Include figure files
\usepackage{dcolumn}% Align table columns on decimal point
\usepackage{bm}% bold math
\usepackage[ruled]{algorithm2e}  %linesnumbered
\usepackage{amsmath,amsfonts,amssymb,amsthm}
\usepackage{subcaption}
\usepackage[caption=false]{subfig}
\usepackage{comment}
\usepackage{booktabs}
\usepackage{xcolor}
\usepackage{arydshln}
\usepackage{setspace}
\usepackage[normalem]{ulem} 
\usepackage{quiver} % for the commutative diagram
\usepackage{multirow}
\usepackage{mathtools}
\mathtoolsset{showonlyrefs=true}

\renewcommand{\>}{\rangle}
\newcommand{\R}{\mathbb{R}}
\newcommand{\C}{\mathbb{C}}
\renewcommand{\H}{\mathbb{H}}
\newtheorem{theorem}{Theorem}
\newtheorem{prop}{Proposition}
\newtheorem{remark}{Remark}
\newtheorem{corollary}{Corollary}

\newcommand{\adjoint}{\dagger} % use \dagger for physics, * for applied math

\DeclareMathOperator*{\trace}{tr}
\DeclareMathOperator*{\base}{base}
\DeclareMathOperator{\mineig}{eig_{min}}
\SetKwInput{kwInput}{Input}
\SetKwInput{kwReturn}{Return}

\newcommand{\APOVM}[2]{A_#1^{(#2)}} % we can choose the formatting here...
\newcommand{\Ail}{\APOVM{k}{\ell}}   %  May 11 2025, let's use k not i. Save i for flat/linear indexing
\newcommand{\fPOVM}[2]{f_#1^{(#2)}} % we can choose the formatting here...
\newcommand{\fil}{\fPOVM{k}{\ell}}

\newcommand{\nPOVM}{M} % number of different POVMs. Let's use M or L (earlier we used K)
\newcommand{\nEachPOVM}{m} % size of each POVM
\newcommand{\nTotalPOVM}{\widetilde{m}} % product of the above, mM, 

\newcommand{\boldf}{\mathbf{f}}
\newcommand{\likelihood}{\mathcal{L}}

\newcommand{\f}{J}
\newcommand{\fvec}{j} % vectorized version

\renewcommand{\vec}[1]{\boldsymbol{#1}} % curly
\newcommand{\mat}[1]{\vec{#1}}
\newcommand{\U}{\mat{U}} 
\newcommand{\uu}{\vec{u}}  % don't use \u since it's used for accents
\renewcommand{\Vec}{\text{vec}}

\DeclareMathOperator{\hvec}{hvec} % h for Hermitian
\DeclareMathOperator{\hmat}{hmat}
\DeclareMathOperator{\rvec}{vec} % r for regular; synonym with \Vec
\DeclareMathOperator{\rmat}{mat} % synonhym with \Mat
\DeclareMathOperator{\rank}{rank}

\begin{document}
% \preprint{APS/123-QED}

\title{Rigorous Maximum Likelihood Estimation for Quantum States
}% 

\author{K. Aditi}
\email{email: kuchibhotla.aditi@colorado.edu}
% \altaffiliation[Also at ]{}%Lines break automatically or can be forced with \\
\author{Stephen Becker}%
\email{email: stephen.becker@colorado.edu}
\affiliation{%
University of Colorado Boulder, Boulder, CO, USA
}%

%\collaboration{CLEO Collaboration}%\noaffiliation

%\date{\today}% It is always \today, today,
             %  but any date may be explicitly specified

\begin{abstract}
Existing quantum state tomography methods are limited in scalability due to their high computation and memory demands, making them impractical for recovery of large quantum states. In this work, we address these limitations by reformulating the maximum likelihood estimation (MLE) problem using the Burer–Monteiro factorization, resulting in a non-convex but low-rank parameterization of the density matrix. We derive a fully unconstrained formulation by analytically eliminating the trace-one and positive semidefinite constraints, thereby avoiding the need for projection steps during optimization. Furthermore, we determine the Lagrange multiplier associated with the unit-trace constraint \emph{a priori}, reducing computational overhead. The resulting formulation is amenable to scalable first-order optimization, and we demonstrate its tractability using limited-memory BFGS (L-BFGS). Importantly, we also propose a low-memory version of the above algorithm to fully recover certain large quantum states with Pauli-based POVM measurements. Our low-memory algorithm avoids explicitly forming any density matrix, and does not require the density matrix to have a matrix product state (MPS) or other tensor structure. For a fixed number of measurements and fixed rank, our algorithm requires just $\mathcal{O}(d\log d)$ complexity per iteration to recover a $d \times d$ density matrix. Additionally, we derive a useful error bound that can be used to give a rigorous termination criteria. We numerically demonstrate that our method is competitive with state-of-the-art algorithms for moderately sized problems, and then demonstrate that our method can solve a $20$ qubit problem on a laptop in under 5 hours.
\end{abstract}

%\keywords{Suggested keywords}%Use showkeys class option if keyword
                              %display desired
\maketitle

%\tableofcontents

\section{\label{sec:intro}Introduction} 
Many emerging quantum technologies such as 
%Quantum information has emerged as a promising paradigm, enabling groundbreaking advancements in
quantum computing, quantum communication, quantum sensing, quantum machine learning and quantum cryptography rely on preparing suitable quantum states. A basic requirement is to 
%, and a wide range of other critical applications. To enable these applications, %suitable quantum states need to be prepared. More importantly, we need to
verify that the prepared state is indeed what was intended %to create 
and whether the output state of a quantum circuit is what is theoretically expected. Quantum State Tomography (QST) \cite{cramer2010efficient, Flammia_2012, PhysRevLett.111.020401, Guta_2020} is the process of inferring the state of a quantum system and is a fundamental task in quantum information. 
QST is the most comprehensive estimation technique, providing maximal information about an unknown quantum state. In situations where no prior structural assumptions can be made about the states or processes, QST is indispensable. It allows for the complete reconstruction of the state and facilitates the identification of deviations from the intended target state due to imperfect implementations. In particular, QST is a \emph{diagnostic} tool for experimental quantum platforms~\cite{eisert2020quantum}, in contrast to some methods which can only \emph{certify accuracy}.
%
%. As such, QST serves as the ultimate diagnostic tool for benchmarking and improving quantum experimental platforms \cite{eisert2020quantum}. 
It enables the detection of systematic errors, such as coherent phase drifts or correlated noise across qubits which cannot be identified 
via other means, 
%from expectation values of a small set of observables,
%making it indispensable
and thus has found use 
in high-precision settings \cite{PhysRevLett.111.120403},
and is sometimes referred to as the ``gold standard"~\cite{PRXQuantum.3.040310} for characterizing quantum systems.
% For
% It also excels at characterizing 
% non-linear properties like entanglement across bipartition of a system %, there is no better method known than full QST
% \cite{PRXQuantum.3.040310}. 
% For these reasons, and
Thus despite the considerable resource demands of QST, it remains worthwhile to develop algorithms that extend its applicability to larger system sizes. %, pushing the limits of what can be achieved with current hardware capabilities.
These efforts in QST are complementary to approaches that do partial tomography~\cite{PermutationallyInvariant_Toth2010}, directly estimate the fidelity with a target state~\cite{PhysRevLett.106.230501}, or perform a benchmarking, certification or verification task~\cite{pallister2018optimal,buadescu2019quantum,Huang2020,RandomizedBenchmarking_Dankert2009}.

% in certification and benchmarking and partial tomography.
% Quantum State Verification (QSV)~\cite{pallister2018optimal} and Quantum State Certification~\cite{buadescu2019quantum}
%   also compare with classical shadows~\cite{Huang2020}
% randomized benchmarking~\cite{RandomizedBenchmarking_Dankert2009}
%   DFE 

The major difficulty in QST is the exponential growth in the size of the Hilbert space dimension, $d = 2^n$, as the number of qubits $n$ grows. Any quantum system of importance that could solve a real problem, otherwise difficult to solve classically, would require hundreds to thousands of qubits. For instance, even our current Noisy Intermediate-Scale Quantum (NISQ) era is characterized by quantum processors containing up to $1000$ qubits. Given the large dimensionality, there are two challenges: an experimental challenge of collecting $\nTotalPOVM=\mathcal{O}(d^2)$ measurements needed to determine all the degrees of freedom in the $d\times d$ state matrix, and a computational challenge to compute and store the state matrix.

Addressing these challenges usually requires some structural assumptions on the target state, such as assuming it is pure (or nearly pure) or more generally low-rank~\cite{Guta_2020, JMLR:v16:koltchinskii15a, KUENG201788} or has a matrix product state (MPS) structure~\cite{cramer2010efficient, PhysRevLett.111.020401, Baumgratz_2013, lanyon2017efficient, PRXQuantum.4.040345}. There is a long line of work 
showing that these assumptions 
% making progress
can lead to a reduced demand for measurements, with $\nTotalPOVM=\mathcal{O}(n)$ for MPS and $\nTotalPOVM=\mathcal{O}(d)$ for low-rank assumptions, as well as reduced computational time. %, and there is a long line of work making progress.  
Our contribution is to present an algorithm that solves the traditional maximum likelihood estimation (MLE) formulation~\cite{vrehavcek2001iterative,QuantumMLE_Hradil1997} for a fixed measurement size $\nTotalPOVM$ in $\mathcal{O}(\nTotalPOVM d)$ flops (per iteration) and  $\mathcal{O}(\nTotalPOVM d)$ memory which, whenever $\nTotalPOVM <d^2$ and the target state is nearly pure, improves upon existing state-of-the-art schemes that require $\mathcal{O}(d^3)$ flops and memory.  We reformulate the convex MLE problem as a non-convex problem to exploit structure and remove constraints, but show how we can still recover provable {\em a posteriori} guarantees as if we had solved the original convex formulation.

Specifically, we assume the target state is low-rank or at least well-approximated by a low-rank state, and use the Burer-Monteiro (BM)~\cite{burer2005local} technique to parameterize the positive semidefinite density matrix as $\mat{U}\mat{U}^\dagger$, thus automatically enforcing the positive semi-definiteness (PSD) constraint, where $\mat{U}$ is a rectangular matrix of rank much lower than the dimension of the density matrix. This procedure is inherently more scalable as it eliminates the need for any eigenvalue decomposition, which would scale as $\mathcal{O}(d^3)$. Such BM approaches for QST have been proposed before, but either under unrealistic assumptions on the measurements~\cite{becker2013randomized} and/or without provable guarantees, and usually only for the least-squares objective which does not properly consider statistical uncertainty from Born's rule.

%In contrast, we show that the Lagrangian associated with the non-convex maximum-likelihood estimation (MLE) with the BM reformulation can be made unconstrained where the dual variable associated with the unit-norm constraint can be determined {\em a priori}. We also show that the optimization error of the proposed non-convex MLE objective can be bounded by a computationally tractable estimate. We solve the non-convex optimization problem using first-order methods such as L-BFGS and accelerated gradient descent. Our experimental results confirm that the optimization error is indeed upper bounded for our proposed setup and that our method is comparable to the best existing state-of-the-art solvers in moderate regimes, and our method is the only applicable solver for $n \ge 18$ qubits without resorting to the MPS assumption.

% To motivate the readers to follow the technical details of this article,
We now briefly outline the contributions and provide a summary of our results. To begin with, we show that the Lagrangian associated with the non-convex maximum-likelihood estimation (MLE) with the BM reformulation can be made unconstrained where the dual variable associated with the unit-norm constraint can be determined {\em a priori}. This fact substantially reduces the difficulty of the optimization and enables a well-structured approach, as opposed to relying solely on heuristic strategies. Furthermore, we show that the optimization error of the proposed non-convex MLE objective can be upper bounded by a computationally tractable quantity that depends on the current feasible estimate and the derivative of the likelihood function evaluated at this estimate. This result provides theoretical guarantees on the quality of the solution of our approach. 

Our numerical simulations confirm that the optimization error remains consistently upper bounded in practice, aligning with the theoretical analysis. Moreover, in moderate system sizes, our method achieves performance comparable to the best existing state-of-the-art solvers. %, validating both its accuracy and robustness.
Finally, we demonstrate both analytically and computationally that our low-memory method has reduced computational complexity, scaling only log-linearly with the system size $d$, making it the only applicable solver for $n \geq 18$ qubits, enabling full recovery of quantum states with specific structure. 
% tab:complexity

\subsection{\label{sec:related}Related work}

QST has a long history and we review some of the literature, categorizing it by the assumptions made on the quantum state.
%Over the years, the problem of QST has been tackled through various approaches based on linear inversion, low-rank recovery methods, SDPs, compressed sensing, and lately with non-convex methods.
We start with methods that add no assumptions on the quantum state. 
% \subsubsection{No assumptions}
The simplest approach is solving an unconstrained least-squares (LS) problem, leading to the linear inversion method. While computationally cheap, and even cheaper for specially structured measurements~\cite{Guta_2020}, the method
can generate an infeasible state that is not a density matrix because its trace is not unity and/or the eigenvalues are negative \cite{Qi_2013}; this is often partially addressed by the heuristic of \emph{a posteriori} normalizing (which does not correspond to Euclidean projection). From a statistical point of view, an enormous number of low-noise measurements are needed in order to have an accurate reconstruction.

The projective least-squares (PLS) method analyzed in \cite{Guta_2020} is a two-step procedure starting with an unconstrained least-squares estimate, as in the linear inversion method, followed by a second step to project the LS estimate onto a set of feasible density matrices which requires an eigenvalue decomposition and hence a computational cost of $\mathcal{O}(d^3)$. Their paper provides theoretical recovery guarantees, although the simulation results go up to only $5$ qubits. Other similar multi-step approaches include~\cite{Hou_2016, PhysRevLett.108.070502}. 

Going farther, the physical constraints can be added to the least-squares problem, leading to a constrained least-squares problem, though requiring an optimization solver and being significantly more computationally costly than unconstrained least-squares. Examples of this approach include~\cite{liu2011universal}. %\textcolor{red}{[TODO]}

While the least-squares objective is common, mostly due to its simplicity and computational ease, a more principled approach is maximum likelihood estimation (MLE) which takes into account both the statistical nature of the measurements and the physical constraints. While this approach is attractive due to its statistically motivated nature and the fact that it returns physically feasible states, and hence has been explored in many works including \cite{PhysRevA.64.052312, PhysRevA.63.020101, PhysRevA.75.042108, PhysRevA.61.010304, Acharya_2019}, it is significantly more computationally expensive than any of the unconstrained or constrained least-squares approaches.

% We discussed above several techniques used for quantum tomography. Another aspect that plays a crucial role in reconstructing the state from the measurements is the objective function. The choice of the objective function influences the reconstruction method, computational efficiency, and robustness to noise. Some of the most common objective used for optimization in QST are least-squares and maximum-likelihood estimation. The least-squares (LS) objective minimizes the error between the empirical frequencies (probabilities) and the predicted probabilities generated by the unknown state using Born's rule. The LS solution requires inversion of a system formed by $m$ linear equations. The solution may be not be a feasible quantum state since it can have negative eigenvalues if there are not enough measurements. On the other hand, Maximum Likelihood Estimate (MLE) is a very popular approach in QST accounting the probabilistic nature of quantum measurements and allows the inclusion of physical constraints to obtain a valid state ~\cite{PhysRevA.64.052312, PhysRevA.63.020101, PhysRevA.75.042108, PhysRevA.61.010304}. ML estimate is the state that is most likely to produce the observed measurements in the manifold of density matrices, or in other words, the state that maximizes the likelihood function. Unlike LS, MLE produces state estimate that does not have non-negative eigenvalues, so we always have a feasible solution. 

% \subsubsection{Structural assumptions}
A second class of methods places structural assumptions on the state matrix, beyond the guaranteed physical constraints of positive semi-definiteness and unit trace.  Because a pure state is rank one, and since most target states are intended to be pure, the most common structural assumption is that the target state is low-rank or close to a low-rank matrix, e.g., a pure state contaminated with a small amount of depolarizing noise or other uncertainty.

An early work to exploit the low-rank assumption was \cite{PhysRevLett.105.150401}, inspired by results from compressed sensing (CS)~\cite{1614066} on recovering sparse vectors from under-determined linear systems. Specifically, \cite{PhysRevLett.105.150401} proposed a compressed sensing formulation using the least-squares objective and a trace penalty, showing that the method could reduce the required number of measurements from $d^2$ to $\mathcal{O}(rd \log ^2  d)$, where $r$ is the rank of the estimated density matrix. The optimization problem was solved using singular value thresholding \cite{doi:10.1137/080738970} which required an eigenvalue decomposition at every step. Later improvements included showing that random Pauli matrices satisfy the \emph{restricted isometry property} (RIP) \cite{31416} giving more theoretical justification to the method, with further theoretical analysis on error bounds and sample complexity in \cite{Flammia_2012}. If the low-rank assumption is correct, and given enough measurements, the proposed low-rank least-squares estimators 
provably perform well 
%outperformed MLE
when values of certain parameters are set correctly.

% Compressed sensing is a technique to recover sparse vectors from fewer measurements \cite{1614066}. In \cite{PhysRevLett.105.150401}, compressed sensing based algorithm is proposed to reduce the measurement complexity from $d^2$ to $\mathcal{O}(rd \log ^2  d)$, where $r$ is the rank of the estimated density matrix. A trace-norm convex minimization problem is solved using singular value thresholding \cite{doi:10.1137/080738970}. The Pauli matrices were proved to satisfy \emph{restricted isometry property (RIP)} in \cite{31416}. This property was used in \cite{Flammia_2012} in providing theoretical analysis on error bounds and sample complexity. The proposed estimators outperformed MLE when values of certain parameters are set correctly. 
% The work in \cite{kalev2015quantum} draws a connection between CS and the positivity of density matrices arguing that the inclusion of the positivity constraint in the estimators leads to a low rank solution from within the space of all possible density matrices of any rank when RIP is satisfied and such a measurement setup is also robust to noise.

% \textcolor{red}{[TODO, cite some of MPS work from Mike Wakin's group]}
There exists a line of work \cite{cramer2010efficient, PhysRevLett.111.020401, Baumgratz_2013, lanyon2017efficient, PRXQuantum.4.040345, lidiak2022quantum, Qin_2024} that focuses on solving QST for quantum states generated by one-dimensional quantum computers such as the ones based on trapped ions. These states have a compact representation, and examples include matrix product states (MPS), matrix product operators (MPO), etc. Due to their compact representation, the number of copies of the state required scales polynomially rather than exponentially. Therefore, it is possible to build QST algorithms for MPS/MPO that are scalable to large systems. For example, the work in \cite{lidiak2022quantum} reconstructs a target state well-represented by a MPO using the tensor-train cross approximation. Furthermore, a lower bound on number of measurements is established in \cite{Qin_2024} for states that can be represented by MPO under different measurement schemes. The results show that state copies only polynomial in the number of qubits are required for a stable recovery.

Other literature has focused on using the low-rank assumption in order to speed up computation time of the estimator ~\cite{PhysRevA.61.010304, PhysRevA.95.062336, kyrillidis2017provable, photonics10020116, 7029630, JMLR:v16:koltchinskii15a}.
Most of the literature on non-convex optimization in QST uses the idea of factorizing the positive semi-definite density matrix using techniques such as Cholesky decomposition, Burer-Monteiro factorization \cite{Burer2003ANP}, \cite{burer2005local}, etc. The factorization technique introduced by Burer and Monteiro naturally takes care of the PSD constraint by construction. According to this technique, the density matrix can be expressed as $\rho = \mat{U}\mat{U}^\dagger$ for $U \in \mathbb{C}^{d \times r}$ which, if $r<d$, also imposes the prior knowledge of the density matrix being low-rank. This factorization leads to significant reduction in computation as well as memory storage. If $r$ is chosen sufficiently large, it was shown in \cite{burer2005local} that the BM factorization does not introduce any extraneous local minima despite being a non-convex problem, though in practice $r$ is usually set much smaller than this bound.
%where $\rho, U$ belong to the space of bounded operators $\mathcal{B(H)}$ on the $d$-dimensional Hilbert space $\mathcal{H}$. 

The method proposed in \cite{PhysRevA.95.062336} uses non-linear conjugate gradient (CG) for optimization in the factored domain (with $r=d$), then switches over to an accelerated first-order method in the original (convex) $\rho$ domain. The accelerated projected-gradient (APG) step requires computation of eigenvalues which are then projected on to the simplex, requiring $\mathcal{O}(d^3)$ flops and memory. %Moreover, the use of non-linear CG is unreliable when the line search is inaccurate.
The method outlined in \cite{kyrillidis2017provable} uses compressed sensing ideas along with the Burer-Monteiro factorization to promote a low-rank density matrix under Gaussian measurement error. The work further assumes certain regularity conditions,
and that the initial point is sufficiently close to the true solution.
% to be satisfied in order to find the global minimum with a constant step size when an initial point close to the global optimum is provided.
The work in \cite{photonics10020116} extends this idea further by proposing an accelerated version of the algorithm which provides accelerated linear convergence rate under common assumptions. In \cite{7029630}, a non-convex least-squares optimization is carried out by factorizing the density matrix using Burer-Monteiro factorization, where both $\mat{U}$ and $\mat{U}^\dagger$ are the optimization variables and updates are performed using Wirtinger derivatives. 

The recent work in \cite{PhysRevResearch.6.033034} addresses the problem of low-rank state reconstruction for full-rank target state. The so-called $P$-order
absolute state mapping combines the ideas of factorization and state projection approaches in QST to arrive at the proposed algorithm, and reconstruction is achieved through momentum-based Rprop (MRprop). They demonstrate that their proposed mapping with MRprop can successfully perform full-rank tomography of an 11-qubit mixed state under a minute on a GPU hardware. However, this is achieved under ideal conditions—without considering statistical noise in measurements and a specially structured set of positive operator-valued measurements (POVMs), and using $10^{10}$ shots per measurement.
% . Though there are some important details to be noted. They utilize $10^{10}$ measurement samples, an exceptionally large number, which ensures the accurate reconstruction of the target quantum state with high fidelity. The 11-qubit state they consider is noise-free contributing to the faster convergence of their method. 

% TODO, describe these
% \cite{RiemannianPRL2024}
% the good news is that this is not applicable to us! They do the least-squares formulation, and drop the trace and PSD constraints, and also assume the measurements have the restricted isometry property (RIP)... all very restrictive. I can write something about this in our draft to explain why it's not relevant.
% \cite{gaikwad2025gradient}
% this is also only for the least-squares formulation, so again we can ignore it! (or briefly mention it)
% We find that rank-controlled ansatzes in our stochastic mini-batch GD-QST algorithms effectively handle noisy and incomplete data sets, yielding significantly higher reconstruction fidelity than other methods. Simulations achieving full-rank seven-qubit QST in under three minutes on a standard laptop

The method proposed in \cite{RiemannianPRL2024} uses Riemannian gradient descent (RGD) to solve a least-squares formulation for low-rank quantum state tomography where the unit trace and PSD constraints are relaxed. Relaxing these constraints very likely leads to infeasible quantum states in between iterations, making it difficult to interpret performance metrics meaningfully. Further, the linear mapping obtained from Pauli measurements is also assumed to satisfy the RIP property. This assumption increases the measurement cost and is highly restrictive. In \cite{gaikwad2025gradient} a mini-batch gradient descent-based approach called GD-QST is proposed that uses three different rank controlled parameterizations, namely, Cholesky decomposition, Stiefel manifold and projective normalization to construct a valid density matrix, enabling rank-$r$  estimation of the unknown quantum state. The parameterization is optimized using a constrained least-squares loss function via vanilla gradient descent and the $Adam$ optimizer. 

Some works also study the effect of the number of POVMs and their repeated measurements on the performance of QST algorithms. The work in \cite{qin2024optimal} explores the trade-off between the number of POVMs 
$\nPOVM$ 
%$K$ 
and the number of repeated measurements (or ``shots'') $N$ required for each POVM,  
%$M$,
keeping the product $\nPOVM N$ %$KM$
constant. It was found that for recovery of low-rank states using Pauli measurements, 
it was generally better to increase $\nPOVM$ at the expense of lowering $N$; for large qubit systems, this could mean $N=1$ even.  It is in this regime where $N\approx 1$, i.e., very noisy measurements, that the MLE formulation makes much more sense than the least-squares approach.

\section{\label{sec:preliminaries}Preliminaries}
We typically use Greek letters like $\rho$ or $\sigma$ to denote density matrices of size $d\times d$ (representing possibly mixed quantum states) and capital Latin letters like $A_i$ to denote $d\times d$ elements of a POVM,
and capital bold, like $\mat{U}$, to indicate generic matrices (possibly complex and rectangular), and lowercase bold, like $\vec{u}$, to indicate vectors (possibly complex), with the non-bold version like $u_i$ to denote the $i^\text{th}$ components of $\vec{u}$.
We write $\rho^\adjoint=\bar{\rho}^\top$ to denote the Hermitian conjugate (or adjoint) of a matrix,
and denote $\H^d\subset \C^{d\times d}$ the real vector space of $d\times d$ complex Hermitian matrices.  We write $\rho\succeq \mat{0}$ to indicate that $\rho$ is Hermitian positive semi-definite.  Physical density matrices satisfy $\rho\succeq \mat{0}$ and $\trace(\rho)=1$. 
% On the space of $\C^{d\times d}$ we use the standard inner product $\langle \rho , \sigma \> = \trace(\rho^\adjoint \sigma)$. 
% % TODO: why is this real if psd?
The letters $\mat{I}$ and $\mat{0}$ indicate an identity matrix or zero matrix, respectively, of appropriate size for the context.
The fidelity between two states is defined as 
$F(\rho,\sigma) = \left(\trace \sqrt{\sqrt{\rho}\sigma \sqrt{\rho} } \right)^2
% $F(\rho,\sigma) = \left(\trace \left( \rho^{\frac12}\sigma \rho^{\frac12} \right)^{\frac12} \right)^2
\in[0,1]
$, which simplifies to $F(\rho,\sigma)=\trace(\rho \sigma)$ if either state is pure. The infidelity of $\rho$ and $\sigma$ is the non-negative number $1-F(\rho,\sigma)$, which is zero if and only if $\rho=\sigma$.
We write $\|\U\|_F$ to denote the Frobenius or Hilbert-Schmidt norm of a matrix.

% On the real vector space of 
% $\H^d$
% %$\C^{d\times d}$ 
% we use the standard inner product $\langle \rho , \sigma \>_\H = \Re[ \trace(\rho^\top \sigma) ]$ where $\Re$ indicates taking the real part (though if $\rho,\sigma\succeq \mat{0}$ then $\trace(\rho^\top \sigma)\in\R$). This is isometrically isomorphic to the vector space $\R^{d^2}$ with the usual Euclidean dot product under the mapping $\hvec: \H^d \to \R^{d^2}$ which vectorizes (in any fixed manner) the $d$ real entries on the diagonal along with $\sqrt{2}$ times the $d(d-1)/2$ real and $d(d-1)/2$ imaginary entries from the upper triangular portion of the input matrix; we define $\hmat=\hvec^{-1}$ which reshapes the real vector back into the Hermitian matrix. This vectorization is useful for deriving certain results as well as for numerical implementation of algorithms and has little consequential effect on the math since it is a unitary linear transformation.

% =====================================================================================
% =====================================================================================
\subsection{\label{sec:meas}Measurements}

Measurements are performed with respect to a positive operator-valued measurement (POVM), which is a collection of $m$ Hermitian matrices (or operators) $\{A_k\}_{k=1}^\nEachPOVM$ such that each 
$A_k\in\H^d$,
% $A_i\in\C^{d\times d}$,
$A_k\succeq \mat{0}$ and $\sum_k A_k = \mat{I}$.
Especially with ``coarse'' POVMs (i.e., $\nEachPOVM=2$), it is common to use multiple POVMs, in which case we write $\{ \{\Ail\}_{k=1}^\nEachPOVM\}_{\ell=1}^\nPOVM$ to denote the set of $\nPOVM$ POVMs, where for the sake of exposition we assume each POVM has the same size $\nEachPOVM$, though of course this is easily generalized to sets of POVMs with different sizes if needed.
We say a set of $M\ge 1$ POVMs $\{ \{\Ail \}_{k=1}^\nEachPOVM\}_{\ell=1}^\nPOVM$ is \emph{informationally complete} if for all states $\rho,\sigma$ with $\rho\neq \sigma$, there exists a POVM $\ell$ and index $k$ such that $\trace(\Ail\rho)\neq \trace(\Ail \sigma)$, or equivalently, if the span of $\{ \{\Ail \}_{k=1}^\nEachPOVM\}_{\ell=1}^\nPOVM$ is $\H^d$. %all $d\times d$ Hermitian matrices.

% TODO: binary outcomes, can we merge wlog?  I don't think so actually

A single measurement (or ``shot'') of the state $\rho$ with the POVM $\{\Ail\}_{k=1}^\nEachPOVM$ is an outcome $k\in\{1,\ldots,\nEachPOVM\}$ where the probability of observing $k$ is $\trace(\Ail \rho )$.  Given $N$ measurements on identically prepared copies of the state $\rho$, the $N$ outcomes thus follow a multinomial distribution, which reduces to a binomial distribution if $\nEachPOVM=2$.  A sufficient statistic for these $N$ measurements, with a fixed POVM indexed by some $\ell$, is the set of empirical frequencies $\{ \fil \}_{k=1}^\nEachPOVM$ where $f_k$ records the number of times outcome $k$ was observed divided by $N$.
For measurement with respect to multiple POVMs, for simplicity we assume each POVM has $N$ measurements, so we take a total of $\nPOVM N$ shots, and the vector of $\nEachPOVM \nPOVM$ real numbers  $\{\{ \fil \}_{k=1}^\nEachPOVM\}_{\ell=1}^\nPOVM$ is a sufficient statistic.
Assuming independent measurements,  as the number of shots $N$ increases $N\to\infty$, then 
$\fil \to \mathbb{E}[\fil] =  \trace(\Ail \rho)$ by the law of large numbers.

% TODO: describe Paulis as a special case, and also introduce the tetrahedral scheme...

% % In quantum tomography, the measurements are obtained as follows. 
% Given $N$ number of identically prepared copies of the unknown quantum state, %these are
% the state $\rho$ is 
% measured indirectly by a set of Hermitian operators known as POVM which satisfy these two properties: $A_i \succeq \mat{0}$ and $\sum_i A_i = I$.

Our simulations will focus on two types of POVM ensembles. The first setup is the collection of Pauli POVMs. 
The $n$-qubit Pauli matrices $W_1, \dots W_{d^2}$ are defined as the $n$-fold tensor product of the single qubit Pauli matrices $\mathcal{W} = \{I, \sigma_x, \sigma_y, \sigma_z \}$, where $W_\ell = \bigotimes_{j=1}^n w_j^{(\ell)}$ for $w_j^{(\ell)}\in\mathcal{W}$.  There are $4^n = d^2$ such matrices.  Each Pauli matrix defines a binary ($\nEachPOVM=2$) POVM defined as $\{\frac{1}{2} (I \pm W_k)\}$.  As is standard, we exclude the $n$-qubit identity matrix since its resulting POVM $\{I,0\}$ is uninformative, leaving $\nPOVM=d^2-1$ POVMs. This ensemble of Pauli POVMs is informationally complete. Since $\nPOVM\approx 4^n$ and $\nEachPOVM=2$, we refer to this as an example of a ``coarse'' POVM.

% We consider Pauli observables for measuring the unknown quantum state. The $n$-fold tensor product of Pauli matrices $w : = \{I, \sigma_x, \sigma_y, \sigma_z \}$ give rise to $n$-qubit Pauli matrices $W_1, \dots W_{d^2}$ of the form $\bigotimes w_i$, where $w_i \in w$ and $d= 2^n$. The $n$-qubit Pauli matrices are not POVMs (Positive Operator Valued Measure) since $\sum_i W_i \neq I$. Instead, we can measure the projectors associated with these Pauli matrices derived as $\{\frac{1}{2} (I \pm W_i)\}$ pertaining to two mutually exclusive outcomes. The trace of these projectors with the unknown quantum state gives the probabilities of obtaining outcomes $i$ as governed by the Born's rule. Since the projector associated with Pauli matrix $I$ is trivial, hence, we do not consider this as being part of the measurement. We distribute $N$ copies of the quantum state randomly between $(d^2 -1)$ 2-outcome POVMs. Upon measurement, we end up with $2(d^2 -1)$ empirical frequency values $f_k$, which normalized by the number of copies $N$, give us the empirical probability values $p_k$ following a multinomial distribution. Once the measurement data is obtained, a popular choice for the estimation of the unknown quantum state is Maximum Likelihood Estimation (MLE). The MLE strategy is to maximize the likelihood obtained from the empirical frequencies or minimize the negative log-likelihood.

% \textcolor{red}{TODO: add tetrahedral scheme info here}
The second setup is the collection of tetrahedral POVM. A single-qubit tetrahedral POVM consists of four elements given by $   A_k = \frac{1}{4} \left( I + \frac{\vec{e}_k \cdot (\sigma_x, \sigma_y, \sigma_z)}{\sqrt{3}} \right) $, where  $(\sigma_x, \sigma_y, \sigma_z)$  denotes the list of Pauli matrices, and  $\vec{e}_k \in \mathbb{R}^3$  are the four  vectors
$\vec{e}_1 = (1,1,1)$, $\vec{e}_2 = (-1,-1,1)$, $\vec{e}_3 = (-1,1,-1)$, and $\vec{e}_4 = (1,-1,-1)$. These vectors correspond to the directions pointing toward the centers of the faces of a regular tetrahedron inscribed in the Bloch sphere. Each  $A_k$  is a rank-1 positive operator, forming an informationally complete and symmetric measurement \cite{PhysRevA.70.052321}. For $n$-qubits, the single ($\nPOVM=1$) tetrahedral POVM consists of $\nEachPOVM=4^n$ measurements, and we refer to this as an example of a ``fine'' POVM.

We write $\nTotalPOVM=\nEachPOVM\nPOVM$ to denote the total number of POVM elements, and this is the fundamental quantity that controls the memory usage of algorithms. The statistical accuracy of the MLE will depend on $\nEachPOVM$, $\nPOVM$ and the number of shots $N$.

% =====================================================================================
% =====================================================================================
% \subsection{\label{}Optimization problem}
\subsection{\label{sec:likelihood}Likelihood function}
% \newcommand{\boldf}{\mathbf{f}}
% % \newcommand{\boldf}{\bm{f}}
We present the likelihood function in order to setup the maximum likelihood estimation (MLE) optimization problem~\cite{vrehavcek2001iterative,QuantumMLE_Hradil1997}, which is the focus of the paper.  The experimental data are the empirical frequencies $\{\{ \fil \}_{k=1}^\nEachPOVM\}_{\ell=1}^\nPOVM$ which can be represented by a $\nPOVM \times \nEachPOVM$ matrix.  To avoid both subscripting and superscripting, we will often use linear indexing of the frequencies and corresponding POVM matrices, i.e., writing $f_i$ for $\fil$ and $A_i$ for $\Ail$, where $i = k + (\ell-1)M$, for $i=1,\ldots,\nEachPOVM\nPOVM$. In particular, $\boldf=(f_k)_{i=1}^{\nTotalPOVM}$.

The likelihood for the measured data $\boldf$ generated by a state $\rho$ (assuming independent  measurements) is
\begin{equation*}
\mathcal{P}(\boldf \vert \rho) =  \prod_{i=1}^{\nTotalPOVM} p_i^{N f_i},\quad p_i = \trace(A_i\rho),
\end{equation*}
and hence the negative log-likelihood (dropping the scaling factor of $N$ for simplicity) is 
% % \newcommand{\likelihood}{\ell} % NEW, or \mathcal{L{}}
% \newcommand{\likelihood}{\mathcal{L}}
\begin{equation}\label{eq:MLEobjective}
    \likelihood(\rho) = -\sum_{i=1}^{\nTotalPOVM} f_i \log( \trace(A_i \rho) ).
\end{equation}
The domain of $\likelihood$ is a subset of $\C^{d\times d}$, so it is not immediately obvious that the output is real valued, but this is true since $A_i$ and $\rho$ are both Hermitian positive semi-definite. We assume that $A_i \neq 0$ since $A_i=0$ is not physically interesting and would mean $\likelihood$ has empty domain; we have $f_i \ge 0$ and assume not all $f_i = 0$.
The gradient of the negative log-likelihood is easy to derive via calculus:
\begin{equation}\label{eq:gradJ}
    \nabla \likelihood: \H^d \to \H^d,\quad 
    \nabla \likelihood(\rho) = -\sum_{i=1}^{\nTotalPOVM} \frac{f_i}{\trace(A_i \rho)} A_i.
\end{equation}

The MLE optimization problem is to maximize the likelihood, or equivalently minimize the negative log-likelihood, subject to being a physically realizable quantum state, i.e., 
\begin{align} 
    & %\underset{\rho}{\min}
    \min_{\rho\in\H^d}
    \quad \likelihood(\rho) \label{cvx-obj} \\ %\tag{MLE} \\
    & \text{subject to} \quad \rho \succeq \mat{0}, \ \trace(\rho) = 1. \notag % \label{cvx-cons}.
\end{align}
This is a convex optimization problem and is similar to a semi-definite program (SDP), though not exactly a SDP due to the $\log$ terms in the objective. For moderate dimensions, it can be solved to high accuracy using off-the-shelf optimization software.  Note that the trace constraint is essential: one cannot solve the problem without it and then, say, project or normalize afterwards.

When we know that the physical state is low-rank, we may also be interested in the following variant of the MLE (which is non-convex if $r<d$):
\begin{align} 
    & %\underset{\rho}{\min}
    \min_{\rho\in\H^d}
    \quad \likelihood(\rho) \label{eq:MLE-rankr} \\ %\tag{$r$-MLE}\\
    & \text{subject to} \quad \rho \succeq \mat{0}, \ \trace(\rho) = 1, \ \rank(\rho)\le r. \notag % \label{cvx-cons}.
\end{align}

% \begin{equation}
%     \mathcal{P}(\boldf \vert \rho) = \underset{i}{\Pi} \ p_i^{f_i/N} =  \underset{i}{\Pi} \ \trace(A_i \rho)^{f_i/N}
% \end{equation}
% where, $i \in \{1, 2, \dots 2d^2\}$. The negative log-likelihood $F(\rho) = -\frac{1}{N} \mathcal{P}(M \vert \rho) $ is, 

% \begin{equation}
%     F(\rho) = - \sum_i  f_i \log(\<A_i \mid  \rho\>)
% \end{equation}
% where, $\langle \sigma \mid \rho \rangle := \trace( \sigma^* \rho )$ and $\langle \sigma \mid A \mid  \rho \rangle := \trace( \sigma^* A \rho )$. The objective is to maximize the likelihood function $\mathcal{P}(M \vert \rho) $ or equivalently, minimize the negative log-likelihood function $F(\rho)$. Therefore, the objective function becomes as follows,
% \begin{align}
%     & \underset{\rho}{\min} \quad F(\rho) \label{cvx-obj} \\
%     & \text{subject to} \quad \rho \succeq \mat{0}, \ \trace(\rho) = 1 \label{cvx-cons}.
% \end{align}

% \begin{remark}
% One cannot solve Eq.~\eqref{cvx-obj} by ignoring the $\trace(\rho)=1$ constraint, finding a ``solution'' $\rho^\star$, and then adjusting $\rho^\star$ to be unit trace (by either scaling or projecting).   TODO...

% TODO
% \end{remark}

Below we state a useful proposition that helps clarify when we expect the MLE solution to match the true physical state.
\begin{prop} [Solution of the MLE with infinite measurements] \label{prop:1}
If there is a physical quantum state $\rho_0$ such that $f_i = \trace(A_i \rho_0)$ (i.e., $N\to\infty$) for all $i=1,\ldots,\nTotalPOVM$, then $\rho_0$ is a solution to the MLE optimization problem. If the set of POVMs is informationally complete, then furthermore $\rho_0$ is the unique solution.
\end{prop} 
\begin{proof}
We start with the case of a single POVM. 
For any physical state $\rho$, the vector $\vec{p}$ with $p_i = \trace(A_i\rho)$ defines a valid discrete probability distribution due to the properties of POVMs.  Since $\vec{p}$ and $\vec{f}$ are probability distributions, it holds that $-\sum_i f_i \log f_i \le -\sum_i f_i \log p_i$ via Gibbs' inequality, which implies that $p_i=f_i$ is optimal, hence $\rho_0$ is a minimizer. For multiple POVMs, we can concatenate them together and normalize by $1/\nPOVM$ to make a new larger effective POVM, and hence the same reasoning applies.  For uniqueness, we use that Gibbs' inequality is an equality if and only if $p_i=f_i$ for all $i$, hence $\vec{f}$ is the unique minimizer, and thus if the set of POVMs is informationally complete, no other physical state could generate the same set of measurements $\vec{f}$.
\end{proof}

\begin{remark}[Error metrics]
Some works in the literature on solving MLEs only report the (in)fidelity with the target state, but if $N <\infty$ then the physical state that generated the measurement is \emph{not} the solution to the MLE, hence the infidelity should not be $0$, and in this case measuring the fidelity conflates the optimization error with the statistical error. %Given $N<\infty$ shots, the MLE estimator is not the same as the physical state (the statistical or estimation error).
Since our focus is on computing the MLE, we either compare with the true MLE solution (computed using an established reference solver), or take $N=\infty$ and compare with the true physical state that generated the measurements, or compare the values of $\likelihood$.
% use a metric that disentangles the optimization error from this statistical error.  TODO... (depends: do we just take $N=\infty$??)
\end{remark}

% 5/19/2025, let's put this in the appendix; we can also then defer mentioning hmat until later

% \newcommand{\vecrho}{\vec{x}} % choose something better>
% It will also be convenient to examine $\likelihood$ when we vectorize the input. Letting $\vecrho = \hvec(\rho)$ and $\vec{a}_i = \hvec(\overline{A_i})\in\R^{d^2}$, we define
% \begin{equation}\label{eq:grad_l}
% \ell(\vecrho) = \likelihood(\hmat(\vecrho)) = -\sum_{i=1}^{\nEachPOVM\nPOVM} f_i \log( \vec{a}_i^\top\vecrho ), \quad \text{hence}\; \nabla \ell:\R^{d^2}\to\R^{d^2}, \;
%     \nabla \ell(\vecrho) = -\sum_{i=1}^{\nEachPOVM\nPOVM} \frac{f_i}{\vec{a}_i^\top\vecrho} \vec{a}_i.
% \end{equation}
% The Hessian of $\likelihood$, $\nabla^2 \likelihood$, is a $d\times d\times d\times d$ tensor whereas the Hessian of $\ell$ is a $d^2\times d^2$ symmetric matrix that is easier to work with. Specifically,
% \begin{equation}\label{eq:hess_l}
%  \nabla^2 \ell:\R^{d^2}\to\R^{d^2 \times d^2}, \;
%     \nabla^2 \ell(\vecrho) = \sum_{i=1}^{\nEachPOVM\nPOVM} \frac{f_i}{\left(\vec{a}_i^\top\vecrho\right)^2} \vec{a}_i\vec{a}_i^\top \succeq \mat{0}.
% \end{equation}    
% Since $\nabla^2 \ell(\vecrho) \succeq \mat{0}$, $\ell$ is convex, and since $\likelihood(\rho)=\ell(\hvec)$ is the composition of a convex function with a linear function, $\likelihood$ is also convex.

% =====================================================================================
% =====================================================================================
\subsection{Burer-Monteiro factorization}
% Let's add this here before the gradient calculation
% TODO, introduce basic idea

% original  \cite{Burer2003ANP,burer2005local}
% Recent works: \cite{levin2025effect,cifuentes2021burer,boumal2016non,WaldspurgerBM,AravkinBM}
% and https://arxiv.org/abs/2107.03877 (Kileel)

The MLE problem in Eq.~\eqref{cvx-obj} is similar to a SDP, and can be solved by similar techniques. Algorithms for solving general-purpose SDPs have been continually improving, but generally scale superlinearly in $d^2$ and are not efficient enough for QST when $n\gtrsim 10$.
% In the last few years, there have been really efficient algorithms to solve SDPs. But they do suffer when the problem is large and not ideal for many practical applications such as QST. 
For improved speed and decreased memory requirements, we turn to 
the factorization technique introduced by Burer and Monteiro \cite{Burer2003ANP,burer2005local} which naturally takes care of the PSD constraint by construction. The technique parameterizes the density matrix as $\rho = \mat{U}\mat{U}^\adjoint$ for $\mat{U} \in \mathbb{C}^{d \times r}$; choosing $r\ll d$ imposes the prior knowledge of the density matrix being low-rank. As we will see, when $r$ is small, this factorization leads to significant reduction in computation as well as memory storage. 
The MLE objective defined in Eq.~\eqref{eq:MLEobjective}, modified with the Burer-Monteiro factorization, is 
\begin{align} \label{eq:J}
    \f: \C^{d\times r} \to \R,\quad
    \f(\U) &= \likelihood(\U\U^\adjoint) \\ 
    &=  -\sum_{i=1}^{\nTotalPOVM} f_i \log( \trace(\U^\adjoint A_i \U ) ),
\end{align}
and has gradient 
\begin{align} \label{eq:J_grad}
    \nabla\f: \C^{d\times r} \to \C^{d\times r},\quad
    \nabla \f(\U) &=   -2\sum_{i=1}^{\nTotalPOVM} \frac{f_i}{ \trace(\U^\adjoint A_i \U ) } A_i \U \\
    &= 2\nabla \likelihood(\U\U^\adjoint)\U
\end{align}
where use the circular invariant property of trace to rewrite
$\trace(A_i \U\U^\adjoint)$
in the more computationally efficient form $\trace(\U^\adjoint A_i \U)$.
% The Hessian is a little unwieldy to work with, but for reference we state it below for the special case when $r=1$ (and write $\vec{u}$ rather than $\U$ to indicate this):
% \begin{equation*}
%     \nabla^2\f(\vec{u}) =   2\sum_{i=1}^{\nTotalPOVM} f_i\left( -\frac{1}{\vec{u}^\adjoint  A_i\vec{u}} A_i + \frac{2}{(\vec{u}^\adjoint  A_i  \vec{u})^{2}} A_i\vec{u} (\vec{u}^\adjoint  A_i) \right) \in \C^{d\times d}.
% \end{equation*}

% TODO (elsewhere): review those two new papers we found- added in the related work section.

The Burer-Monteiro factorization always destroys convexity, since if $\U$ is optimal, %meaning that $\rho=\U\U^\adjoint$ is an optimal point for the original convex problem,
then so is $\U\mat{Q}$ for any unitary matrix $\mat{Q}\in\C^{d\times d}$, but the set of unitary matrices is not convex. However, this non-convexity is typically benign.
%, and an illustrative example is considering the 1D case, when $u$ and $A_i$ are scalars, and hence $A_i\ge 0$. Each term in the sum is $-f_i\log(A_i|u|^2)$  % THIS EXAMPLE IS COMPLICATED SINCE u CAN BE COMPLEX...

Under this reparameterization, we recast the MLE problem of Eq.~\eqref{cvx-obj} to the following (non-convex) optimization problem
\begin{align} \label{eq:MLE-BM}
    & %\underset{\rho}{\min}
    \min_{\U\in\C^{d\times r}}
    \quad \f(\U) \\ %\label{cvx-obj} \\
    & \text{subject to} \quad \|\U\|_F = 1 \notag % \label{cvx-cons}.
\end{align}
where we use the change of variables $\rho=\U\U^\adjoint$ which enforces $\rho \succeq \mat{0}$ automatically, and $\trace(\rho)=\|\U\|_F^2=1$. 
Note that, if desired, we can relax the non-convex constraint $\|\U\|_F^2=1$ to $\|\U\|_F^2\le 1$ since if $\|\U\|_F<1$ then letting $c=\|\U\|_F$ we have $\f(\U/c) = J(\U) + 2\sum_{i=1}^{\nTotalPOVM} f_i\log(c) < \f(\U)$ since $c<1$, meaning that $\|\U\|_F=1$ will happen automatically at global optimality.  
If $r=d$, then the convex formulation Eq.~\eqref{cvx-obj} and the non-convex formulation Eq.~\eqref{eq:MLE-BM} are equivalent problems, in the sense that a global minimizer of one can be used to easily find the global minimizer of the other.

% If $r=d$ (or under conditions similar to $r\lesssim\sqrt{d}$, cf.~CITE-TODO), any global minimizer $\U$ of the above problem leads to a global minimizer of the original convex MLE in Eq.~\eqref{cvx-obj} by setting $\rho=\U\U^\adjoint$.
The foundational results of Burer-Monteiro theory show that under mild conditions, this non-convex formulation has no spurious local minima and yields global optima when the rank is sufficiently large \cite{burer2005local, boumal2016non,WaldspurgerBM},
i.e., if $r\gtrsim\sqrt{d}$, then any global minimizer $\U$ of the above problem leads to a global minimizer of the original convex MLE in Eq.~\eqref{cvx-obj}. Recent work has extended this guarantee to broader classes of SDPs \cite{AravkinBM,cifuentes2021burer}, and further clarified the rank bounds necessary for optimality \cite{WaldspurgerBM}. More generally, the impact of smooth reparameterizations—such as those used in the Burer–Monteiro method—on optimization landscapes has been theoretically characterized, showing they often improve the geometry of non-convex problems by removing undesirable critical points \cite{levin2025effect}.  
An example of a typical result in the literature is the following theorem (proved for the real-valued case, but easily extended to the complex case):
\begin{theorem}[Burer-Monteiro style result, simplified version of Thm.~4.1 in \cite{AravkinBM}]
\label{thm:BM}
Consider an optimization problem of the form
\begin{align} \label{eq:Arvakin1}
    &\min_{\rho\succeq \mat{0}} \; \likelihood(\rho) \tag{*}\\ 
    &\text{subject to }\; h(\rho)=0, \;\rank(\rho)\le r \notag 
\end{align}
where $\rho\in\R^{d\times d}$ is positive semidefinite and $\likelihood$ and $h$ are continuous. Consider also the reparameterization
\begin{align} \label{eq:Arvakin2}
    &\min_{\U\in\R^{d\times r}} \; \likelihood(\U\U^\adjoint) \tag{**}\\ 
    &\text{subject to }\; h(\U\U^\adjoint)=0.\notag 
\end{align}
Let $\rho^\star =\U^\star\U^{\star\adjoint}$ where $\U^\star$ is any feasible point for Eq.~\eqref{eq:Arvakin2}. Then $\rho^\star$ is a local minimum of Eq.~\eqref{eq:Arvakin1} if and only if $\U^\star$ is a local minimum of Eq.~\eqref{eq:Arvakin2}.
\end{theorem}

In our context, this leads to the following theorem.
\begin{theorem}
Any local minimum of Eq.~\eqref{eq:MLE-BM} corresponds to a local minimum of Eq.~\eqref{eq:MLE-rankr}.
In particular, if there is an optimal solution $\rho^\star$ of Eq.~\eqref{cvx-obj} with rank $r^\star$, then any local minimum of Eq.~\eqref{eq:MLE-BM} for $r\ge r^\star$ corresponds to a global minimizer of Eq.~\eqref{cvx-obj}.
\end{theorem}
\begin{proof}
We invoke Thm.~\ref{thm:BM}, identifying $h(\rho) = \trace(\rho)-1$. Both $\likelihood$ and $h$ are continuous on their domains, and $\trace(\U\U^\adjoint)=\|\U\|_F^2=1$. %Furthermore, we can relax $\|\U\|_F^2=1$ to $\|\U\|_F^2\le 1$ without changing any minimizers as follows: we first note that $\U=\mat{0}$ is not a feasible point, so we can assume $\U\neq \mat{0}$. Then if $\|\U\|_F<1$, take any $c\in [\|\U\|_F, 1]$ so that $\U/c$ is still feasible. Then 
% $\f(\U/c) = J(\U) + 2\sum_{i=1}^{\nTotalPOVM} f_i\log(c) < \f(\U)$ since $c<1$. Taking $c\to 1$ we thus see that $\U$ cannot have been a local minimizer.
% 
% We can relax the non-convex constraint $\|\U\|_F^2=1$ to $\|\U\|_F^2\le 1$ since if $\|\U\|_F<1$ then letting $c=\|\U\|_F$ we have $\f(\U/c) = J(\U) - 2\sum_{i=1}^{\nTotalPOVM} f_i\log(c^{-2}) < \f(\U)$ since $c<1$, meaning that $\|\U\|_F=1$ will happen automatically at optimality.  
% 
\end{proof}

% For QST, we will use $r\ll \sqrt{d}$, so the  theoretical guarantees from the literature do not apply, which motivates our {\em a posteriori} error estimates which we discuss later.

% =====================================================================================
% =====================================================================================
% % \subsubsection{\textbf{Gradient Calculation}}
% % \subsection{Efficient Function and Gradient Calculations} \label{sec:efficientCalculation}
% \section{Efficient Function and Gradient Calculations} \label{sec:efficientCalculation}
% \input{section_efficientGradient}  % Wed, May 14, putting this in a new file so easy to move around

% =====================================================================================
% =====================================================================================
\section{\label{sec:Algo}Algorithm and derivations}
The key to our efficient algorithm is rewriting the MLE optimization problem as an unconstrained optimization problem. We present the theory about this equivalence, then derive a useful error bound, and finally present some algorithms to solve the unconstrained problem.

% https://q.uiver.app/#q=WzAsNCxbMCwwLCJcXHRleHR7Y29udmV4IGNvbnN0cmFpbmVkIE1MRX0iXSxbMCwyLCJcXHRleHR7bm9uLWNvbnZleCBjb25zdHJhaW5lZCBNTEV9Il0sWzIsMCwiXFx0ZXh0e2NvbnZleCB1bmNvbnN0cmFpbmVkIE1MRX0iXSxbMiwyLCJcXHRleHR7bm9uLWNvbnZleCB1bmNvbnN0cmFpbmVkIE1MRX0iXSxbMCwxLCJcXHRleHR7QnVyZXItTW9udGVpcm8gcGFyYW1ldGVyaXphdGlvbn0iLDFdLFsyLDMsIlxcdGV4dHtCdXJlci1Nb250ZWlybyBwYXJhbWV0ZXJpemF0aW9ufSIsMV0sWzAsMiwiXFx0ZXh0e1RobS4gVEJEfSJdLFsxLDMsIiIsMSx7InN0eWxlIjp7ImJvZHkiOnsibmFtZSI6ImRvdHRlZCJ9fX1dXQ==
% https://q.uiver.app/#q=WzAsNCxbMCwwLCJcXHRleHR7Y29udmV4IGNvbnN0cmFpbmVkIE1MRX0iXSxbMCwyLCJcXHRleHR7bm9uLWNvbnZleCBjb25zdHJhaW5lZCBNTEV9Il0sWzIsMCwiXFx0ZXh0e2NvbnZleCB1bmNvbnN0cmFpbmVkIE1MRX0iXSxbMiwyLCJcXHRleHR7bm9uLWNvbnZleCB1bmNvbnN0cmFpbmVkIE1MRX0iXSxbMCwxLCJcXHRleHR7QnVyZXItTW9udGVpcm8gcGFyYW1ldGVyaXphdGlvbn0iLDFdLFsyLDMsIlxcdGV4dHtCdXJlci1Nb250ZWlybyBwYXJhbWV0ZXJpemF0aW9ufSIsMV0sWzAsMiwiXFx0ZXh0e1RobS4gVEJEfSJdLFsxLDMsIiIsMSx7InN0eWxlIjp7ImJvZHkiOnsibmFtZSI6ImRvdHRlZCJ9fX1dXQ==
\begin{figure*}[t]
\[\begin{tikzcd} %[cramped]
	{\text{convex constrained MLE, Eq.~}\eqref{cvx-obj}} && {\text{convex unconstrained MLE, Eq.~}\eqref{cvx-obj-unconstrained}} \\
	\\
	{\text{non-convex constrained MLE, Eq.~}\eqref{eq:MLE-BM}} && {\text{non-convex unconstrained MLE, Eq.~}\eqref{eq:Lag-U}}
	\arrow["{\text{Thm.\ \ref{thm:equivalence}}}", red,leftrightarrow,from=1-1, to=1-3]
	\arrow["{\text{Burer-Monteiro parameterization}}"{description}, from=1-1, to=3-1]
	\arrow["{\text{Burer-Monteiro parameterization}}"{description}, from=1-3, to=3-3]
	\arrow[dotted, from=3-1, to=3-3]
\end{tikzcd}\]
    \caption{Diagram of our approach to solving Eq.~\eqref{cvx-obj} by solving Eq.~\eqref{eq:Lag-U}.}
    \label{fig:commutativeDiagraom}
\end{figure*}
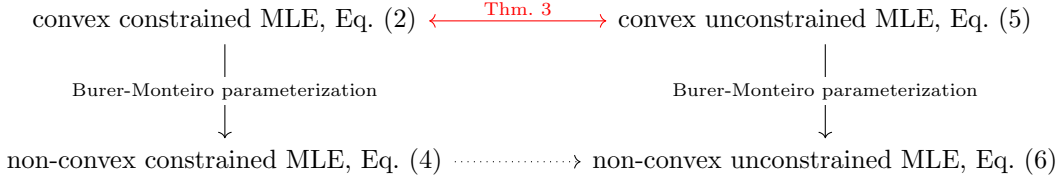

\subsection{Reformulation of MLE as an unconstrained problem}
\begin{comment}
The motivation for solving the QST problem using Burer-Monteiro factorization is the scalability issue for large-scale QST problems. For low-rank estimate of the density matrix, one could choose $U^{d \times r}$, where, $r \ll d$. This is a huge improvement because we have much fewer parameters to optimize over and we can store the iterates in the form of $U$ instead of $\rho$. Due to the merits of the BM factorization, the optimization in \eqref{cvx-obj} - \eqref{cvx-cons} can be posed in a lower dimension by applying Burer-Monteiro: $\rho = UU^\adjoint$, \ $U\in \C^{n \times r}$
\begin{equation}
\min_{U, \; \|U\|_F=1 }\, \sum_{i=1}^m -f_i \log(  \<U \mid A_i \mid U\> )
\end{equation}
By construct, the BM factorization takes care of the positive semi-definite constraint and promotes low-rankness, and the unit trace constraint translates to unit norm constraint. This factorization renders our new objective function to be non-convex which generally requires some special structure on $A_i$ for gradient descent based methods to find the global minima along with a good initialization point near to the optimal set \cite{doi:10.1137/18M1224738}.
\end{comment}

% \begin{align} \label{eq:MLE-BM}
%     & %\underset{\rho}{\min}
%     \min_{\U\in\C^{d\times r}}
%     \quad \f(\U) \\ %\label{cvx-obj} \\
%     & \text{subject to} \quad \|\U\|_F \le 1. \notag % \label{cvx-cons}.
% \end{align}

The non-convex formulation of the MLE, Eq.~\eqref{eq:MLE-BM}, is parameterized so that the positive semi-definiteness constraint is encoded automatically, but it still requires $\|\U\|_F = 1$ to ensure normalization, i.e., $\trace(\rho)=1$. Below we present a fully unconstrained version of the MLE problem that replaces the constraint $\trace(\rho)=1$ with a penalty $\lambda \trace(\rho)$, or in the factored case, replaces 
$\|\U\|_F = 1$ with a penalty $\lambda\|\U\|_F^2$. 
% with a penalty:
% \begin{equation} \label{eq:Lag-U}
% % \Lagrangian(\U, \lambda) = 
% \min_{\U\in\C^{d\times r}} \underbrace{\f(\U) + \lambda \|\U\|_F^2}_{\f_\lambda(\U)}.
% \end{equation}
In optimization theory, this penalized objective is called the Lagrangian, and $\lambda$ is known as the dual variable or Lagrange multiplier.  For our particular problem, it turns out we can set the value of $\lambda$ {\em a priori}.

\begin{theorem}[Unconstrained MLE] \label{thm:equivalence}
Consider
\begin{align} 
    & %\underset{\rho}{\min}
    \min_{\rho\in\H^d}
    \quad \likelihood(\rho) + \lambda\trace(\rho) \label{cvx-obj-unconstrained} \\ %\tag{MLE} \\
    & \text{subject to} \quad \rho \succeq \mat{0}. \notag % \label{cvx-cons}.
\end{align}
with $\lambda = \sum_{i=1}^{\nTotalPOVM} f_i$. Then, like Eq.~\eqref{cvx-obj}, this optimization problem is convex; and $\rho$ is an optimal solution for Eq.~\eqref{cvx-obj} if and only if $\rho$ is an optimal solution for Eq.~\eqref{cvx-obj-unconstrained}, i.e., they are equivalent problems.
\end{theorem}
We delay the proof for a moment.
\newcommand{\Ustar}{\U_{\lambda}}
Given the equivalence of the MLE, we can now apply the Burer-Monteiro parameterization to the unconstrained version Eq.~\eqref{cvx-obj-unconstrained}, giving  %Eq.~\eqref{eq:Lag-U}. 
\begin{equation} \label{eq:Lag-U}
% \Lagrangian(\U, \lambda) = 
\min_{\U\in\C^{d\times r}} \underbrace{\f(\U) + \lambda \|\U\|_F^2}_{\f_\lambda(\U)}.
\end{equation}

The reparameterization does not affect the constraint, which we confirm in the following theorem.
% ================ Theorem ===============
\begin{theorem}\label{thm:nonconvex-equivalence}
Set $\lambda =   \sum_{i=1}^{\nTotalPOVM} f_i$ and let $\Ustar$ be any non-zero stationary point for Eq.~\eqref{eq:Lag-U} with this value of $\lambda$.  Then $\|\Ustar\|_F^2 = 1$, i.e., if $\rho=\Ustar\Ustar^\adjoint$ then $\trace(\rho)=1$.
% Then $\Ustar$ is a KKT point for the constrained problem Eq.~\eqref{eq:MLE-BM}.
\end{theorem}

Now we return to the proofs of the above theorems.
\begin{proof}[Proof of Thm.~\ref{thm:equivalence}]
We start by examining the optimality conditions for Eq.~\eqref{cvx-obj}. We note that $\rho=I$ is a strictly feasible point for the constraints, hence Slater's conditions hold, implying that the Karush–Kuhn–Tucker (KKT) conditions are necessary and sufficient for optimality~\cite{boyd2004convex}.  To derive the KKT conditions, we introduce the dual variables $\sigma\in\H^d$ and $\mu\in\R$ and the Lagrangian 
% The corresponding Lagrangian is:
\begin{equation} \label{eq:lagrangian}
    \Psi(\rho, \sigma, \mu) = -\sum_{i=1}^{\nTotalPOVM} f_i \log( \trace(A_i \rho) ) -
    % \langle\rho , \sigma \rangle 
    \trace(\rho \sigma )
    + \mu(\trace(\rho)-1).
\end{equation}
The KKT conditions for the above Lagrangian are to find $\rho,\sigma\in\H^d,\;\mu\in\R$ that satisfy 
\begin{subequations}\label{eq:KKT1}
\begin{align}
% \text{find }\rho,\sigma\in\H^d,\;\mu\in\R \text{ s.t. } \quad
\nabla_\rho \Psi(\rho, \sigma, \mu) = \mat{0}, \quad\text{i.e.,}\quad
\mat{0} &= 
    % \underbrace{
    -\sum_{i=1}^{\nTotalPOVM} \frac{f_i}{\trace(A_i \rho)} A_i%}_{\nabla \likelihood(\rho)} 
    - \sigma + \mu I  \label{KKT1a} \\ %\tag{stationarity}\\
    \rho &\succeq \mat{0}, \;\trace(\rho)=1\label{KKT1b}  \\% \tag{primal feasiblity}\\
    \sigma &\succeq \mat{0} \label{KKT1c} \\% \tag{dual feasibility}\\
    %\<\rho, \sigma \> 
    \trace(\rho \sigma )
    &= 0 \label{KKT1d} % \tag{complementary slackness}.
\end{align}
\end{subequations}
known as stationarity, primal feasiblity, dual feasiblity and complementary slackness, respectively.
\noeqref{eq:KKT1,KKT1a,KKT1b,KKT1c,KKT1d,KKT2a,KKT2b,KKT2c,KKT2d}

Similarly, the problem Eq.~\eqref{cvx-obj-unconstrained} also satisfies Slater's condition, and its KKT conditions are to find $\rho_\lambda,\sigma_\lambda\in\H^d$ that satisfy 
\begin{subequations} \label{eq:KKT2}
\begin{align}
\mat{0} &= 
    % \underbrace{
    -\sum_{i=1}^{\nTotalPOVM} \frac{f_i}{\trace(A_i \rho_\lambda)} A_i%}_{\nabla \likelihood(\rho)} 
    - \sigma_\lambda + \lambda I  \label{KKT2a} \\
    \rho_\lambda &\succeq \mat{0} \label{KKT2b} \\
    \sigma_\lambda &\succeq \mat{0}\label{KKT2c} \\
    %\<\rho_\lambda, \sigma_\lambda \> 
    \trace(\rho_\lambda \sigma_\lambda )
    &= 0.\label{KKT2d} 
\end{align}
\end{subequations}
Suppose $\rho_\lambda,\sigma_\lambda$ satisfy Eq.~\eqref{eq:KKT2} for $\lambda = \sum_{i=1}^{\nTotalPOVM} f_i$.  We claim that setting $\rho=\rho_\lambda$, $\sigma=\sigma_\lambda$ and $\mu = \lambda$ will satisfy Eq.~\eqref{eq:KKT1}. The only equation that is not obviously satisfied is $\trace(\rho_\lambda)=1$ from Eq.~\eqref{KKT1b}. Take Eq.~\eqref{KKT2a} and multiply on the right by $\rho_\lambda$ and then take the trace:
\begin{align*}
0 &= -\sum_{i=1}^{\nTotalPOVM} \frac{f_i}{\trace(A_i \rho_\lambda)} \trace(A_i \rho_\lambda)
    - \trace(\sigma_\lambda\rho_\lambda ) + \lambda \trace(I \rho_\lambda) \\ 
    &= -\sum_{i=1}^{\nTotalPOVM} f_i - 0 + \lambda \trace(\rho_\lambda) \\
    &= \lambda(-1 + \trace(\rho_\lambda) )
\end{align*}
hence $\trace(\rho_\lambda)=1$, where we used $\lambda= \sum_{i=1}^{\nTotalPOVM} f_i$ and $\lambda\neq 0$.

Conversely, if $\rho,\sigma$ and $\mu$ satisfy Eq.~\eqref{eq:KKT1}, by the same calculation we see that $\mu=\lambda$, and hence $\rho_\lambda=\rho$ and $\sigma_\lambda=\sigma$ are primal and dual optimal for Eq.~\eqref{eq:KKT2}.
\end{proof}

\begin{proof}[Proof of Thm.~\ref{thm:nonconvex-equivalence}]
By definition, a stationary point of Eq.~\eqref{eq:Lag-U} satisfies $\mat{0} = \nabla \f(\Ustar) + 2\lambda \Ustar$, i.e., 
\begin{equation} \label{eq:stationary-b}
\mat{0} = -2\sum_{i=1}^{\nTotalPOVM} \frac{f_i}{ \trace(\Ustar^\adjoint A_i \Ustar ) } A_i \Ustar + 2\lambda \Ustar.
\end{equation}
Multiplying 
the above equation 
% Eq.~\eqref{eq:stationary-b}
on the left by $\U^\adjoint$ and taking the trace, then via linearity of multiplication and trace, 
\begin{align*} %\label{eq:stationary_norm_2}
0 &= -2\sum_{i=1}^{\nTotalPOVM} \frac{f_i}{ \trace(\Ustar^\adjoint A_i \Ustar ) } \trace(\Ustar^\adjoint A_i \Ustar) + 2\lambda \trace(\Ustar^\adjoint \Ustar) \\
&=  -2\sum_{i=1}^{\nTotalPOVM} f_i + 2\lambda \|\Ustar\|_F^2
\end{align*}
hence our choice of 
% \begin{equation}
    $\lambda=\sum_{i=1}^{\nTotalPOVM} f_i $
% \end{equation}
gives $\|\Ustar\|_F^2=1$.
% Looking at the constrained problem Eq.~\eqref{eq:MLE-BM}, the constraint is equivalent to $\|\U\|_F^2 = 1$, so we identify $g(\U) = \|U\|_F^2 - 1$, and hence $\nabla g(\U) = 2 \U$.  This setting $\nu^\star = \frac{1}{2}\lambda$, we see that $\Ustar$ satisfies Eq.~\eqref{eq:stationarity1}, and furthermore it is feasible as shown above, hence it is a KKT point.
\end{proof} 

Note that there is no similar trick that works for the least-squares formulation, as we exploited the properties of $\log$ and its gradient that the least-squares objective does not have.
%
%We have shown that the Burer-Monteiro approach is unconstrained and the same is \emph{not} true for the least-squares MLE formulation.
% The Lagrangian formulation also leads to another advantage over the leading strategy in \cite{PhysRevA.95.062336}, where during the CG run, unit trace is ensured by normalizing once the state is updated in the factored domain. On the contrary, the Lagrange multiplier $\lambda$ ensures that the updated factor $\U$ has unit norm and therefore, $\rho = \U\U^\adjoint$ has unit trace without any explicit operation required. 

% =====================================================================================
% =====================================================================================
% \subsubsection{\textbf{Gradient Calculation}}
% \subsection{Efficient Function and Gradient Calculations} \label{sec:efficientCalculation}
\subsection{Efficient Function and Gradient Calculations} \label{sec:efficientCalculation}
% \input{section_efficientGradient_v2}  % Wed, May 14, putting this in a new file so easy to move around

% Sept 30 2025, putting the section_efficientGradient_v2 here
All the optimization methods we consider require calculating the function and gradient of the objective function, either $\likelihood : \H^d \to \R$ or $\f: \C^{d\times r} \to \R$.  Because $\nTotalPOVM =  \nEachPOVM\nPOVM$ is often large, a dominant cost of these computations is computing either $\rho \mapsto \trace(A_i \rho)$ or $\vec{U} \mapsto \trace(\vec{U}^\adjoint A_i \vec{U})$, for all $i\in 1,\ldots,\nTotalPOVM$.  For each $i$, computing $\trace(A_i \rho)$ by first multiplying $A_i \rho$ and then taking the trace will cost $\mathcal{O}(d^3)$ flops. This can be reduced to $\mathcal{O}(d^2)$ flops by exploiting $\trace(A_i^\top \rho) = \hvec(A_i)^\top \hvec(\rho)$ (see Appendix \ref{sec:vectorization} for details on $\hvec$), but this is still expensive.  If $\nTotalPOVM\approx 4^n=d^2$ (as it is when the measurements are informationally complete), this leads to $\mathcal{O}(d^4)$ flops for every function or gradient evaluation.  The low-rank case with $\mat{U}$ is similar (and we will discuss its subtleties below).  Since $\mathcal{O}(d^4)$ is unacceptably slow, we discuss tricks that exploit structure in order to speed up the calculation.  We divide these strategies into two groups, depending on the style of the problem. The first strategy is suitable for moderate numbers of qubits (roughly $n \lesssim 16$) where $d=2^n$ is not too large and it is also reasonable to collect $\nTotalPOVM\approx d^2$ measurements.  The second strategy focuses on very large dimensions $n \gtrsim 17$ where $\nTotalPOVM\ll d^2$, and will only be applicable to the factored formulation, $\vec{U} \mapsto \trace(\vec{U}^\adjoint A_i \vec{U})$. In this second setting, $d$ may be so large that the computer memory cannot store the explicit $d\times d$ matrix $\vec{U}\vec{U}^\adjoint$, so we avoid explicitly creating the matrix and hence refer to the resulting technique as the ``low memory'' technique, since the memory requirements will be lower than $d^2$.

% , and while we will consider several variants of cost functions, they all build on the objective function $J$ defined in Eq.~\eqref{eq:MLEobjective}. 

% The gradient of $J$ itself is straightforward to calculate using the chain rule:
% \begin{equation}\label{eq:gradJ}
%     \nabla \likelihood(\rho) = -\sum_{i=1}^{\nTotalPOVM} \frac{f_i}{\trace(A_i \rho)} A_i.
% \end{equation}
% ... However, calculating this efficiently is nuanced.

% TBD

\subsubsection{Efficient calculation for moderate dimensions: Quantum Measurement Transform}
\label{sec:QMT}
We first consider the setting when $n$ is not too large so that we can explicitly store the $d\times d$ matrix $\rho$ on the computer (i.e., 16 qubits or smaller, since a 16 qubit density matrix requires 34.3 GB if stored in double precision), and we also assume we have a structured set of POVMs created by tensor products of single qubit POVMs, like the Pauli and tetrahedral schemes discussed in Sec.~\ref{sec:meas}. The two keys to faster computation are (1) exploiting this tensor product structure, and (2) computing the value of $\trace(A_i \rho)$ for all $i=1,\ldots,\nTotalPOVM$ {\em simultaneously} in order to re-use computations.  This approach has been used before in the literature (e.g., \cite{PhysRevA.95.062336, PhysRevResearch.6.033034}), and we roughly follow the exposition in \cite{PhysRevA.95.062336} where it is called the Quantum Measurement Transform (QMT).

The product structure leverages the fact that measurements 
% and quantum states
can often be expressed as tensor products across multiple subsystems (qubits). By decomposing the POVMs into individual single-qubit measurements that operate on individual qubits, the computational complexity is significantly reduced compared to computing probabilities using full density matrices. 
Consider the case of Pauli POVMs, where each matrix $\Ail$ has the structure of $\frac12(I \pm W_\ell)$ for some tensor Pauli matrix $W_\ell$. Because the operation $\trace(\Ail \rho)$ is linear in $\Ail$, clearly the dominant cost reduces to computing $\trace(W_\ell \rho)$.

Instead of constructing $W_\ell = \bigotimes_{j=1}^n w_j^{(\ell)}$, 
where each $w_j^{(\ell)} \in \mathcal{W} = \{I, \sigma_x, \sigma_y, \sigma_z \}$, 
% explicitly (which would be computationally prohibitive),
the algorithm iterates over subsystems, applying $w_j^{(\ell)}$ to each subsystem within the full density matrix. In the process, it also reuses the partial trace computation of one subsystem for subsequent subsystem, leading to computational savings. %More details on the formulation and implementation of this efficient gradient calculation technique, termed the Quantum Measurement Transform (QMT), can be found in \cite{PhysRevA.95.062336, PhysRevResearch.6.033034}. 
The algorithm is described in the pseudocode Algo.~\ref{alg:QMT} \cite{PhysRevA.95.062336} (see Appendix \ref{appendix:qmt}), and then Algo.~\ref{alg:grad} shows how to use it to calculate the gradient $\nabla\likelihood$, all specialized to the case of Pauli POVMs.  
The computation complexity of the QMT calculation is brought down to $\mathcal{O}(nK^{n+1})$ instead of the naive complexity $\mathcal{O}(K^{n} 4^n)$ \cite{PhysRevA.95.062336}, where $K=|\mathcal{W}|$. For an $n$-qubit system with Pauli measurements, $K = 4$, so the QMT is $\mathcal{O}(\log(d) d^2)$ rather than $\mathcal{O}(d^4)$.

% ==========================
\begin{algorithm}[ht]%\onehalfspacing
\DontPrintSemicolon
\caption{Gradient calculation of $\likelihood$ using QMT for Pauli POVMs}\label{alg:grad}
\kwInput{$\rho \in \mathbb{C}^{d \times d}, \vec{f} \in \mathbb{R}^{\nTotalPOVM}, n$} 
$\vec{x} \gets \texttt{QMT} (\rho)$, $\quad x_{\ell} = \trace(W_\ell \rho )\;\forall\ell=1,\ldots,\nPOVM=4^n$ 
%\Comment*[r]{via Algo.~\ref{alg:QMT}} % No \; after comments
 \tcp{Via Algo.~\ref{alg:QMT}}
Reshape: $\vec{f}_{\nEachPOVM\times \nPOVM} \gets \vec{f}_{\nTotalPOVM\times 1}$\;
$\mat{G}=0^{d\times d}$\;
\For {$\ell = 1 \dots \nPOVM $}{
$\mat{G}\gets\mat{G}+ \frac{f_{1}^{\ell}}{\frac12(1+x_\ell)}\frac12(I+W_\ell)$ \tcp{$\APOVM{1}{\ell}=\frac12(I+W_\ell)$}
$\mat{G}\gets\mat{G}+ \frac{f_{1}^{\ell}}{\frac12(1-x_\ell)}\frac12(I-W_\ell)$ \tcp{$\APOVM{2}{\ell}=\frac12(I-W_\ell)$}
}
% \Ail}{\APOVM{k}{\ell}}
% $\nabla \likelihood \gets \textit{\textbf{Adjoint-QMT}} \ (f_i, \trace(A_i x), n)$\;

% \kwReturn{$ \nabla \likelihood = - \sum_{i=1}^{\nTotalPOVM} \frac{f_i}{\trace(A_i \rho)} A_i$}
\kwReturn{gradient $\mat{G}=\nabla \likelihood(\rho)\in\mathbb{C}^{d\times d}$ \tcp{$\nabla \likelihood(\rho) = -\sum_{i=1}^{\nTotalPOVM} \frac{f_i}{\trace(A_i \rho)} A_i$} }

% \textbf{Subroutine:} \textit{\textbf{Forward-QMT}}

% $\mat{P} \gets [\vec{I}, \vec{\sigma_x}, \vec{\sigma_y},  \vec{\sigma_z}] \in\C^{4 \times 4}$

% Reshape: $x_{\underbrace{\scriptstyle 2 \times 2 \times \cdots \times 2}_{2n \text{ times}}} \gets x_{\scriptstyle d \times d}$

% % \begin{bmatrix}
% % \vline height 0.5ex & \vline height 0.5ex & \vline height 0.5ex & \vline height 0.5ex \\
% % I & \sigma_x & \sigma_y & \sigma_z \\
% % \vline height 0.5ex & \vline height 0.5ex & \vline height 0.5ex & \vline height 0.5ex
% % \end{bmatrix}

% \For {$i = 1 \dots n $}{
% Reshape: $x_{\scriptstyle 4 \times 2^{n-2}}$

% $x \gets \mat{P}^\adjoint x$

% $x \gets x^\top$
  
% }

\end{algorithm}

Similar tricks can be done to calculate the function value $\likelihood$ itself, and also to extend to the tetrahedral POVMs.
When doing the Burer-Monteiro factorization, due to the extreme advantage of the QMT, we simply calculate $\f(\U)$ by forming $\rho=\U\U^\adjoint$ (at a cost of $\mathcal{O}(rd^2)$ flops) and then compute $\likelihood(\rho)$, and for $\nabla \f(\U)$ we re-use $\rho=\U\U^\adjoint$ and calculate $\nabla \f(\U) =   2\nabla \likelihood(\U\U^\adjoint)\U$ via the chain rule.

% ================================================================
% ================================================================
\subsubsection{Low memory MLE with BM }\label{sec:lowmem}
% ================================================================

The QMT method \cite{PhysRevA.95.062336} is very efficient for calculation of probabilities $\trace(A_i \rho)$ for moderate size quantum systems where one considers the complete set of operators of a POVM, but it is easily seen that as the number of qubits increases in a system, the QMT method and a na\"ive  storage of measurement operators, in addition to the density matrix, become very memory intensive, having $\mathcal{O}(d^2)$ storage requirement per density matrix or POVM element, hence an overall  $\mathcal{O}(\nTotalPOVM d^2)$ storage requirement. For $n=15$, a single state matrix takes $8.57$ GB to store in double precision; $n=16$ takes $34.3$ GB; and $n=17$ takes $137.2$ GB. Hence for $n\gtrsim 16$ or $17$, one cannot even store the state matrix on a consumer grade laptop, and for $n\gtrsim 20$ one cannot do it on a supercomputer either. 
In this regime, the convex formulation Eq.~\eqref{cvx-obj} is intractable to solve with standard algorithms, % due to the requirement that it forms the density matrix, 
so we focus exclusively on the non-convex form Eq.~\eqref{eq:MLE-BM}.

Furthermore, if $\nTotalPOVM\approx d^2$ (i.e., informationally complete), then for large $n$,  not only is this a physically unrealistic number of measurements to collect, but the runtime and memory requirements becomes unreasonable as well. 
Hence, we approach the very large dimensional case with the assumption that we have a fixed number of measurements $\nTotalPOVM$ which does \emph{not} increase proportionally with $d^2$ but instead is either fixed or grows with $\mathcal{O}(d)$; such settings have been discussed in the context of compressed sensing~\cite{liu2011universal} and direct fidelity estimation~\cite{PhysRevLett.106.230501}. In this setting, we cannot use the QMT since it needlessly calculates all $d^2$ measurements and because it requires direct formation of the density matrix $\rho$.

% As an example, consider $n=15$. In this case, we are looking at a Hilbert space where each of such matrix can take up around $8$GB of memory. In order to carry out measurements at these scales,

To deal with these challenges, we exploit structure and are also careful to never form the density matrix $\rho=\U\U^\adjoint$ explicitly. We describe it below for Pauli POVMs, though it can be extended to other tensor-product measurements.
%we propose a low-memory version of our MLE with Burer Monteiro (BM-MLE) approach for the case of Pauli operators.
As for the QMT approach, the computational bottleneck is computing $\trace(W \U \U^\adjoint)=\trace(\U^\adjoint W \U )$
for some Pauli matrix $W = \bigotimes_{j=1}^n w_j$, 
where each $w_j \in \mathcal{W} = \{I, \sigma_x, \sigma_y, \sigma_z \}$. We write this as $W = \sigma_1 \otimes \sigma_2 \otimes \dots \otimes \sigma_n$ (i.e., $\sigma_i = w_i$) to make it clear that each product matrix is a $2\times 2$ Pauli matrix. In particular, we need to calculate the matrix product $W \U$. To keep things more understandable, we show the derivation below for the case $\U=\vec{u}$, i.e., $r=1$ so that $\U$ only has one column. For $r>1$ columns, one can either loop explicitly over the columns or use tensors.

%  Our low-memory version is built on the following interpretation of the $\trace(A_i \rho)$ where $A_i$ is a Pauli operator of the form $A_i = \bigotimes_{j=1}^n w_j, w_j \in \mathcal{W}$.

% Consider the case of Pauli POVMs, where each matrix $\Ail$ has the structure of $\frac12(I \pm W_\ell)$ for some tensor Pauli matrix $W_\ell$. Because the operation $\trace(\Ail \rho)$ is linear in $\Ail$, clearly the dominant cost reduces to computing $\trace(W_\ell \rho)$.

% Instead of constructing $W_\ell = \bigotimes_{j=1}^n w_j^{(\ell)}$, 
% where each $w_j^{(\ell)} \in \mathcal{W} = \{I, \sigma_x, \sigma_y, \sigma_z \}$, 
% the algorithm iterates over subsystems, applying $w_j^{(\ell)}$ 

Crucial to the derivation is the fact that for a tensor (aka Kronecker) product,  matrix multiplication satisfies the identity $\left( \mat{A} \otimes \mat{B} \right) \vec{x} = \Vec\left( \mat{B} \mat{X} \mat{A}^\top\right)$ where $\mat{X}$ is a  matrix version of the vector $\vec{x}$, reshaped to appropriate dimensions in Fortran ordering.  The pseudocode for this is presented in Algo.~\ref{alg:lowmemau2}.
%Hence, letting $\vec{u}_1 = \vec{u}$, 
% and `$\Vec$' and `$\Mat$' be reshaping operations of appropriate sizes, then 
% \newcommand{\uu}{\vec{u}}

% \textcolor{red}{Trying to fix it...}
% \begin{align*}
% W \uu_1 = (\sigma_1 \otimes \sigma_2 \otimes \dots \otimes \sigma_n)\, \Vec(\uu_1)
% &= (\sigma_1 \otimes (\sigma_2 \otimes \dots \otimes \sigma_n))\, \Vec(\uu_1) \\
% &= \Vec\left((\sigma_2 \otimes \dots \otimes \sigma_n)\, \Mat(\uu_1)\, \sigma_1^\top\right) \\
% &= \Vec\left((\sigma_3 \otimes \dots \otimes \sigma_n)\, \Mat(\uu_2)\, \sigma_2^\top \sigma_1^\top\right) \\
% &\vdots\\
% &= \Vec\left(\sigma_n\, \Mat(\uu_{n-1})\, \sigma_{n-1}^\top \cdots \sigma_1^\top \right)
% \end{align*}
% where $\uu_k := \Mat(\uu_{k-1})\, \sigma_{k-1}^\top$

% \textcolor{red}{Is this right? Using $\sigma^2=I$ ??}
% \begin{align*}
% W \uu_1 = (\sigma_1 \otimes \sigma_2 \otimes \dots \otimes \sigma_n)\, \Vec(\uu_1)
% &= (\sigma_1 \otimes (\sigma_2 \otimes \dots \otimes \sigma_n))\, \Vec(\uu_1) \\
% &= \Vec\left((\sigma_2 \otimes \dots \otimes \sigma_n)\, \Mat(\uu_1)\, \sigma_1^\top\right) \\
% &= \Vec\left((\sigma_3 \otimes \dots \otimes \sigma_n)\, \Mat(\uu_2)\, \sigma_2^\top \sigma_1^\top\right) \\
% &\vdots\\
% &= \Vec\left(\sigma_n\, \Mat(\uu_{n-1})\, \sigma_{n-1}^\top \cdots \sigma_1^\top \right)
% \end{align*}
% where, $\uu_k := \Mat(\uu_{k-1})\, \sigma_{k-1}^\top$

% \textcolor{blue}{Alternate algorithm for low-memory implementation: To be discussed}

\begin{algorithm}[H]
\DontPrintSemicolon
\caption{\textsc{Lowmem} — Apply Pauli product operator on $\U$ to get $W_\ell\U$}
\label{alg:lowmemau2}
\KwIn{\( \U \in \mathbb{C}^{d \times r} \), measurement index \( \ell \in \{0, \dots, 4^n - 1\}\)}
%\KwOut{\( W u \), \( W = \sigma_{m_1} \otimes \cdots \otimes \sigma_{m_n} \)}

% \Begin{

    \( n \gets \log_2(d) \) \tcp*{Number of qubits}

    %Convert $M$ to base-4: \( \mathbf{p} \gets (p_0, p_2, \dots, p_{n-1}) \in \{0,1,2,3\}^n \) 
    $ \vec{q} \gets (q_0, q_1, \dots, q_{n-1}) = \base_{4} (\ell) $ \tcp*{\( W_\ell = \sigma_{q_1} \otimes \cdots \otimes \sigma_{q_n} \)}
    
    \For{ $k = 0$  \KwTo  $n - 1$ }{
         % $p \gets p_{n - k - 1}$  \tcp*{Least-significant qubit first}

        % $\U \gets$ \texttt{PauliFcn[$p$]}$(\U, d, k)$\;
        $\U \gets$ \texttt{PauliFcn[$q_{n - k - 1}$]}$(\U, d, k)$\tcp*{Least-significant qubit first}
    }
% \( \U := W\U \)    
% }
\kwReturn {$ \U $}
\end{algorithm}

Algorithm~\ref{alg:lowmemau2} is designed to calculate the product $W \U$ and represents the action of $n$-fold tensor product based POVMs on a state vector or any $\U \in \mathbb{C}^{d \times r}$. The algorithm achieves this by recursive application of the above identity, which helps us match the tensor product to a sequence of matrix multiplications. This formulation avoids explicit calculation of tensor products and only requires operation with single-qubit operators at any step, for example, $\sigma_x \U$. In the particular case of the Pauli POVMs, application with Pauli matrices are simple operations of scaling and/or swapping rows which can be performed solely by indexing having $\mathcal{O}(dr)$ complexity. 
%Using this fact, we further take a step to completely eliminate any matrix-matrix multiplications encountered for single-qubit operations. 
This implementation  achieves $\mathcal{O}(rd\log d)$ complexity for a single Pauli measurement on $n (=\log_2 d)$ qubits, and hence $\mathcal{O}(\nTotalPOVM rd\log d)$ complexity for $\nTotalPOVM$ measurements
which is better than the QMT if $\nTotalPOVM < d/(r \log d)$; see also Table~\ref{tab:complexity}. The \texttt{PauliFcn} in Algo.~\ref{alg:lowmemau2} denotes a function mapping that applies the local Pauli matrix  $\sigma_q \in \{I, Z, X, Y\}$  to $\U$, acting on a specific qubit (indexed by $k$) within an $n$-qubit system. Rather than forming full Kronecker products, each function uses a structured transformation:
\begin{align*}
q &= 0 \ (\sigma_I): \text{returns  $\U$  unchanged} \\
q &= 1 \ (\sigma_Z): \text{\parbox[t]{0.7\linewidth}{multiplies appropriate rows of $\U$ by a $\pm1$ sign pattern}} \\
q &= 2 \ (\sigma_X): \text{\parbox[t]{0.7\linewidth}{permutes appropriate rows of $\U$ using a bit-flip pattern on the $k$-th qubit}} \\
q &= 3 \ (\sigma_Y): \text{\parbox[t]{0.7\linewidth}{combines $\sigma_Z$ and $\sigma_X$ operations since $Y = iXZ$.}}
\end{align*}
More details related to the implementation are provided in  Appendix~\ref{appendix:apply-paulis}.

%We find that the computational complexity of the proposed low memory version is $\mathcal{O}(d \log d)$ and the storage complexity is $\mathcal{O}(rd)$ where $r$ is the rank of the factor $\U$ used to estimate $\rho = \U \U^\adjoint$.

\begin{figure}[ht]
    \centering
    \begin{subfigure}[h]{0.45\textwidth} 
        \includegraphics[width=\textwidth,trim=0.4cm 0.2cm 1.6cm 1.4cm,clip]{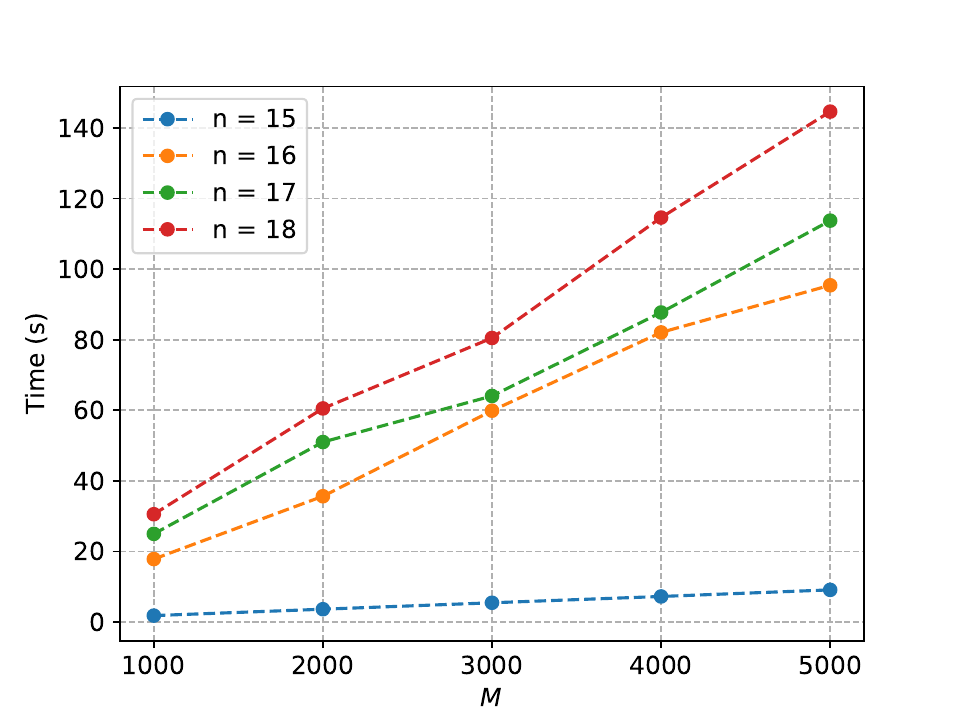}
        \caption{}
        %\caption{Time taken for calculation of probability values as a function of $M$ with varied number of qubits $n$.}
        \label{fig:timeVsM}
    \end{subfigure}
    \hfill
    \begin{subfigure}[h]{0.46\textwidth}
        \includegraphics[width=\textwidth,trim=0.5cm 0.7cm 0.4cm 0.4cm,clip]{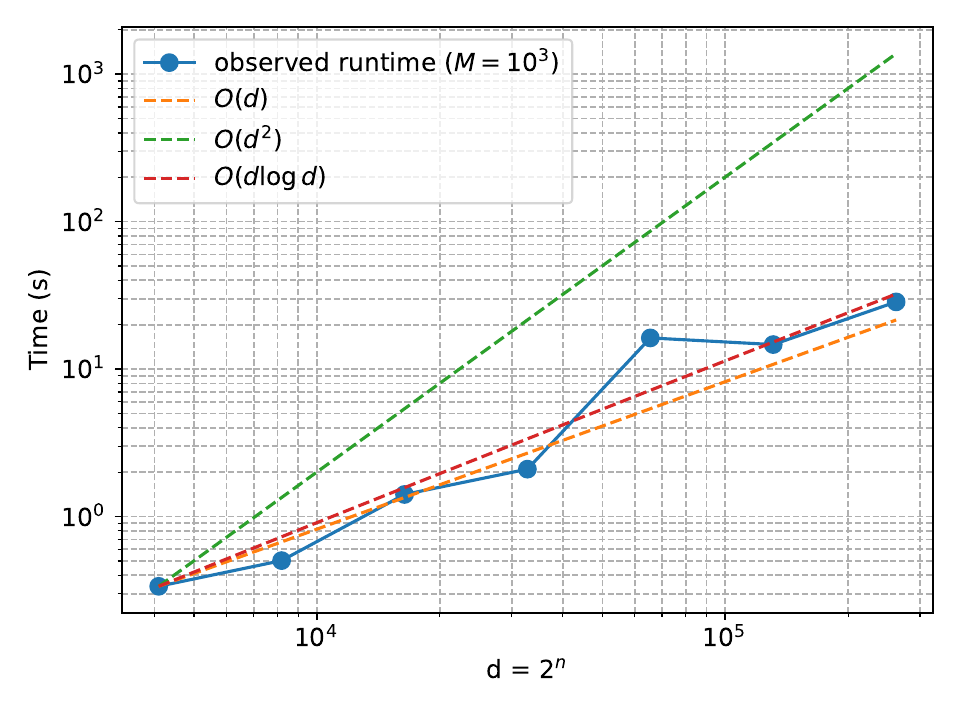}
        \caption{}
        %\caption{Observed runtime complexity of the proposed algorithm as a function of Hilbert space dimension $d$; $\mathcal{O}(d)$, $\mathcal{O}(d\log d)$ and $\mathcal{O}(d^2)$ scalings shown for reference.}
        \label{fig:runtimeComplexity}
    \end{subfigure}
    \caption{Runtime performance of the proposed low-memory BM-MLE for large qubit systems. (a) Time taken for calculation of probability values as a function of $M$ with varied number of qubits $n$. (b) Observed runtime complexity of the proposed algorithm as a function of Hilbert space dimension $d$; $\mathcal{O}(d)$, $\mathcal{O}(d\log d)$ and $\mathcal{O}(d^2)$ scalings shown for reference.}                                  
    \label{fig:runtime}
\end{figure}

We provide some results related to our low-memory implementation of MLE with BM factorization. Fig.~\ref{fig:timeVsM} shows the time taken by the low-memory algorithm to compute probabilities for a $W$-state (rank-1) when the number of Pauli observables, denoted by $M$, is varied and are chosen randomly. The plot confirms the linear relationship between computation time and $M$. %We can also observe that calculating probabilities for an $18$ qubit system for $5000$ operators takes $\sim 2.5$ minutes. 
Fig.~\ref{fig:runtimeComplexity} compares the observed runtime complexity of our algorithm against $d$ for a fixed number of samples $M$. %$\mathcal{O}(d^2)$.
The results in the figure are consistent with our theoretical runtime analysis of   $\mathcal{O}(d\log d)$, in contrast to the $\mathcal{O}(d^2)$ complexity of the QMT. %other non-BM, high-memory methods.
These results show that if the number of measurements is fixed, then our low-memory algorithm provides significant improvements over the QMT when $n$ is large.

\subsection{\label{sec:full-algo}Full Algorithm}
Our proposed approach is to solve Eq.~\eqref{eq:Lag-U} using a standard unconstrained optimization algorithm. Specifically, we focus on L-BFGS and Accelerated Gradient Descent (AccGD) as these are both excellent time-tested solvers. These solvers are designed for optimizing over vectors, not matrices, so when $r>1$, we convert $\U\in\C^{d\times r}$ into $\uu\in\C^{dr\times 1}$ via the reshaping operator $\uu=\rvec(\U)$ and then get back the final answer via $\U=\rmat(\uu)$ where $\rmat=\rvec^{-1}$. We define $\fvec_\lambda(\uu) = \f_\lambda(\rmat(\uu))$ and hence $\nabla \fvec_\lambda(\uu) = \rvec( \nabla \f_\lambda(\rmat(\uu)) )$.

L-BFGS is a limited-memory variant of the quasi-Newton method BFGS, using a computationally cheap approximation of the inverse Hessian in order to approximate Newton's method, and also includes a linesearch in order to stabilize the method~\cite{nocedal2006numerical}. Pseudocode of using L-BFGS to minimize $\fvec_\lambda(\uu)$ is in Algo.~\ref{algo:L-BFGS}. In practice, we either use the Python library \verb|scipy.optimize.minimize| with method \verb|L-BFGS-B| which is equivalent to L-BFGS when there are no constraints, or we use the version of L-BFGS in \verb|torch|.

\begin{algorithm}[ht]
\DontPrintSemicolon
\caption{L-BFGS to solve $\min_{\uu}\; \fvec_\lambda(\uu)$}\label{algo:L-BFGS}
\kwInput{$\uu_1 \in \C^{dr \times 1}$, max iterations $T$, memory size $h \ge 1$}
\For {$t = 1, 2, \dots, T$}{
% \For {$k = 1, 2, \dots, k_{\max}$}{
Update inverse Hessian approximation $\mat{H}_t$ with memory size $h$ \tcp*{see \cite{nocedal2006numerical} for details}
Compute search direction $\vec{p}_t \gets -\mat{H}_t \nabla \fvec_\lambda(\uu_t)$\;
Update $\uu_{t+1}\gets \uu_t + \alpha_t \vec{p}_t$ via linesearch \tcp*{stepsize $\alpha_t$ chosen to satisfy Wolfe conditions}
}
% \kwReturn{$\uu_{k_{\max}}$ and, optionally, $\rho = \U\U^\adjoint$ where $\U=\rmat(\uu_{k_{\max}})$}
\kwReturn{$\uu_{T}$ and, optionally, $\rho = \U\U^\adjoint$ where $\U=\rmat(\uu_{T})$}
% $\U \gets \texttt{gradientBasedSolver}\textit{($u_0$, \textbf{func()}, \textbf{grad()}, $r$, **args)}$
% \tcp*{e.g.: L-BFGS, Acc.\ GD}
% % \quad \Comment{e.g.: L-BFGS, Acc. GD}
% $\rho \gets \U\U^\adjoint$\;
\end{algorithm}

The second method we consider has more rigorous convergence guarantees in the convex case, even though it is often slower in practice. It is similar to gradient descent but has an auxiliary variable, slightly analogous to ``momentum'' techniques such as Polyak's heavy ball method.
The first accelerated method was presented by Nesterov in \cite{Nesterov83} and there are now many variants, cf.~\cite{FISTA,BeckBook2017}. It reduces to gradient descent if $\theta_t = 0$. As the implementation is simpler than L-BFGS, we use our in-house open source code~\footnote{\protect\url{https://github.com/stephenbeckr/convex-optimization-class/blob/main/utilities/firstOrderMethods.py}}.

\begin{algorithm}[ht]
\DontPrintSemicolon
\caption{Accelerated Gradient Descent (AccGD) to solve $\min_{\uu}\; \fvec_\lambda(\uu)$}  \label{algo:AGD}
\kwInput{$\uu_1 \in \C^{dr \times 1}$, max iterations $T$, stepsize $\eta>0$}
$\vec{v}_1 \gets \uu_1$\; 
\For {$t = 1, 2, \dots, T$}{
% \For {$k = 1, 2, \dots, k_{\max}$}{
$\uu_{t+1} \gets \vec{v}_t - \eta \nabla \fvec_\lambda(\vec{v}_t)$\;
$\vec{v}_{t+1} = \uu_{t+1} + \theta_k\left( \uu_{t+1} - \uu_{t}\right)$ \tcp*{$\theta_t = \frac{t}{t+3}$ or other choice from literature}
}
% \kwReturn{$\uu_{k_{\max}}$ and, optionally, $\rho = \U\U^\adjoint$ where $\U=\rmat(\uu_{k_{\max}})$}
\kwReturn{$\uu_{T}$ and, optionally, $\rho = \U\U^\adjoint$ where $\U=\rmat(\uu_{T})$}
% $\U \gets \texttt{gradientBasedSolver}\textit{($u_0$, \textbf{func()}, \textbf{grad()}, $r$, **args)}$
% \tcp*{e.g.: L-BFGS, Acc.\ GD}
% % \quad \Comment{e.g.: L-BFGS, Acc. GD}
% $\rho \gets \U\U^\adjoint$\;
\end{algorithm}

\paragraph{Why solve an unconstrained problem?}
\begin{figure*}[ht]

    \centering
    \includegraphics[width=1.0\linewidth, trim=0.2cm 0.2cm 0.2cm 0.2cm,clip]{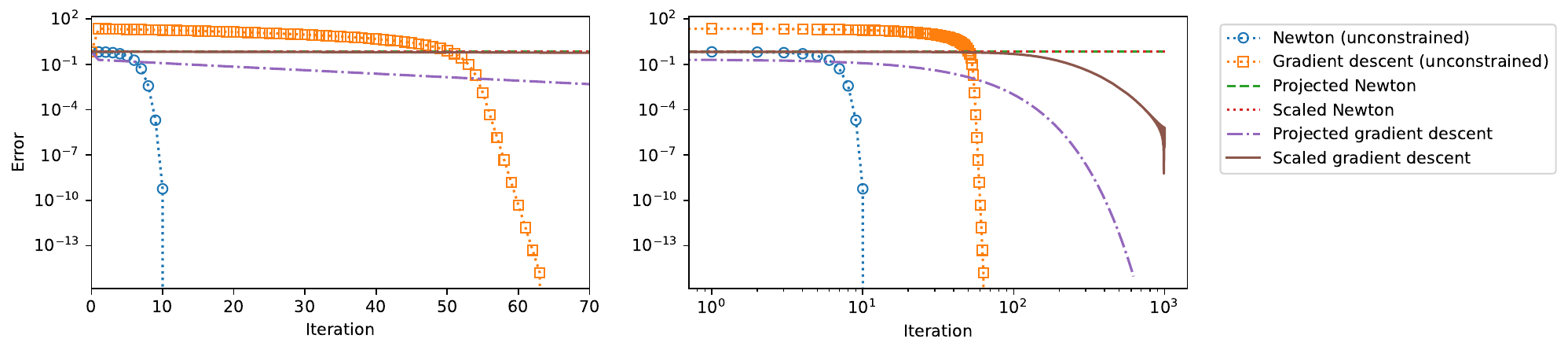}
    \caption{Solving the toy problem
with parameters $a=1.1$, $b=c=10^4$ and $d=1$, starting at $\vec{u}_0 = (.99,.01)$; the optimal solution is at $(0.500005, 0.499995)$. Left: linear $x$-scale. Right: logarithmic $x$-scale.
The two unconstrained solvers solve Eq.~\eqref{problem:toy_unconstrained} and converge quickly, whereas the constrained solvers that solve the (equivalent) problem Eq.~\eqref{problem:toy} converge slowly or not at all. 
% stepsizeList = np.logspace(-5,2,100)
% $ $
    % [0.500005 0.499995
% {'N': None, 'PN': None, 'SN': None, 'GD': np.float64(0.24201282647943834), 'PGD': np.float64(0.01291549665014884), 'SGD': np.float64(37.64935806792464)}
% [0.99 0.01]
% 1.1 10000.0 10000.0 1
    }
    \label{fig:toyExample}
\end{figure*}
Algorithms for unconstrained optimization are much more powerful than algorithms for constrained optimization, which we illustrate explicitly on a small problem. 
Consider a 2-dimensional toy problem over the variable $\vec{u}=(u_1,u_2)$
% $\vec{u}=\begin{bmatrix}
%     u_1 & u_2 
% \end{bmatrix}^\top$
\begin{align} \label{problem:toy}
&\min_{\vec{u}\in\R^2,\;\vec{u} \ge 0}\; %j(\vec{u}) = 
-\log( a u_1 + b u_2 ) - \log( c u_1 + d u_2 ) \\
    &\text{subject to }\; u_1 + u_2 \le 1  \notag
\end{align}
which is of similar form to the MLE problem, and we can similarly find that (under some conditions) the $u_1+u_2\le 1$ constraint can be removed and replaced with a penalty of $\lambda\cdot (u_1+u_2)$ in the objective, for $\lambda=2$:
\begin{align} \label{problem:toy_unconstrained}
\min_{\vec{u}\in\R^2,\;\vec{u} \ge 0}\; %j(\vec{u}) = 
-\log( a u_1 + b u_2 ) & - \log( c u_1 + d u_2 ) \\
& + 2(u_1+u_2)
    % &\text{subject to }\; \vec{u} \ge 0. \notag
\end{align}
% The scripts for the toy problem are here: https://colab.research.google.com/drive/1qGPSPNiZ3grPJMVgNq2mnuFge_FWGsCz?usp=sharing

We choose this example because it is simple enough to derive a closed form expression for the optimal solution~\footnote{
specifically, $u_1^\star = \frac{-1}{\lambda}\left(
\frac{d}{c-d} - \frac{b}{b-a} \right)$ and $u_2^\star=\frac{-1}{\lambda}\left(
\frac{a}{b-a} - \frac{c}{c-d}\right)$.} and because we can adjust the parameters $a,b,c$ and $d$ to make the problem ill-conditioned and hard to solve. This ill-conditioning can come in the form of coupling one variable to another, or from different ``length-scales'' of the variables, and will slow down methods like gradient descent and nonlinear conjugate gradient.  Newton's method, which modifies gradient descent by accounting for the curvature, is barely affected by the ill-conditioning and will converge much faster.  However, Newton's method does not easily extend to problems with non-trivial constraints, especially for large-scale problems; properly extending it involves solving an expensive subproblem at every step~\cite{ProxNewton_LeeSunSaunders}.

The results in Fig.~\ref{fig:toyExample} show an example of this. Newton's method (and even gradient descent) converge quite rapidly when applied to the unconstrained problem Eq.~\eqref{problem:toy_unconstrained}.  We can attempt to adapt these methods to handle the constraints in Eq.~\eqref{problem:toy} by either projecting or scaling/normalizing the iterates to remain feasible, but for Newton's method, this leads to a method that does not converge, and for gradient descent, the method converges slowly, even for an optimally chosen stepsize~\footnote{
All the gradient descent solvers need a stepsize, unlike Newton's method. For each variant of gradient descent, we separately perform a grid search over $100$ stepsizes ranging between $10^{-5}$ and $10^2$ and use the best one.}.
The scaled gradient descent is similar to the scaled non-linear conjugate gradient approach taken by \cite{PhysRevA.95.062336} for solving the MLE.

For the MLE problem, we cannot use any variant of Newton's method because of the large dimensions and hence the high cost of inverting the Hessian matrix. We can, however, use efficient and very effective quasi-Newton methods like L-BFGS, which are nearly as good.  Such methods do not apply straightforwardly to constrained problems; e.g., the constrained solver L-BFGS-B is like L-BFGS but only handles box constraints so it cannot handle a constraint like $\trace(\rho)=1$ and even with box constraints it relies on repeatedly guessing for the ``active'' constraints.

We note other issues with scaling and projection methods. Scaled gradient descent (e.g., renormalizing the variable to satisfy a constraint, such as $\rho \gets \rho/\trace(\rho)$) is a heuristic method that typically does not enjoy any formal convergence guarantees.  Projected gradient descent typically does have convergence guarantees, but the projection can be expensive. For projecting $\rho$ onto the set of unit trace psd Hermitian matrices, one must take an eigenvalue decomposition~\cite{boyd2004convex}, as done in \cite{PhysRevA.95.062336,Guta_2020}, which takes $\mathcal{O}(d^3)$ flops.

% =============================================
\subsection{\label{sec:error-bound}Error Bound}
For real data, there is no ``ground truth'' answer for which to verify that the optimization is indeed working, but instead we can use optimization duality theory to give a bound on the sub-optimality of a candidate solution $\rho$ with respect to the optimal maximum likelihood estimate.

% TODO; mention B-M, split into two methods

% Due to the lack of knowledge of the true quantum state, the performance of different algorithms cannot be compared with respect to the noiseless ground truth. Instead, we can check how far the estimated state is from the true $optimization$ solution. Let's consider the convex MLE formulation. It is possible to upper bound the optimization error using the Burer-Monteiro factorization in MLE objective for quantum state tomography.

% ----------------------
\newcommand{\rhohat}{\tilde{\rho}}
\begin{prop} \label{prop:suboptimality}
Let $\rhohat\in\H^d$ be any feasible point for the MLE problem Eq.~\eqref{cvx-obj}, i.e., $\rhohat \succeq \mat{0}$ and $\trace(\rhohat)=1$.
Then the optimization error can be upper bounded by 
\begin{align}
    &\likelihood(\rhohat)-\likelihood^\star \leq 
    \trace\left( \nabla \likelihood(\rhohat)  \rhohat \right) + \mu, \\
    % \< \nabla \likelihood(\rho) + \mu I, \rho \>
    & \mu = -\min\left\{ 0, \mineig(\nabla \likelihood(\rhohat))\right\}
\end{align}
where %$\mu = \lambda_{\min}(\nabla \likelihood(\rho))$.
% $\mu$ is the minimum eigenvalue of $\nabla \likelihood(\rho)$ 
$\mineig(\cdot)$ is the minimum eigenvalue of a matrix, 
and $\likelihood^\star$ is the optimal value of Eq.~\eqref{cvx-obj}.
\end{prop}
% ----------------------
% label{eq:KKT1}
% ----------------------
\begin{proof}
We will construct dual variables $\sigma$ and $\mu$ based on the KKT equations Eq.~\eqref{eq:KKT1} for the optimization problem Eq.~\eqref{cvx-obj} and then compute the duality gap. The stationarity condition Eq.~\eqref{KKT1a} can be equivalently written as
\begin{equation}
    \vec{0} = \nabla \likelihood(\rhohat) - \sigma + \mu I.
\end{equation}
We set $\sigma =  \nabla \likelihood(\rhohat) + \mu I$ and choose $\mu = -\min\left\{ 0, \mineig(\nabla \likelihood(\rhohat))\right\}$ to ensure dual feasibility, $\sigma \succeq \mat{0}$.  Let $\Psi$ be the Lagrangian from Eq.~\eqref{eq:lagrangian}, and let $\rho^\star$ be an optimal point, i.e., $\rho^\star$ is feasible and $\likelihood(\rho^\star)=\likelihood^\star$.

Thus by stationarity and by convexity of $\Psi$, we have that $\rhohat$ minimizes $\Psi(\rho,\sigma,\mu)$, hence 
\begin{align*}
% \likelihood(\rhohat) - \trace(\rhohat \sigma )+ \mu(\underbrace{\trace(\rhohat)-1}_{=0})
% = 
   \Psi(\rhohat, \sigma, \mu) &= \min_{\rho\in\H^d} \Psi(\rho,\sigma,\mu) \\
   &\le \Psi(\rho^\star,\sigma,\mu) 
\end{align*}
and $\Psi(\rhohat, \sigma, \mu) \le \Psi(\rho^\star, \sigma, \mu)$ implies
\begin{align*}
\likelihood(\rhohat) - \trace(\rhohat \sigma )+ \mu(\underbrace{\trace(\rhohat)-1}_{=0})
&\le 
\likelihood(\rho^\star) - \underbrace{\trace(\rho^\star \sigma )}_{\ge 0} \\
& + \mu(\underbrace{\trace(\rho^\star)-1}_{=0}) \\
&\le \likelihood(\rho^\star)
\end{align*}
using $\trace(\rho^\star)=1$ and $\rho^\star,\sigma \succeq \mat{0}$.
Rearranging, this gives $\likelihood(\rhohat) \le \likelihood^\star + \trace(\rhohat \sigma)$, and substituting in $\sigma =  \nabla \likelihood(\rhohat) + \mu I$ gives the result.
\end{proof}

% \begin{algorithm}[ht]
% \caption{Error Bound}\label{alg:errbound}
% \kwInput{$\U \in \C^{d \times r}, y$}
% \kwReturn{\textit{Error-Bound}}

% $\rho \gets \U\U^\adjoint$ 

% $\mu \gets -\lambda_{\text{min}}(\nabla f(\rho))$

% $\sigma = \nabla f(\rho) + \mu I$

% $\textit{Error-Bound} \gets \< \rho, \sigma \>$

% \end{algorithm}

The following corollary is immediate and lets us have a rigorous {\em a posteriori} error estimate even when solving a non-convex problem.
\begin{corollary} \label{corollary:errorbound}
Given any $\U \in \C^{d\times r}$ with unit Hilbert-Schmidt norm (e.g., obtained by solving the MLE objective using the BM technique), $\rhohat = \U\U^\adjoint$ forms a feasible density matrix and hence the bound of Prop.~\ref{prop:suboptimality} applies.
\end{corollary}
We note that at convergence our method guarantees $\U$ to have unit norm, but for intermediate iterations this need not be the case, hence we use the variable $\U/\|\U\|_F$ when applying the corollary.
Figure \ref{fig:comparisons} in the numerical results section illustrates the tightness and usefulness of this bound.  Since the computation of the explicit matrix $\rhohat$ and the explicit matrix $\nabla \likelihood( \rhohat )$ can be expensive, as well as the eigenvalue computation, in practice we suggest only computing the error bound occasionally, such as every 20 iterations.

% % =====================================================================================
% % =====================================================================================
% % \subsubsection{\textbf{Gradient Calculation}}
% % \subsection{Efficient Function and Gradient Calculations} \label{sec:efficientCalculation}
% \subsection{Efficient Function and Gradient Calculations} \label{sec:efficientCalculation}
% \input{section_efficientGradient}  % Wed, May 14, putting this in a new file so easy to move around

\subsection{Complexity analysis}\label{sec:complexity}
Table \ref{tab:complexity} outlines the total memory requirements and flops for our proposed method using either Algo.~\ref{algo:L-BFGS} or Algo.~\ref{algo:AGD}, as well as that of the state-of-the-art method CG-APG for comparison.
%L-BFGS and CG-APG algorithms.
CG-APG operates in the full Hilbert space of dimension $d$, resulting in a memory requirement of $\mathcal{O}(d^2)=\mathcal{O}(4^n)$.
It is a two-part method, first running non-linear CG %(similar to L-BFGS)
on a factored version of the MLE but using $r=d$, then running a APG (a projected variant of AGD) but on the original convex MLE Eq.~\eqref{cvx-obj}. Since this second phase uses projection of the density matrix, it needs an eigenvalue decomposition at a cost of $\mathcal{O}(d^3)=\mathcal{O}(8^n)$ flops.
% Additionally, the eigenvalue decomposition in the APG phase incurs a computational cost of $\mathcal{O}(d^3)$ or $\mathcal{O}(8^n)$ flops.
While feasible for small quantum systems, these requirements become prohibitive as the system size increases. 
%As opposed to these, L-BFGS solver for BM-MLE requires $\mathcal{O}(hr2^n)$ memory and $\mathcal{O}(rmd\log d)$ FLOPS when using our low-memory implementation. 
In contrast, our method %requires only $\mathcal{O}(rmd\log d)$ per iteration, and
uses much less memory and computation.

\begin{table*}[ht]
\caption{
% use Acc-GD or AGD?
Memory and flop cost per iteration, showing leading order terms only. 
$h$ is the history size in L-BFGS; for AGD, $h=1$.
We list complexity in terms of both qubits $n$ and dimension $d=2^n$, and denote  the total number of POVM elements as $\nTotalPOVM=\nEachPOVM\nPOVM$.
%The number of qubits is denoted as $n$ and $d=2^n$, and the total number of POVM elements is $\nTotalPOVM=\nEachPOVM\nPOVM$.
The two groups of columns ($\nTotalPOVM=d^2$ vs $\nTotalPOVM\ll d^2$) refer to calculations as described in Section~\ref{sec:QMT} or Section~\ref{sec:lowmem}, resp.  
\label{tab:complexity}}
% \begin{ruledtabular}
\begin{tabular}{p{6.5cm}*{4}{p{1.8cm}}}
% \begin{tabular}{p{5.2cm}*{3}{p{2.8cm}}p{4.1cm}}
\toprule 
 &\multicolumn{2}{c}{QMT}&\multicolumn{2}{c}{low-memory}\\
 &\multicolumn{2}{c}{$\nTotalPOVM = \mathcal{O}(4^n)$ i.e., $\mathcal{O}(d^2)$}&\multicolumn{2}{c}{$\nTotalPOVM = \text{const} (\ll d^2)$}\\
 \cmidrule(r){2-3} \cmidrule(l){4-5}
 Method & memory & flops & memory & flops \\  %\hline
\midrule 
 %\hline 
 % CG-APG~\cite{PhysRevA.95.062336}&$\mathcal{O}(4^n)$&$\mathcal{O}(8^n)$ &$\mathcal{O}(4^n)$&$\mathcal{O}(8^n)$ \\
 % Proposed BM-MLE (Algo.~\ref{algo:L-BFGS} or Algo.~\ref{algo:AGD})&$\mathcal{O}(hr2^n)$&$\mathcal{O}(4^{n})$&$\mathcal{O}(hr2^n)$&$\mathcal{O}(r\tilde{m}n2^n)$\\
% 
 %  CG-APG~\cite{PhysRevA.95.062336}&$\mathcal{O}(4^n)$ / $\mathcal{O}(d^2)$ & $\mathcal{O}(8^n)$ / $\mathcal{O}(d^3)$  & $\mathcal{O}(4^n)$ / $\mathcal{O}(d^2)$ & $\mathcal{O}(8^n)$ / $\mathcal{O}(d^3)$ \\
 % Proposed BM-MLE (Algo.~\ref{algo:L-BFGS} or \ref{algo:AGD})&$\mathcal{O}(hr2^n)$ / $\mathcal{O}(hrd)$ &$\mathcal{O}(4^{n})$ / $\mathcal{O}(d^2)$ & $\mathcal{O}(hr2^n)$ / $\mathcal{O}(hrd)$ & $\mathcal{O}(r\tilde{m}n2^n)$ / $\mathcal{O}(r\tilde{m}d\log d)$\\
% 
\multirow{2}{6.5cm}{CG-APG~\cite{PhysRevA.95.062336}} 
   &$\mathcal{O}(4^n)$&$\mathcal{O}(8^n)$ &$\mathcal{O}(4^n)$&$\mathcal{O}(8^n)$ \\
   &$\mathcal{O}(d^2)$&$\mathcal{O}(d^3)$ &$\mathcal{O}(d^2)$&$\mathcal{O}(d^3)$ \\
   \cmidrule{2-5}
 \multirow{2}{6.5cm}{Proposed BM-MLE (Algo.~\ref{algo:L-BFGS} or \ref{algo:AGD})}&$\mathcal{O}(hr2^n)$&$\mathcal{O}(n 4^{n})$&$\mathcal{O}(hr2^n)$&$\mathcal{O}(r\nTotalPOVM n2^n)$\\
 &$\mathcal{O}(hrd)$ & $\mathcal{O}(d^2 \log d)$ & $\mathcal{O}(hrd)$ & $\mathcal{O}(r\nTotalPOVM d\log d)$\\
 \bottomrule
\end{tabular}
% \end{ruledtabular}
\end{table*}

% ========================================
% ========================================
\section{\label{sec:results}Numerical Results}

\subsection{Comparison to other methods}
For the purpose of analyzing the performance of the proposed algorithm, we compare with the following state-of-the-art alternatives:
\begin{description}
    \item[CG-APG] Conjugate-Gradient--Accelerated Projected-Gradient (CG-APG) \cite{PhysRevA.95.062336}.  The CG-APG algorithm starts with implementing a nonlinear conjugate gradient (CG) algorithm in the factored domain (with $r=d$) and switches to accelerated projected gradient (APG) when the Hessian of the objective function starts to settle down and operates directly on $\H^d$, i.e., non-factored domain.
    APG is a variant of accelerated gradient descent, Algo.~\ref{algo:AGD}, that incorporates projection steps; in this case, the projection is onto the set of positive semi-definite trace 1 Hermitian matrices.
    We use their QMT algorithm for all algorithms whenever possible. %; the only algorithm that cannot fully exploit the QMT is $B$-sample LB-SDA, and we 
    % We adopt their QMT algorithm to reduce the computational cost and speed across all the aforementioned algorithms except B-sample LB-SDA. We
    % will later discuss implications of that on its performance.
    \item[LB-SDA] $B$-sample Stochastic Dual Averaging with the Logarithmic Barrier (LB-SDA) \cite{pmlr-v238-tsai24a}. $B$-sample LB-SDA is a mini-batch stochastic first-order method, where $B$ denotes the mini-batch size.
    The mini-batch consists of uniform random selection from among all possible $\nTotalPOVM$ measurements $f_i$.
    %consisting of uniformly random
    %sample correspond to the number of measurements randomly chosen to estimate the gradient.
    It uses the proposition that the gradient of the MLE objective is bounded in the dual local norm associated with the logarithmic barrier, helping provide guarantees for convergence to an  $\epsilon$-optimal solution. From the results of \cite{pmlr-v238-tsai24a}, 1-sample LB-SDA and $d$-sample LB-SDA ($B=d$) and MD have good performance in terms of fidelity and optimization error, so we include them here for comparison purposes.
    $B$-sample LB-SDA is the only algorithm that cannot fully exploit the speed of the QMT, and we will discuss the consequences of this later.
    %We refer to 1-sample LB-SDA as simply ``LB-SDA''.
    \item[MEG] Low-rank Matrix Exponentiated Gradient (MEG) \cite{10.1287/moor.2022.1332}. Low-rank MEG performs a low-rank SVD-based update at each iteration. Under a strict complementarity condition and with a suitably chosen initialization, the method locally converges to the same solution as its full-rank counterpart by effectively exploiting the low-rank structure.
    \item[MD]  Mirror Descent (with Armijo line search) \cite{li2018general}. The MD algorithm with Armijo line search establishes convergence guarantees for the QST objective by relying on a local relative smoothness condition, which is necessary because the objective is neither globally smooth nor has a bounded gradient. %Its performance is compared against algorithms like 
    The simulations of \cite{li2018general} show that MD performs better than other applicable methods such as 
    diluted $\text{R}\rho \text{R}$ \cite{PhysRevA.75.042108}, SCOPT \cite{JMLR:v16:trandihn15a}, and the modified Frank-Wolfe \cite{odor2016frank}. %where it shows better performance than these algorithms. %Therefore, in our work, we include MD with Armijo line search as one of the algorithms for comparison purposes.
    \item[MGD] Momentum-Inspired Factored Gradient Descent \cite{photonics10020116}. As described in Section \ref{sec:related}, MGD is an accelerated version of the algorithm in \cite{kyrillidis2017provable} that uses compressed sensing along with the Burer-Monteiro factorization to promote a low-rank density matrix under Gaussian measurement error.
    \item[iMLE] Iterative MLE \cite{Lvovsky_2004}. iMLE is an iterative maximum-likelihood algorithm to reconstruct quantum optical states from homodyne tomography data and applies an Expectation-Maximization (EM) method to update the density matrix iteratively.
    \item[UGD] Unified (factored and projected) Gradient Descent. 
    In this approach, the density matrix is parameterized through a Cholesky-based factorization, and optimization is carried out using momentum-accelerated gradient descent \cite{PhysRevResearch.6.033034}. It does not provide explicit control over the rank of the reconstructed state. The paper \cite{PhysRevResearch.6.033034} mentions several methods and we use their $\text{Fac}_\text{H}$ parameterization and MRprop optimizer, as these were shown to be the most effective.
\end{description}

We refer to our method, i.e., solving Eq.~\eqref{eq:Lag-U}, as the BM-MLE approach, with two particular instantiations,  ``L-BFGS'' (Algo.~\ref{algo:L-BFGS}) and ``Acc-GD'' (Algo.~\ref{algo:AGD}).
The Python-based implementation of our work is available at \cite{RigorousMLE}. 
Third party algorithms that were not already written in Python were reimplemented in Python, running unit tests on subroutines to verify correctness as well as confirming convergence on small qubit problems.
% We converted most of these algorithms in Python, preserving the original logical flow of their implementations.
All parameters were set to their default values as specified in the original sources, and we used their default initializations; for our methods, we initialized our guess for $\U$ as a draw of a complex valued Gaussian vector that was then normalized. 
Some of the third party algorithms were tailored to specific POVMs, so we test one subset of algorithms on our Pauli POVM example, and another subset on the tetrahedral POVM example.

\subsection{Simulations with Pauli POVMs on the CPU}  \label{sec:CPU}
In the first set of simulations, we use POVMs constructed from Pauli operators to generate measurement. %These measurements yield probabilities corresponding to the outcomes one would observe when applying the Pauli operators to the unknown quantum state. \sout{For the first set of experiments, we use the POVMs derived from Pauli operators to obtain measurements representing the associated probability when these operators are measured on the unknown quantum state.} %An $n$ qubit system constitutes $4^n$ Pauli matrices and therefore, $2 \times 4^n$ POVMs corresponding to the positive and negative outcomes which are the expected value of these operators for the unknown state $\rho$. 
% The setup for the experiments is as follows.
We compare different algorithms by varying the number of qubits $n$ ranging from $2$ to $10$. 
The data are generated by measuring a quantum state consisting of the W-state with $10\%$ depolarizing noise.
%We create the original quantum state to be a W-state and introduce depolarizing noise equal to $0.1$.
We take the number of shots to be $N=100$ for each POVM, and use all $\nPOVM=d^2-1$ non-trivial POVMs. 
We compare our proposed BM-MLE framework, in particular both the ``L-BFGS'' (Algo.~\ref{algo:L-BFGS}) and ``Acc-GD'' (Algo.~\ref{algo:AGD}) instantiations, against relevant third party solvers.
We use the heuristic choice of $r=d/4$ to set the rank of our BM-MLE approach (and use the same $r$ for the low-rank MEG algorithm). This is likely a conservative choice and tuning $r$ would lead to better performance for our method. For L-BFGS, we use the the \texttt{scipy.optimize.minimize} implementation with method \texttt{L-BFGS-B} and memory size \texttt{maxcor} set to $2$.

% The proposed BM-MLE algorithm takes a rank parameter $r$ as an input that imposes the low-rank structure in the solution. We take this parameter to be $r=d/4$ based on a heuristic choice. In a similar way, low-rank MEG takes a rank parameter for performing the rank-$r$ SVD during estimation, with $r$ set to match the value used in BM-MLE. The value of this parameter along with the complementarity property determines the quality of the estimate. 

% ======= Pauli ==========

\begin{figure}[]
    \centering
    \begin{subfigure}[b]{0.45\textwidth}
        \centering
        \includegraphics[width=\textwidth,trim=0cm 0cm 1.6cm 1.4cm,clip]{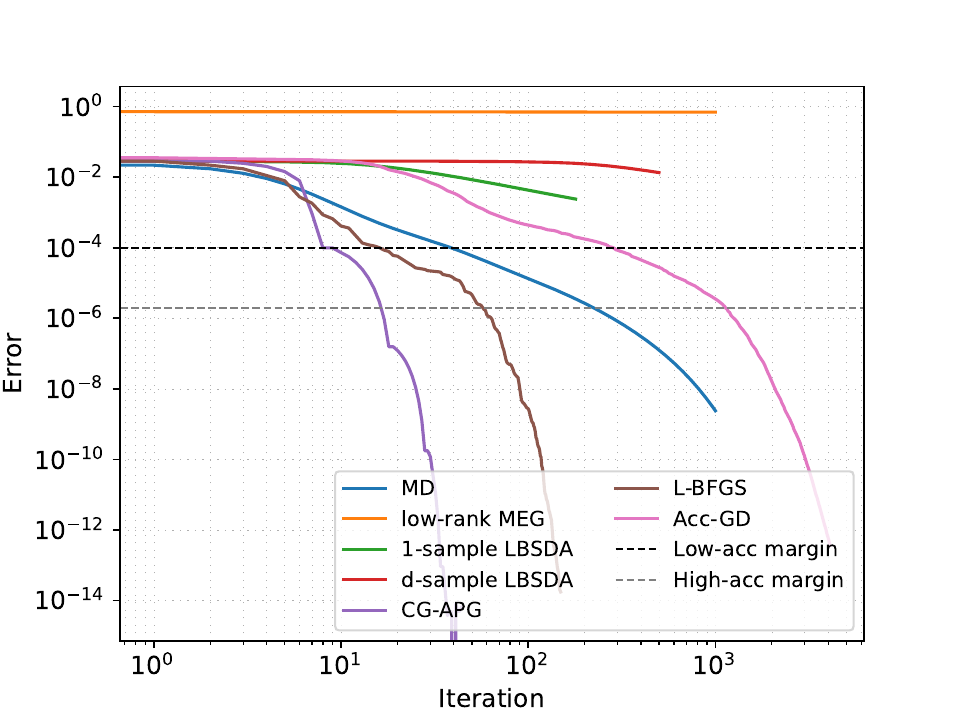}
        \caption{}
        % \caption{$n=4$}
        \label{fig:comparisons-Pauli-error_a}
    \end{subfigure}
    \hfill
    \begin{subfigure}[b]{0.45\textwidth}
        \centering
        \includegraphics[width=\textwidth,trim=0cm 0cm 1.6cm 1.4cm,clip]{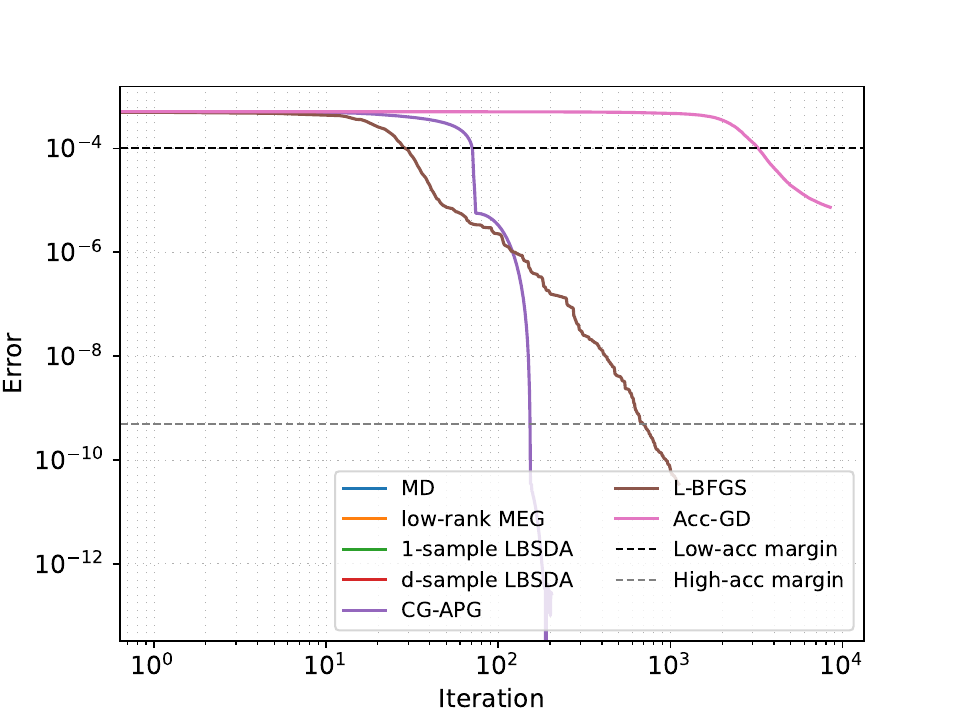}
        \caption{}
        %\caption{$n=10$}
        \label{fig:comparisons-Pauli-error_b}
    \end{subfigure}
    
    \caption{Error performance of different methods for complete Pauli POVM measurements under $10\%$ depolarizing noise on the $W$ state for \textemdash{} (a) $n=4$ and (b) $n=10$. Note that the number of iterations is not equivalent to runtime; the per-iteration runtimes are listed in Table~\ref{tab:timing}, where our algorithm demonstrates better performance than the others.}
    \label{fig:comparisons-Pauli-error}
\end{figure}

Figure~\ref{fig:comparisons-Pauli-error} shows error in the objective $\likelihood$ as a function of iteration. We do not know the true minimum of the likelihood, so we approximate it by using the smallest value observed in the simulation across all methods; we denote this value as $\likelihood_{\text{best}}$. We look at a small qubit example ($n=4$, left) and a larger example ($n=10$, right). For $n=4$ qubits, most solvers start to converge other than low-rank MEG, though $1$-sample LBSDA and $d$-sample LBSDA are slow.  For $n=10$ qubits, we only show the algorithms that make non-trivial improvement; these are CG-APG and our two methods.  

The CG-APG method is consistently faster than our method in the sense that it requires fewer iterations. However, it requires much more time per step.  Table~\ref{tab:timing} shows the average time per step of the algorithms, and we see that for $10$ qubits, CG-APG takes over $6\times$ longer per step than L-BFGS.

\begin{table*}[ht]
    \centering
    \newcommand{\NA}{\textcolor{gray}{\texttt{N/A}}}
    %\begin{tabular}{lcccc}
    %\begin{tabular}{p{5.2cm}*{4}{p{.8cm}}}
    \begin{tabular}{p{5.2cm}@{\hskip 0.6cm}p{.8cm}@{\hskip 0.6cm}p{.8cm}@{\hskip 0.6cm}p{.8cm}@{\hskip 0.6cm}p{.8cm}}
    \toprule
    & \multicolumn{4}{c}{number of qubits, $n$}  \\ 
     \cmidrule(){2-5}
     & $7$ & $8$ & $9$ & $10$ \\
     \midrule
    MD (Mirror Descent) &  $0.37$ & $3.64$ & \mbox{$>1$hr} & \mbox{$>1$hr} \\
    low-rank MEG & $0.41$ & $3.21$  & \NA & \NA \\
    $1$-sample LB-SDA & \NA & \NA & \NA & \NA \\
    $d$-sample LB-SDA & \NA & \NA & \NA & \NA \\
    CG-APG & $0.11$ & $3.04$ & $8.73$ & $26.43$\\
    L-BFGS: our proposed Algo.~\ref{algo:L-BFGS} & $0.07$ & $0.30$ & $1.12$ & $\hphantom{2}4.29$\\
    Acc-GD: our proposed Algo.~\ref{algo:AGD} & $0.03$ & $0.08$ & $0.24$ & $\hphantom{2}0.85$\\
    \bottomrule
    \end{tabular}
    \caption{Average time per iteration (in seconds) for the Pauli POVM CPU simulation. \NA\ denotes that the algorithm was not run for this setting because of slow convergence.
    }
    \label{tab:timing}
\end{table*}

Since $n \lesssim 10$ is relatively small, we use the QMT whenever applicable, as it is fast and does not require explicit formation of the POVMs.
% The QMT algorithm \cite{PhysRevA.95.062336} allows us to not form the POVMs explicitly and store them. 
% Though, s
Since the B-sample LB-SDA is a stochastic mini-batch gradient algorithm that requires selection of random operators in each iteration to approximate the gradient, we form the POVMs explicitly in our code for $n \leq 6$ since computing the full QMT is wasteful. However, that means it does not exploit the cost savings of the QMT and it runs slowly---basically, there is little advantage to subsampling when using the QMT. For this reason, the algorithm runs slowly and 
% Due to this reason B-sample LB-SDA is the slowest algorithm and
fails to attain low optimization error in time comparable to other algorithms. Similar observations can be made from the original paper \cite{pmlr-v238-tsai24a}, where these algorithms took more than $10^4$ seconds to attain error value $10^{-4}$ for $n=6$ which also confirms the slow convergence. The low-rank MEG fails to converge, which could be for a few reasons. Firstly, the algorithm requires the objective function to be Lipschitz continuous and smooth and the QST objective does not meet this requirement \cite{li2018general}. Secondly, the strict complementarity condition may not be met. Lastly, the value of rank parameter for SVD also determines the time to convergence where overestimating the true rank slows down the convergence to the desired error margin.

To make systematic comparisons as a function of the number of qubits, we look at the time it takes to reach a given error threshold, excluding algorithms that take longer than $10^3$ seconds. We use two thresholds, a low-accuracy threshold which is achieved when a solver has a negative log-likelihood within $10^{-4}$ of $\likelihood_{\text{best}}$ (i.e., absolute error), and a high-accuracy threshold which is achieved when the the solver has a likelihood ratio (in comparison to the best solution) above $95\%$ (i.e., relative error). These thresholds correspond to the upper and lower horizontal lines, respectively, shown in both plots of Figure~\ref{fig:comparisons-Pauli-error}.

These results are shown in Fig.~\ref{fig:comparison-Pauli}. For lower accuracy results, shown in the left plot, our proposed method using L-BFGS is the clear winner, running an order-of-magnitude faster than the nearest solver at $n=10$ qubits. Our other method, Acc-GD, as well as CG-APG, do well.  The next best method is MD, which shows a very steep slope as a function of $n$, and is not able to run for $n \ge 9$ within the time budget.  All other solvers do not converge to the desired accuracy within the time budget, so they are not shown.

The higher accuracy setting shows that our method L-BFGS and the method CG-APG are approximately equally good. CG-APG converges more rapidly, which is useful for high-accuracy, but this is offset by the larger cost of the iterations.  We note that we chose $r \propto d$ for these simulations, which prevents our method from having a large advantage: because $r<d$, we cannot achieve extremely high accuracy (since the true solution is unlikely to be low-rank, due to the depolarizing noise); while on the other hand, because $r$ does grow with $d$, we do not see extreme speed benefits.  
We will see more of the benefit of our method in Section~\ref{sec:lowmemory} where we choose $d$ very large and keep $r$ small.

\begin{figure}[ht]
    \centering
    
    \begin{subfigure}[b]{0.45\textwidth}
        \centering
        \includegraphics[width=\textwidth,trim=0cm 0cm 1.6cm 1.4cm,clip]{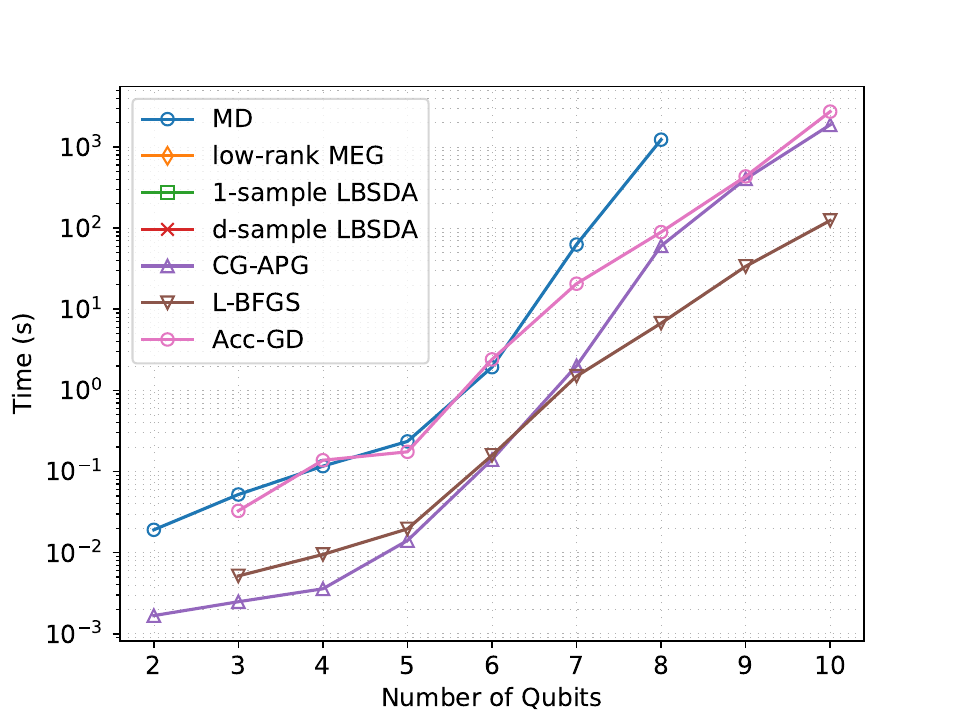}
        \caption{}
        % \caption{Low-accuracy regime}
        \label{low-acc}
    \end{subfigure}
    \hfill
    \begin{subfigure}[b]{0.45\textwidth}
        \centering
        \includegraphics[width=\textwidth,trim=0cm 0cm 1.6cm 1.4cm,clip]{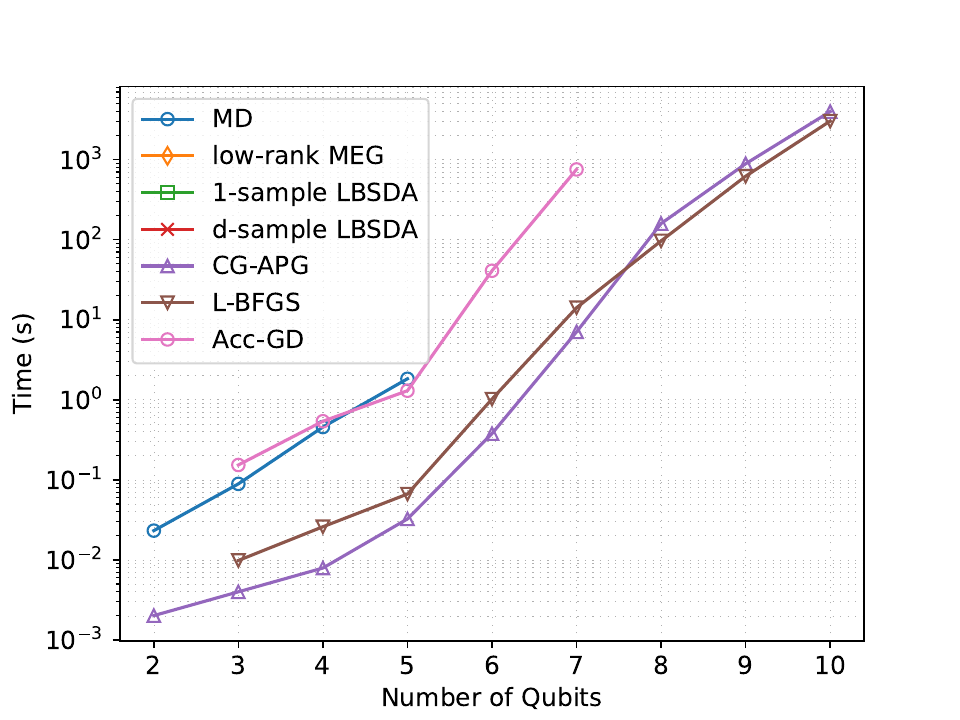}
        \caption{}
        % \caption{High-accuracy regime}
        \label{high-acc}
    \end{subfigure}
    
    \caption{Comparison of different methods in terms of time taken to reach the two accuracy regimes \textemdash{} (a) Low-accuracy regime and (b) High-accuracy regime, under complete Pauli POVM measurements and $10\%$ depolarizing noise on the $W$ state.
    ``L-BFGS'' and ``Acc-GD'' are our proposed methods, Algo.~\ref{algo:L-BFGS} and Algo.~\ref{algo:AGD}, resp. % the actual instantiations of BM-MLE
    Algorithms ``LBSDA'' and ``MEG'' never converge to the required accuracy, so they are not shown.
    }
    \label{fig:comparison-Pauli}
\end{figure}

%On observing Fig. \ref{low-acc} carefully, we find that for the particular problem QST, only 2 past gradient histories are sufficient to get faster convergence compared to CG-APG. A history size of $5$ also gives very similar performance to history size of $2$, but has the disadvantage of requiring more memory to store the additional gradient and higher computational cost per iteration. Moreover, we find that our proposed algorithm also converges slightly faster in the high accuracy regime compared to CG-APG when history size equals $2$ and almost comparably when history size equals $5$. From Fig. \ref{low-acc}, \ref{high-acc}, we can reason why history size of $1$ and $10$ do not perform that well. The history size of $1$ is faster but not accurate when it comes to Hessian approximation and the size of $10$ might be more accurate for Hessian but slower due to more computations required per iteration. Overall, these results show the benefit of using the Burer-Monteiro approach for higher number of qubits and as observed from the simulations, can provide a good alternate formulation of original QST problem when a slightly less accurate solution is acceptable. 

% ======= Duality gap ==========
\begin{figure}[ht]
\centering
\includegraphics[width=0.45\textwidth,trim=0cm 0cm 1.6cm 1.4cm,clip]{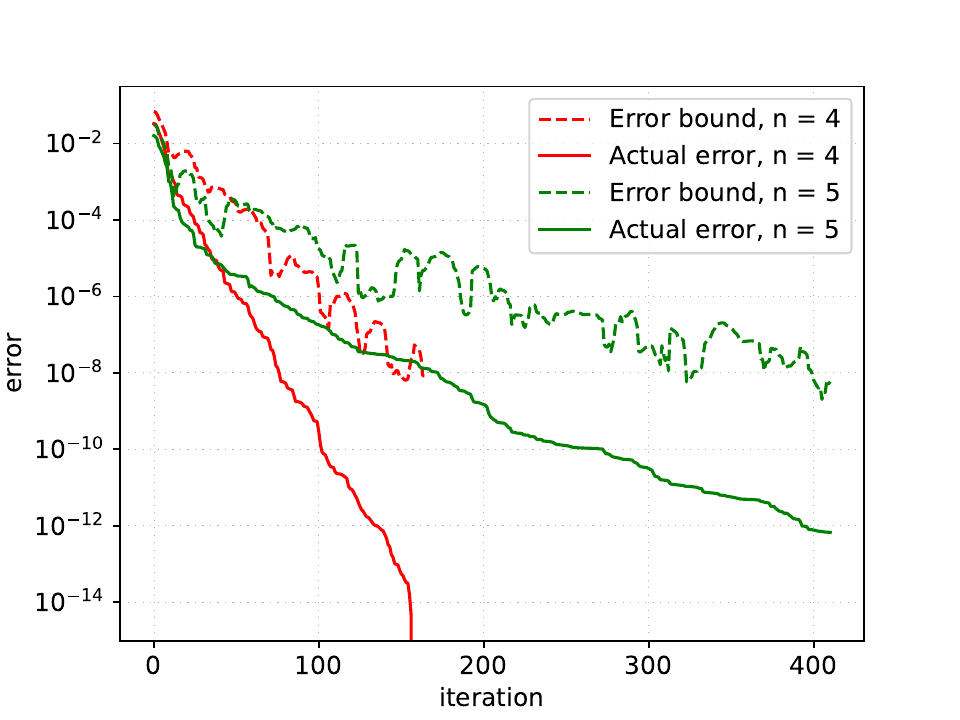}
\caption{Error bound curve for our method, computed using Corollary~\ref{corollary:errorbound}. The dotted and the solid lines represent the error bound and the optimization error, respectively, for example cases of $n= 4, 5$.}                           
\label{fig:comparisons}
\end{figure}

In Figure~\ref{fig:comparisons}, we plot the duality gap bound computed from Corollary~\ref{corollary:errorbound} using L-BFGS and the true optimization error 
$\likelihood(\U\U^\adjoint) - \likelihood^\star$
together to confirm the results of the Corollary and demonstrate that it is reasonably tight.
% show that the duality gap truly upper bounds the optimization error.
The $\likelihood^\star$ value is determined by optimizing the MLE objective using a trusted third-party convex optimization software CVXPY \cite{diamond2016cvxpy, agrawal2018rewriting} using the \texttt{SCS} solver and tolerance set to $10^{-10}$. Since CVXPY can only handle optimizing small quantum systems, we verify our error bound only for $n\le 5$. The trend of the curve for optimization error indicates that the error decreases more or less exponentially with each iteration. The plots indicate that the proposed algorithm converges to a near optimal solution in function value and that our theoretical upper bound on the error is tight enough to be useful.

% ========================================
\subsection{Simulations with tetrahedral POVM with GPU acceleration} The next set of simulations compares our proposed method with the methods in \cite{Lvovsky_2004, photonics10020116, PhysRevResearch.6.033034, PhysRevA.95.062336}. To enable comparison with a broader range of QST algorithms, including few more recent ones, we adapt our simulation setup to align with that of \cite{PhysRevResearch.6.033034}. This also allows us to evaluate the performance of our algorithm under a different measurement setting, extending beyond Pauli POVMs. In Fig. \ref{fig:comparison-tetra4}, we present a performance comparison between our BM-MLE method and other state-of-the-art approaches, focusing specifically on GPU-accelerated execution using \texttt{PyTorch}. For this simulation, the target state to be recovered was a $W$-state under two scenarios---pure state and with $10\%$ depolarizing noise. The measurements for the tomographic procedure were obtained using tetrahedral POVM with $N=100$ shots per POVM. The rank parameter $r$ for $\U$ was heuristically set as $d/4$. 
Since our CPU simulations showed that L-BFGS was faster than Acc-GD, we now focus exclusively on our L-BFGS instantiation. 
We use the \texttt{PyTorch} routine \texttt{torch.optim.LBFGS} for the code, and a memory size of $10$.
All solvers use the QMT for efficiency (implemented on the GPU), since we only test to $n=11$ qubits. We use the QMT implementation from \cite{PhysRevResearch.6.033034}, available at \url{https://github.com/foxwy/QST-UGD}. 
Going much larger than $11$ qubits results in memory issues, especially for GPUs~\footnote{
Specifically, we use the \texttt{qmt\_torch} function from their repository, 
% \protect\url{https://github.com/foxwy/QST-UGD/blob/main/Basis/Basic\_Function.py}
%that was developed with version 2.0.1 of PyTorch, and
which automatically switches from GPU to CPU when $n\ge 12$ because ``torch does not support more dimensional operations,'' though the details of this limitation are perhaps due to memory limitations rather than dimension limits.}.

% May 8, 2023

% The code at 
% The implementation cannot go beyond $12$ qubits due to 

% ======= Tetrahedral ==========

\begin{figure}[ht]
    \centering
    
    \begin{subfigure}[b]{0.455\textwidth}
        \centering
        \includegraphics[width=\textwidth,trim=0cm 0cm 1.4cm 1.4cm,clip]{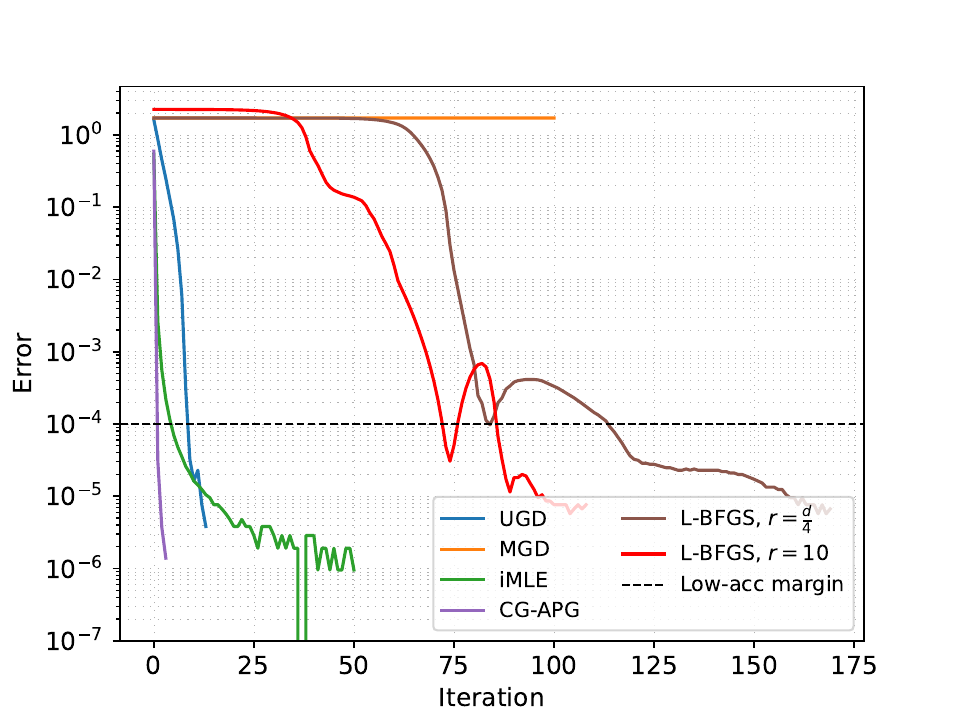}
        \caption{}
        % \caption{No depolarizing noise.}
        \label{fig:tet_errora}
    \end{subfigure}
    \hfill
    \begin{subfigure}[b]{0.45\textwidth}
        \centering
        \includegraphics[width=\textwidth,trim=0cm 0cm 1.6cm 1.4cm,clip]{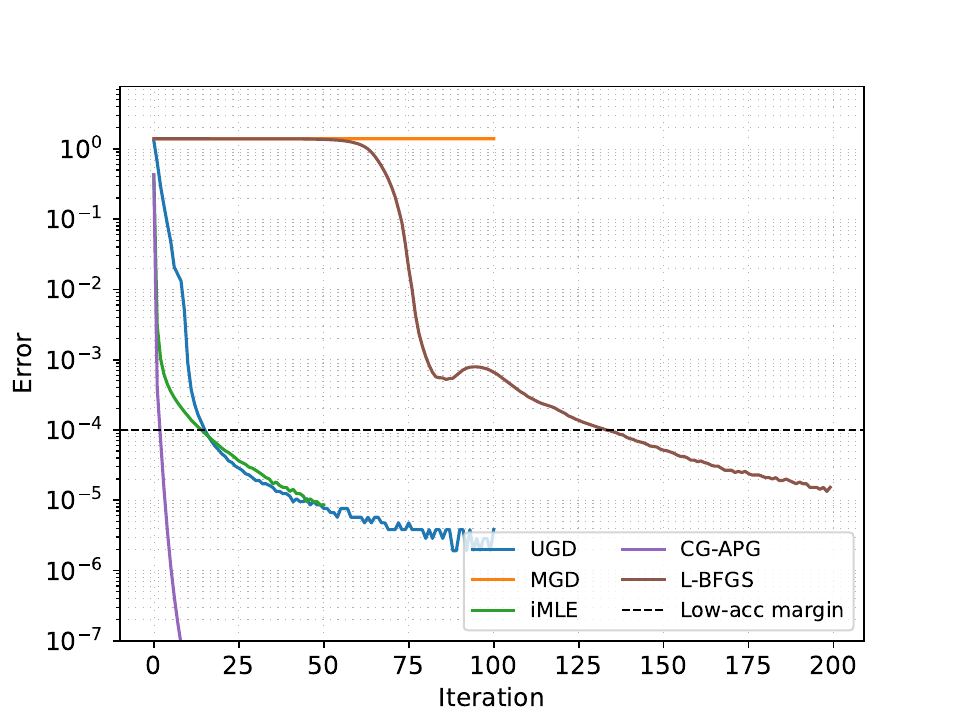}
        \caption{}
        % \caption{$10\%$ depolarizing noise.}
        \label{fig:tet_errorb}
    \end{subfigure}
    
    \caption{Error performance of different methods for tetrahedral POVM measurements on $W$-state and $n=11$ qubits under \textemdash{} (a) no depolarizing noise and (b) $10\%$ depolarizing noise. Again, note that the number of iterations is not equivalent to the per-iteration runtime; refer to Table~\ref{tab:timing}.}
    \label{fig:tet_error}
\end{figure}

Figure~\ref{fig:tet_error} shows the error per iteration for $n=11$, calculated using the same method as in Section~\ref{sec:CPU}. The case without depolarizing noise (Fig.~\ref{fig:tet_errora}) is easy for all solvers other than MGD, and since we reach the error threshold in just a handful of iterations, it is unclear if we can extrapolate this behavior to other settings; also note that while UGD converges in just a few iterations, these iterations are very expensive. Fig.~\ref{fig:tet_errorb} shows the more realistic case with some depolarizing noise, and now all solvers take longer to converge.  As in the Pauli POVM simulation, CG-APG continues to perform very well.  UGD and iMLE are the next best, and L-BFGS takes more iterations (though much less time per iteration).

\begin{figure}[ht]
    \centering
    
    \begin{subfigure}[b]{0.45\textwidth}
        \centering
        \includegraphics[width=\textwidth,trim=0cm 0cm 1.6cm 1.4cm,clip]{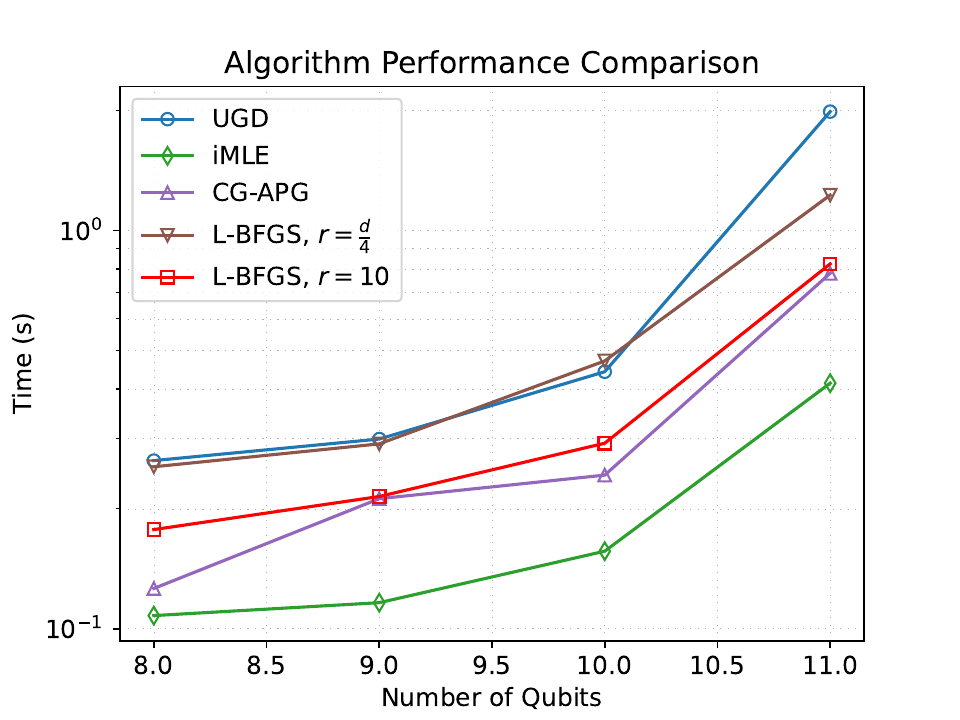}
        % \caption{Low-accuracy regime for no depolarizing noise.}
        \caption{}
        \label{fig:pure}
    \end{subfigure}
    \hfill
    \begin{subfigure}[b]{0.45\textwidth}
        \centering
        \includegraphics[width=\textwidth,trim=0cm 0cm 1.6cm 1.4cm,clip]{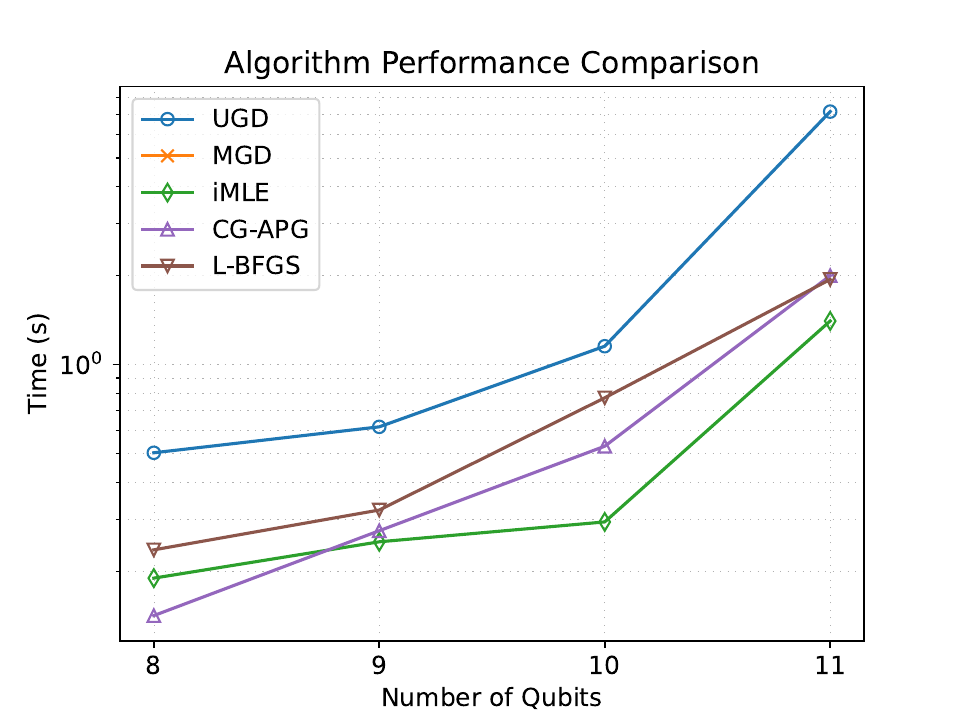}
        \caption{}
        % \caption{Low-accuracy regime for $10\%$ depolarizing noise.}
        \label{fig:noise}
    \end{subfigure}
    
    \caption{Comparison of different methods in terms of time taken to reach the low accuracy regime for tetrahedral POVM measurements on $W$-state under \textemdash{} (a) no depolarizing noise and (b) $10\%$ depolarizing noise.}
    \label{fig:comparison-tetra4}
\end{figure}

As in Section~\ref{sec:CPU}, we also record the time until we hit an error tolerance in order to display performance as a function of the number of qubits. The results are shown in Figure~\ref{fig:comparison-tetra4}.  For the case without depolarizing noise (Fig.~\ref{fig:pure}), all solvers (except MGD) reach the desired accuracy within about 1 second, so they are all adequate and it is hard to speculate about relative performance since the time scales are so short.  Fig.~\ref{fig:noise} shows the case with depolarizing noise. The iMLE method is the fastest, but also shows the greatest slowdown when moving from $n=10$ to $n=11$. The CG-APG method and our L-BFGS method are both equally fast for $n=11$, though our L-BFGS method shows a shallower slope, indicating that it is less sensitive to the size of $n$.

This simulation emphasizes that our method benefits from the computational advantages of parallel processing on the GPU and performs comparably to other methods for moderate qubit systems.
In this regime, computing the objective function is often the most expensive step, so since all methods used the QMT for this step, it is not surprising that all methods (except MGD) perform roughly equivalently.
Again, we note that choosing $r = d/4$ handicaps our methods slightly, since our theoretical advantage is obtained when $r$ is constant or sublinear in $d$. 
As expected, the choice $r=1$ indeed yields better performance in the noiseless setting. However, we emphasize that our primary motivation for including the case $r=d/4$ was to present a consistent benchmark across both noisy and noiseless regimes, since in the presence of noise, the effective rank is typically larger than one. This allows for a fairer comparison between methods and demonstrates the robustness of our approach across different scenarios.

% The main advantages of our method emerge while considering the low-memory version of BM-MLE for even larger systems. 
% TODO, mention r propto d issue again, and QMT... 

% ========================================
\subsection{Low-memory BM-MLE on large systems} \label{sec:lowmemory}

\subsubsection{Sampling a subset of POVMs}\label{sec:sample-Pauli}
For large systems, one cannot afford to take $d^2$ measurements, so we subsample from the set of Pauli POVMs.  Uniform subsampling is suboptimal, so we turn to the weighted subsampling scheme introduced in \cite{PhysRevLett.106.230501} for the purposes of ``direct fidelity estimation'' (DFE). They show that it is possible to estimate fidelity $F(\rho, \sigma)$ up to a constant additive error $\epsilon$ and failure probability $\delta$ by only measuring a constant number of Pauli observables for \emph{well-conditioned} pure target state $\rho$ including states of interest in quantum information such as GHZ and W states, stabilizer states, matrix product states, etc. This subset of Pauli observables is chosen according to a nonuniform probability distribution given by $\trace(W_i\rho)^2/d$.  The chosen measurements are either sampled from this distribution (without replacement) until the measurement budget is reached, or a common variant is to sort the probabilities and take the top largest ones.

When $n\gtrsim 18$,  computing the nonuniform probability distribution is infeasible, and one must sample from it implicitly. Details on how to do this are in Appendix~\ref{appendix:lowmem}.

\subsubsection{Simulations} The previous simulations showed that the L-BFGS instantiation of our algorithm was competitive with state-of-the-art methods (the best of which was CG-APG) on moderately sized problems up to $n=11$ qubits using $r=d/4$. Our final set of simulations is designed to showcase the regime where our method excels---and indeed there is no comparison to other methods, since none of them can handle this scenario due to memory constraints. We consider our low-memory version (i.e., using the techniques in Section~\ref{sec:lowmem} instead of the QMT), applied to 
% In the final set of experiments, we demonstrate the successful application of our low-memory version of BM-MLE applied to an 
$18$ and $20$-qubit systems, for a fixed rank $r=1$.

%In this experiment, we try to verify how close the reconstructed state is to the known physical state. 

In more detail, 
we consider a random product state (i.e., separable) as the physical state $\rho_0$ and measurements are derived using a subset of Pauli POVMs that do not form an informationally complete set. We make the above simulation choice for the reasons discussed in the beginning of this section.
%, under noisy and noiseless settings. 
It is to be noted that our low-memory algorithm applied to large-systems does \emph{not} make use of the separability of the state. The choice of a separable state is made only to simplify the computation of the probability distribution for subsampling POVMs, as used in Appendix~\ref{appendix:lowmem}. We let $N=\infty$ so that we have perfect measurements on those POVMs that we do measure.
Because this is not an informationally complete set of measurements, there is no unique solution $\rho^\star$, but since we take exact measurements, Proposition~\ref{prop:1} guarantees that the physical state $\rho_0$ is one of the global MLE solutions, hence we know $\likelihood^\star = \likelihood(\rho_0)$.

We consider two cases: using the target state $\rho_0$ directly or adding $10\%$ depolarizing noise. 
When we use $\rho_0$ directly, since it is a product state, the probability  $\trace(A_i \rho_0)$  factorizes as the product of single-qubit probabilities and more details are provided in Appendix ~\ref{appendix:lowmem}. %as explained earlier in Section~\ref{sec:sample-Pauli}. 
In implementation, we can sample operators for each qubit independently according to its distribution. We sample $\nPOVM=2559$ operators in total using the formula $\nPOVM=\lceil \log(1/\delta)/ \epsilon^{2} \rceil$ with $\delta=0.1$ and  $\epsilon=0.03$.
This formula is motivated by the \emph{direct fidelity estimation} method of  \cite{PhysRevLett.106.230501}
where  $\delta$ is a failure probability and $\epsilon$ is an additive error. Since we are doing tomography and not directly estimating fidelity, and not using their post-processing, their error bound does not apply, but we find it is a useful rule-of-thumb.
When we consider the case with depolarizing noise, the probability does not factor since the noisy state is not a product state, but for efficiency we sample from the same factored formula as in the first case. This means that the probability distribution is not exactly equivalent to the one from \cite{PhysRevLett.106.230501}, but it still gives a useful subset of samples.

We do not consider our Acc-GD instantiation since the L-BFGS version has shown to be faster. We use the \texttt{scipy.optimize} version of L-BFGS with a history size of $10$. % and initializing our guess for $\uu$ as a complex standard normal vector.
% We use the L-BFGS optimizer initialized with a vector whose entries are drawn from the standard normal distribution and the history size was set to 10.

% Although this analysis is specific to the setup of \cite{PhysRevLett.106.230501}, we found that choosing $M$ in this manner works well for our own setup too.

\begin{figure}[ht]
    \centering
    
    \begin{subfigure}[b]{0.45\textwidth}
        \centering
        \includegraphics[width=\textwidth,trim=0.2cm 0.3cm 0.2cm 0.2cm,clip]{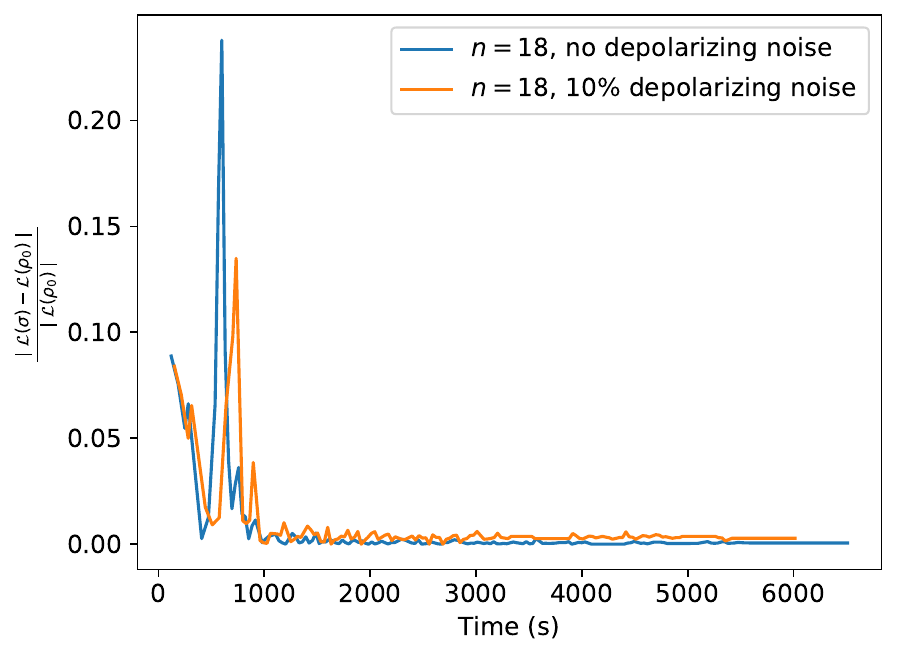}
        % \caption{$n=18$}
        \caption{}
        \label{fig:low-mem-18}
    \end{subfigure}
    \hfill
    \begin{subfigure}[b]{0.45\textwidth}
        \centering
        \includegraphics[width=\textwidth,trim=0.2cm 0.3cm 0.2cm 0.2cm,clip]{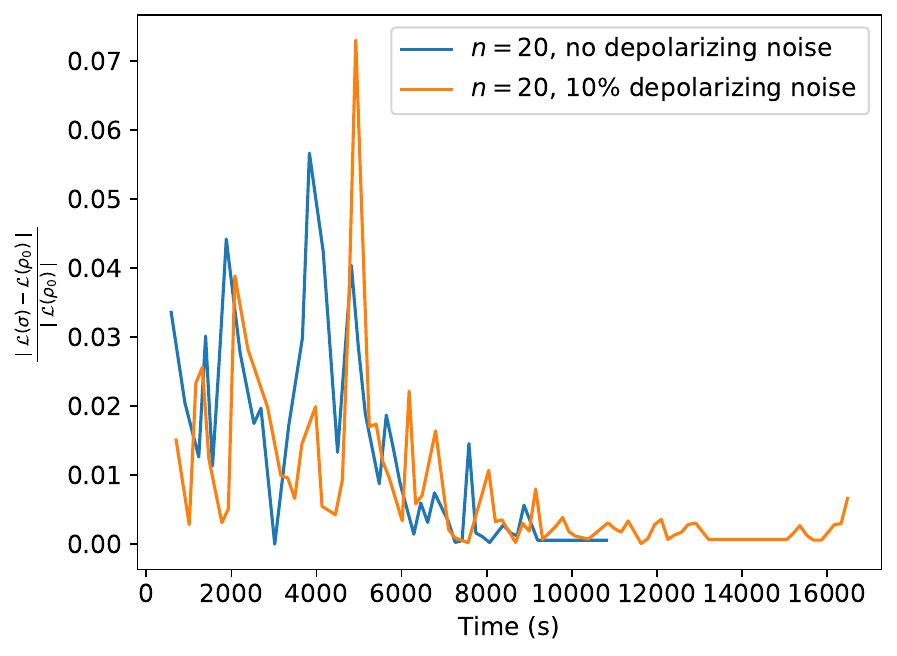}
        % \caption{$n=20$}
        \caption{}
        \label{fig:low-mem-20}
    \end{subfigure}
    \caption{Likelihood error curves obtained with low-memory BM-MLE for \textemdash{} (a) $18$ and (b) $20$-qubit systems with random product state as target.}
    \label{fig:lowmem}
\end{figure}

% Fig. \ref{fig:low-mem-18} shows the relative error in the value of likelihood function, plotted against elapsed time, for both the no-noise and $10\%$ depolarizing noise cases. We observe that the error converges in about $\sim 1.5$ hours and roughly decreases by $2.3$ orders of magnitude with a minimum error of $7.71\times 10^{-6}$. Whereas, the error decreases approximately by $1.5$ orders of magnitude in the presence of depolarizing noise with a minimum error of $9.03 \times 10^{-6}$. For the case of $20$-qubit in Fig. \ref{fig:low-mem-20}, the decrease in error is slightly lower compared to $18$-qubit with a minimum of $2.94\times 10^{-6}$, and the error fails to converge in the presence of depolarizing noise but the minimum attained is $9.86\times 10^{-5}$. 

The results are shown in Figure~\ref{fig:lowmem} and Table~\ref{tab:lowmem}.
The table shows the relative error in the negative log-likelihood of the best iterate the algorithm produced, $\rho_{\text{best}}$. 
This scenario is considerably challenging and takes on the order of a few hours to run on a contemporary laptop (we did CPU only, no GPU), but does demonstrate that tomography up to $20$ qubits is feasible.  We only look at the optimization error in the tomography; statistical accuracy will improve as we take more measurements, and the cost of running our algorithm is directly proportional to the number of measurements, cf.~the right column of Table~\ref{tab:complexity}.  To give some context, in the 18 qubit case with no depolarizing noise, our algorithm gives an estimate with a negative log-likelihood that is over $2$ orders of magnitude better than the value at the random initialization.

\begin{table*}[ht]
    \centering
    \begin{tabular}{p{2.2cm}*{4}{p{1.9cm}}}
    \toprule
     & \multicolumn{2}{c}{no depolarizing noise} & \multicolumn{2}{c}{$10\%$ depolarizing noise} \\ 
     \cmidrule(r){2-3} \cmidrule(l){4-5}
    qubits & rel.~error & time & rel.~error & time \\
    \midrule 
     18    & $7.71\times 10^{-6}$ & $1.8$ hours & $9.03 \times 10^{-6}$ & $1.7$ hours \\
      20   & $2.94\times 10^{-6}$ & $3.1$ hours & $9.86\times 10^{-5}$ & $4.6$ hours\\
      \bottomrule
    \end{tabular}
    \caption{Relative error of the best $\rho$ found by the algorithm, $\frac{\likelihood(\rho_{\text{best}}) - \likelihood^\star}{\likelihood^\star}$, as well as wall clock time for the algorithm. }
    \label{tab:lowmem}
\end{table*}

\section{Conclusion and future work}
We have described an unconstrained, non-convex reformulation for solving the problem of rank-controlled quantum state tomography using maximum likelihood estimation. We showed that with Burer-Monteiro factorization, the convex MLE becomes non-convex but this has little practical downside, and the advantage is that we are able to make the problem unconstrained. % with no optimization parameters to be set.
The rank $r$ of the factor $\U$ is directly proportional to the runtime and memory usage of the algorithm, and introduces little bias as long as the true MLE solution is well-approximated by a rank $r$ state.

Our numerical simulations show that our method is comparable with state-of-the-art methods when there are no memory limitations, i.e., up to about $11$ qubits.  For larger systems, alternative methods have memory issues. Through careful implementation and algorithm design, we are able to run our method up to $20$ qubits, achieving a relative error under $10^{-4}$ within a few hours of computation time.
% which when set close the true rank, yields good performance results. With our low-memory BM-MLE, we go beyond the results demonstrated in the related literature. We showed its performance on $18$ and $20$-qubit systems where both the relative error and the runtime of the algorithm are within accepted limits.
While reconstruction of $20$ qubit states has been done before, it relied on the MPS structure~\cite{PRXQuantum.4.040345}. Our method also requires a structural assumption (that the matrix is well-approximated by a low-rank state) but our structural assumption is in some ways milder than the MPS structure.

% A possible direction for future work can be extending low-memory implementation to more general, non-product states and to other POVM measurements.

Interesting future research can take both theoretical and practical directions. In our numerical simulations, the non-convexity of our formulation has been benign, and existing generic Burer-Monteiro theory explains some of this, but it would be interesting to specialize that theory to our particular negative log-likelihood objective and obtain stronger guarantees. 

Practically speaking, there are at least two interesting avenues of work. The first would be making algorithmic improvements, such as improving on L-BFGS, or deriving a cheap and effective line search, etc.  
The second avenue would be to target variants of MLE, motivated by uncertainty quantification (UQ).  Maximum likelihood estimates themselves are not especially useful unless they have an accompanying credible estimate of uncertainty. This uncertainty is often approximated via Monte Carlo methods, such as bootstrap resampling. These techniques require solving very similar MLE problems thousands of times, so a natural question is whether one can design an algorithm that speeds up this process (but without introducing statistical bias).

Another variant of MLE is profile likelihood estimation. As the number of qubits increases, it is not realistic to take informationally complete measurements, so the statistical problem becomes ill-posed, and the task of solving the MLE efficiently is almost irrelevant.  One interesting work-around is to target a specific metric (say, the fidelity with respect to a fixed state), and then to take the set of all estimates that have a sufficiently high likelihood, and take the minimum and maximum of this metric over that set. This gives something akin to a confidence interval for the estimate of the metric. This technique has been used in quantum tomography at smaller scales~\cite{seshadri2024theory}, and the generic statistical technique, which more generally integrates over an irrelevant nuisance parameter, is known as profile likelihood.

\appendix
% Specify following sections are appendices. Use \appendix* if there
% only one appendix.
%\appendix
%\section{}
\section{}\label{appendix:qmt}

The pseudocodes for the sub-routines of the QMT algorithm \cite{PhysRevA.95.062336} described in Section~\ref{sec:QMT} are outlined in Algo.~\ref{alg:shuffle-forward} and Algo.~\ref{alg:QMT} below.
\SetKwComment{Comment}{// }{}
\begin{algorithm}[H]
\DontPrintSemicolon
\caption{\textsc{ShuffleForward}~\cite{PhysRevA.95.062336} — Reshape and interleave tensor axes}
\label{alg:shuffle-forward}
\KwIn{$ \rho \in \mathbb{C}^{d \times d} ,  \vec{d} = (d_1, \dots, d_n)$ \tcp{$d_i = 2$ for qubits}}

    \( n \gets \text{length}(\vec{d}) \) \tcp*{Number of subsystems}
    Reshape \( \rho \in \mathbb{C}^{d \times d} \) to shape \( (d_1, \dots, d_n, d_1, \dots, d_n) \)\;
    Define axes ordering: \\
    \hspace{2em} \( \pi \gets (1, n+1, 2, n+2, \dots, n, 2n) \)\;
    Permute \( \rho \) using ordering \( \pi \)\;
    %\Return \( \rho \)\;

\kwReturn{$\rho \in \mathbb{C}^{d_1 \times d_1 \times  d_2 \times d_2 \dots \times d_n \times d_n} $ }
\end{algorithm}

\SetKwComment{Comment}{// }{}
% or, comments like \tcp{TBD} or \tcc{TBD} for \\ and \* ... * style

\begin{algorithm}[H]
\DontPrintSemicolon
\caption{Quantum Measurement Transform (QMT) \cite{PhysRevA.95.062336, PhysRevResearch.6.033034} for Pauli POVMs}\label{alg:QMT}
\kwInput{$\rho \in \mathbb{C}^{d \times d}$ \tcp{$n=\log_2(d)$ qubits}}

$\mat{P} = [\vec{I}, \vec{\sigma_x}, \vec{\sigma_y},  \vec{\sigma_z}] \in\C^{4 \times 4}$\;

$x \gets \texttt{\textsc{ShuffleForward}}(\rho)$\;
\For {$i = 1 \dots n $}{
Reshape: $\vec{x}_{\scriptstyle K \times 2^{n-2}}$ \tcp{$K = 4$} 
$\vec{x} \gets \mat{P}^\adjoint \vec{x}$\;
$\vec{x} \gets \vec{x}^\top$\;
}
% \kwReturn{$\vec{x}\in\R^{2^{n-2}\times K}$ where $x_i= \trace(W_i \rho)$  \tcp{$W_i =  \sigma_{i_1} \otimes \sigma_{i_2} \otimes \cdots \otimes \sigma_{i_n}, \base_{4}\footnote{base-4 representation.}(i) = (i_1, \dots, i_n) $} }
\kwReturn{$\vec{x}\in\R^{2^{n-2}\times K}$ where $x_i= \trace(W_i \rho)$  \tcp{$W_i =  \sigma_{i_1} \otimes \sigma_{i_2} \otimes \cdots \otimes \sigma_{i_n}, \base_{4}(i) = (i_1, \dots, i_n) $} }
\end{algorithm}

\section{}
% MISC STUFF

\subsection{Vectorization of Hermitian matrices} \label{sec:vectorization}
% On the real vector space of 
% $\H^d$
% %$\C^{d\times d}$ 
% we use the standard inner product $\langle \rho , \sigma \>_\H = \Re[ \trace(\rho^\top \sigma) ]$ where $\Re$ indicates taking the real part (though if $\rho,\sigma\succeq \mat{0}$ then $\trace(\rho^\top \sigma)\in\R$). This is isometrically isomorphic to the vector space $\R^{d^2}$ with the usual Euclidean dot product under the mapping $\hvec: \H^d \to \R^{d^2}$ which vectorizes (in any fixed manner) the $d$ real entries on the diagonal along with $\sqrt{2}$ times the $d(d-1)/2$ real and $d(d-1)/2$ imaginary entries from the upper triangular portion of the input matrix; we define $\hmat=\hvec^{-1}$ which reshapes the real vector back into the Hermitian matrix. This vectorization is useful for deriving certain results as well as for numerical implementation of algorithms and has little consequential effect on the math since it is a unitary linear transformation.

On the real vector space of complex Hermitian $d\times d$ matrices 
$\H^d$
%$\C^{d\times d}$ 
we use the standard inner product $\langle \rho , \sigma \>_\H = \Re[ \trace(\rho^\top \sigma) ]$ where $\Re$ indicates taking the real part (though if $\rho,\sigma\succeq \mat{0}$ then $\trace(\rho^\top \sigma)\in\R$). This is isometrically isomorphic to the vector space $\R^{d^2}$ with the usual Euclidean dot product under the mapping $\hvec: \H^d \to \R^{d^2}$ which vectorizes (in any fixed manner) the $d$ real entries on the diagonal along with $\sqrt{2}$ times the $d(d-1)/2$ real and $d(d-1)/2$ imaginary entries from the upper triangular portion of the input matrix; we define $\hmat=\hvec^{-1}$ which reshapes the real vector back into the Hermitian matrix. This vectorization is useful for deriving certain results as well as for numerical implementation of algorithms and has little consequential effect on the math since it is a unitary linear transformation.

\newcommand{\vecrho}{\vec{x}} % choose something better>
It is sometimes convenient to examine $\likelihood$ when we vectorize the input. Letting $\vecrho = \hvec(\rho)$ and $\vec{a}_i = \hvec(\overline{A_i})\in\R^{d^2}$, we define
\begin{align}\label{eq:grad_l}
&\ell(\vecrho) = \likelihood(\hmat(\vecrho)) = -\sum_{i=1}^{\nTotalPOVM} f_i \log( \vec{a}_i^\top\vecrho ), \\ %\quad 
&\text{hence}\; \nabla \ell:\R^{d^2}\to\R^{d^2}, \;
    \nabla \ell(\vecrho) = -\sum_{i=1}^{\nTotalPOVM} \frac{f_i}{\vec{a}_i^\top\vecrho} \vec{a}_i.
\end{align}
The Hessian of $\likelihood$, $\nabla^2 \likelihood$, is a $d\times d\times d\times d$ tensor whereas the Hessian of $\ell$ is a $d^2\times d^2$ symmetric matrix that is easier to work with. Specifically,
\begin{equation}\label{eq:hess_l}
 \nabla^2 \ell:\R^{d^2}\to\R^{d^2 \times d^2}, \;
    \nabla^2 \ell(\vecrho) = \sum_{i=1}^{\nTotalPOVM} \frac{f_i}{\left(\vec{a}_i^\top\vecrho\right)^2} \vec{a}_i\vec{a}_i^\top \succeq \mat{0}.
\end{equation}    
Since $\nabla^2 \ell(\vecrho) \succeq \mat{0}$, $\ell$ is convex, and since $\likelihood(\rho)=\ell(\hvec)$ is the composition of a convex function with a linear function, $\likelihood$ is also convex.

\subsection{More information on the factored objective function}
The Hessian of $\f$ is  a little unwieldy to work with since it is a tensor for $r>1$, but for reference we state it below for the special case when $r=1$ (and write $\vec{u}$ rather than $\U$ to indicate this):

\begin{multline}
    \nabla^2\f(\vec{u}) =   2\sum_{i=1}^{\nTotalPOVM} f_i\Bigg( \frac{2}{(\vec{u}^\adjoint  A_i  \vec{u})^{2}} A_i\vec{u} (\vec{u}^\adjoint  A_i) \\
    -\frac{1}{\vec{u}^\adjoint  A_i\vec{u}} A_i \Bigg) \in \H^d. %\C^{d\times d}.
\end{multline}

\section{Application of Pauli matrices} \label{appendix:apply-paulis}

\begin{algorithm}[H]
\DontPrintSemicolon
\caption{\textsc{ApplyX} — Apply Pauli-\(X\) to the \(i\)-th qubit}
\label{alg:applyX}
\KwIn{\( \U \in \mathbb{C}^{2^n \times r} \), qubit index \( k \), $d$}
%\KwOut{State with \( \sigma_X \) applied to qubit \( i \)}

% \Begin{
    \( n \gets \log_2(d) \)\;

    Construct index shift pattern: \\
    \hspace{2em} $ \text{order} \gets \{\{+2^k\}^{2^k}, \{-2^k\}^{2^k} \}^{2^{(n-k-1)}}$\;

    \( \text{idx} \gets [0, 1, \dots, d-1] \)\;

% }
\kwReturn {$ \U[\text{idx} + \text{order}, :] $}
\end{algorithm}

\begin{algorithm}[H]
\DontPrintSemicolon
\caption{\textsc{ApplyZ} — Apply Pauli-\(Z\) to the \(i\)-th qubit}
\label{alg:applyZ}
\KwIn{ \( \U \in \mathbb{C}^{2^n \times r} \), qubit index \( k \), \(d \)}
%\KwOut{State with \( \sigma_Z \) applied to qubit \( i \)}

% \Begin{
    \( n \gets \log_2(d) \)\;

    Construct sign pattern: \\
    \hspace{2em} $ \text{scale}  \gets \{\{+1^k\}^{2^k}, \{-1^k\}^{2^k} \}^{2^{(n-k-1)}}$\;

    $\U \gets \text{scale} \times \U$ \tcp*{row-wise scaling via broadcasting}
    
% }
\kwReturn {$ \U $}
\end{algorithm}
We use the notation $\{a\}^b$  to denote a set containing the value $a $, repeated $b$ times.

\section{Simulation details of Sec.~\ref{sec:sample-Pauli}}\label{appendix:lowmem}

In our low-memory setting with $n \gtrsim 16$, it is too computationally intensive to even calculate all $d^2$ of the probabilities needed for the weighted DFE dsampling.  One could avoid precomputing all probabilities and calculate them only as needed, in a rejection-sampling framework, but this is still inefficient since the rejection rate (using uniform sampling as the base distribution) would be too high.

To circumvent the issue, when we do experiments with very large qubit systems, we will generate the physical state $\rho_0$ as a separable state $\rho_0 = \rho^{(1)} \otimes \ldots \otimes \rho^{(n)}$, so for an $n$-qubit Pauli $W = \sigma_1 \otimes \ldots \otimes \sigma_n$, then
\begin{align*}
p = \trace(W\rho_0)^2
&= \trace( (\sigma_1 \rho^{(1)} ) \otimes \ldots 
\otimes (\sigma_n \rho^{(n)} ) )^2 \\
&= \underbrace{\trace(\sigma_1 \rho^{(1)} )^2}_{p_1} \cdot \ldots 
\cdot \underbrace{\trace(\sigma_n \rho^{(n)} ) )^2}_{p_n} \\
&= p_1 p_2\cdots p_n
\end{align*}
via properties of trace and the Kronecker product.
Since $\rho_0$ is fixed, there are only 4 values $p_i$ can take since $\sigma_i \in \{I, \sigma_x, \sigma_y, \sigma_z \}$.  Thus to sample from this distribution (that is, draw a $W = \sigma_1 \otimes \ldots \otimes \sigma_n$), we first sample $\sigma_1$ from the 4-element distribution 
$\{ \trace(I \rho^{(1)} )^2,\trace(\sigma_x \rho^{(1)} )^2,\trace(\sigma_y \rho^{(1)} )^2,\trace(\sigma_z \rho^{(1)} )^2 \}$, then sample $\sigma_2$ and the rest in a similar fashion. By building up the sampled Pauli matrix one qubit at a time, the process is very efficient.  %When $\rho_0$ is entangled, this derivation does not apply, though we hypothesize that it may be modified using similar ideas. %and form a useful base distribution for importance sampling.
Efficiently sampling from weighted distributions for entangled states in the $n \gtrsim 16$ regime is an interesting topic for further research.

% If you have acknowledgments, this puts in the proper section head.
\begin{acknowledgments}
This material is based upon work supported by the National
Science Foundation under Grant No.\ 2106834.
\end{acknowledgments}

\bibliography{ref}% Produces the bibliography via BibTeX.

%apsrev4-2.bst 2019-01-14 (MD) hand-edited version of apsrev4-1.bst
%Control: key (0)
%Control: author (8) initials jnrlst
%Control: editor formatted (1) identically to author
%Control: production of article title (0) allowed
%Control: page (0) single
%Control: year (1) truncated
%Control: production of eprint (0) enabled
\providecommand{\noopsort}[1]{}\providecommand{\singleletter}[1]{#1}%
\begin{thebibliography}{72}%
\makeatletter
\providecommand \@ifxundefined [1]{%
 \@ifx{#1\undefined}
}%
\providecommand \@ifnum [1]{%
 \ifnum #1\expandafter \@firstoftwo
 \else \expandafter \@secondoftwo
 \fi
}%
\providecommand \@ifx [1]{%
 \ifx #1\expandafter \@firstoftwo
 \else \expandafter \@secondoftwo
 \fi
}%
\providecommand \natexlab [1]{#1}%
\providecommand \enquote  [1]{``#1''}%
\providecommand \bibnamefont  [1]{#1}%
\providecommand \bibfnamefont [1]{#1}%
\providecommand \citenamefont [1]{#1}%
\providecommand \href@noop [0]{\@secondoftwo}%
\providecommand \href [0]{\begingroup \@sanitize@url \@href}%
\providecommand \@href[1]{\@@startlink{#1}\@@href}%
\providecommand \@@href[1]{\endgroup#1\@@endlink}%
\providecommand \@sanitize@url [0]{\catcode `\\12\catcode `\$12\catcode `\&12\catcode `\#12\catcode `\^12\catcode `\_12\catcode `\%12\relax}%
\providecommand \@@startlink[1]{}%
\providecommand \@@endlink[0]{}%
\providecommand \url  [0]{\begingroup\@sanitize@url \@url }%
\providecommand \@url [1]{\endgroup\@href {#1}{\urlprefix }}%
\providecommand \urlprefix  [0]{URL }%
\providecommand \Eprint [0]{\href }%
\providecommand \doibase [0]{https://doi.org/}%
\providecommand \selectlanguage [0]{\@gobble}%
\providecommand \bibinfo  [0]{\@secondoftwo}%
\providecommand \bibfield  [0]{\@secondoftwo}%
\providecommand \translation [1]{[#1]}%
\providecommand \BibitemOpen [0]{}%
\providecommand \bibitemStop [0]{}%
\providecommand \bibitemNoStop [0]{.\EOS\space}%
\providecommand \EOS [0]{\spacefactor3000\relax}%
\providecommand \BibitemShut  [1]{\csname bibitem#1\endcsname}%
\let\auto@bib@innerbib\@empty
%</preamble>
\bibitem [{\citenamefont {Cramer}\ \emph {et~al.}(2010)\citenamefont {Cramer}, \citenamefont {Plenio}, \citenamefont {Flammia}, \citenamefont {Somma}, \citenamefont {Gross}, \citenamefont {Bartlett}, \citenamefont {Landon-Cardinal}, \citenamefont {Poulin},\ and\ \citenamefont {Liu}}]{cramer2010efficient}%
  \BibitemOpen
  \bibfield  {author} {\bibinfo {author} {\bibfnamefont {M.}~\bibnamefont {Cramer}}, \bibinfo {author} {\bibfnamefont {M.~B.}\ \bibnamefont {Plenio}}, \bibinfo {author} {\bibfnamefont {S.~T.}\ \bibnamefont {Flammia}}, \bibinfo {author} {\bibfnamefont {R.}~\bibnamefont {Somma}}, \bibinfo {author} {\bibfnamefont {D.}~\bibnamefont {Gross}}, \bibinfo {author} {\bibfnamefont {S.~D.}\ \bibnamefont {Bartlett}}, \bibinfo {author} {\bibfnamefont {O.}~\bibnamefont {Landon-Cardinal}}, \bibinfo {author} {\bibfnamefont {D.}~\bibnamefont {Poulin}},\ and\ \bibinfo {author} {\bibfnamefont {Y.-K.}\ \bibnamefont {Liu}},\ }\bibfield  {title} {\bibinfo {title} {Efficient quantum state tomography},\ }\href@noop {} {\bibfield  {journal} {\bibinfo  {journal} {Nature communications}\ }\textbf {\bibinfo {volume} {1}},\ \bibinfo {pages} {149} (\bibinfo {year} {2010})}\BibitemShut {NoStop}%
\bibitem [{\citenamefont {Flammia}\ \emph {et~al.}(2012)\citenamefont {Flammia}, \citenamefont {Gross}, \citenamefont {Liu},\ and\ \citenamefont {Eisert}}]{Flammia_2012}%
  \BibitemOpen
  \bibfield  {author} {\bibinfo {author} {\bibfnamefont {S.~T.}\ \bibnamefont {Flammia}}, \bibinfo {author} {\bibfnamefont {D.}~\bibnamefont {Gross}}, \bibinfo {author} {\bibfnamefont {Y.-K.}\ \bibnamefont {Liu}},\ and\ \bibinfo {author} {\bibfnamefont {J.}~\bibnamefont {Eisert}},\ }\bibfield  {title} {\bibinfo {title} {Quantum tomography via compressed sensing: error bounds, sample complexity and efficient estimators},\ }\href {https://doi.org/10.1088/1367-2630/14/9/095022} {\bibfield  {journal} {\bibinfo  {journal} {New Journal of Physics}\ }\textbf {\bibinfo {volume} {14}},\ \bibinfo {pages} {095022} (\bibinfo {year} {2012})}\BibitemShut {NoStop}%
\bibitem [{\citenamefont {Baumgratz}\ \emph {et~al.}(2013{\natexlab{a}})\citenamefont {Baumgratz}, \citenamefont {Gross}, \citenamefont {Cramer},\ and\ \citenamefont {Plenio}}]{PhysRevLett.111.020401}%
  \BibitemOpen
  \bibfield  {author} {\bibinfo {author} {\bibfnamefont {T.}~\bibnamefont {Baumgratz}}, \bibinfo {author} {\bibfnamefont {D.}~\bibnamefont {Gross}}, \bibinfo {author} {\bibfnamefont {M.}~\bibnamefont {Cramer}},\ and\ \bibinfo {author} {\bibfnamefont {M.~B.}\ \bibnamefont {Plenio}},\ }\bibfield  {title} {\bibinfo {title} {Scalable reconstruction of density matrices},\ }\href {https://doi.org/10.1103/PhysRevLett.111.020401} {\bibfield  {journal} {\bibinfo  {journal} {Phys. Rev. Lett.}\ }\textbf {\bibinfo {volume} {111}},\ \bibinfo {pages} {020401} (\bibinfo {year} {2013}{\natexlab{a}})}\BibitemShut {NoStop}%
\bibitem [{\citenamefont {Gu{\c{t}}{\u{a}}}\ \emph {et~al.}(2020)\citenamefont {Gu{\c{t}}{\u{a}}}, \citenamefont {Kahn}, \citenamefont {Kueng},\ and\ \citenamefont {Tropp}}]{Guta_2020}%
  \BibitemOpen
  \bibfield  {author} {\bibinfo {author} {\bibfnamefont {M.}~\bibnamefont {Gu{\c{t}}{\u{a}}}}, \bibinfo {author} {\bibfnamefont {J.}~\bibnamefont {Kahn}}, \bibinfo {author} {\bibfnamefont {R.}~\bibnamefont {Kueng}},\ and\ \bibinfo {author} {\bibfnamefont {J.~A.}\ \bibnamefont {Tropp}},\ }\bibfield  {title} {\bibinfo {title} {Fast state tomography with optimal error bounds},\ }\href {https://doi.org/10.1088/1751-8121/ab8111} {\bibfield  {journal} {\bibinfo  {journal} {Journal of Physics A: Mathematical and Theoretical}\ }\textbf {\bibinfo {volume} {53}},\ \bibinfo {pages} {204001} (\bibinfo {year} {2020})}\BibitemShut {NoStop}%
\bibitem [{\citenamefont {Eisert}\ \emph {et~al.}(2020)\citenamefont {Eisert}, \citenamefont {Hangleiter}, \citenamefont {Walk}, \citenamefont {Roth}, \citenamefont {Markham}, \citenamefont {Parekh}, \citenamefont {Chabaud},\ and\ \citenamefont {Kashefi}}]{eisert2020quantum}%
  \BibitemOpen
  \bibfield  {author} {\bibinfo {author} {\bibfnamefont {J.}~\bibnamefont {Eisert}}, \bibinfo {author} {\bibfnamefont {D.}~\bibnamefont {Hangleiter}}, \bibinfo {author} {\bibfnamefont {N.}~\bibnamefont {Walk}}, \bibinfo {author} {\bibfnamefont {I.}~\bibnamefont {Roth}}, \bibinfo {author} {\bibfnamefont {D.}~\bibnamefont {Markham}}, \bibinfo {author} {\bibfnamefont {R.}~\bibnamefont {Parekh}}, \bibinfo {author} {\bibfnamefont {U.}~\bibnamefont {Chabaud}},\ and\ \bibinfo {author} {\bibfnamefont {E.}~\bibnamefont {Kashefi}},\ }\bibfield  {title} {\bibinfo {title} {Quantum certification and benchmarking},\ }\href@noop {} {\bibfield  {journal} {\bibinfo  {journal} {Nature Reviews Physics}\ }\textbf {\bibinfo {volume} {2}},\ \bibinfo {pages} {382} (\bibinfo {year} {2020})}\BibitemShut {NoStop}%
\bibitem [{\citenamefont {Mogilevtsev}\ \emph {et~al.}(2013)\citenamefont {Mogilevtsev}, \citenamefont {Hradil}, \citenamefont {Rehacek},\ and\ \citenamefont {Shchesnovich}}]{PhysRevLett.111.120403}%
  \BibitemOpen
  \bibfield  {author} {\bibinfo {author} {\bibfnamefont {D.}~\bibnamefont {Mogilevtsev}}, \bibinfo {author} {\bibfnamefont {Z.}~\bibnamefont {Hradil}}, \bibinfo {author} {\bibfnamefont {J.}~\bibnamefont {Rehacek}},\ and\ \bibinfo {author} {\bibfnamefont {V.~S.}\ \bibnamefont {Shchesnovich}},\ }\bibfield  {title} {\bibinfo {title} {Cross-validated tomography},\ }\href {https://doi.org/10.1103/PhysRevLett.111.120403} {\bibfield  {journal} {\bibinfo  {journal} {Phys. Rev. Lett.}\ }\textbf {\bibinfo {volume} {111}},\ \bibinfo {pages} {120403} (\bibinfo {year} {2013})}\BibitemShut {NoStop}%
\bibitem [{\citenamefont {Stricker}\ \emph {et~al.}(2022)\citenamefont {Stricker}, \citenamefont {Meth}, \citenamefont {Postler}, \citenamefont {Edmunds}, \citenamefont {Ferrie}, \citenamefont {Blatt}, \citenamefont {Schindler}, \citenamefont {Monz}, \citenamefont {Kueng},\ and\ \citenamefont {Ringbauer}}]{PRXQuantum.3.040310}%
  \BibitemOpen
  \bibfield  {author} {\bibinfo {author} {\bibfnamefont {R.}~\bibnamefont {Stricker}}, \bibinfo {author} {\bibfnamefont {M.}~\bibnamefont {Meth}}, \bibinfo {author} {\bibfnamefont {L.}~\bibnamefont {Postler}}, \bibinfo {author} {\bibfnamefont {C.}~\bibnamefont {Edmunds}}, \bibinfo {author} {\bibfnamefont {C.}~\bibnamefont {Ferrie}}, \bibinfo {author} {\bibfnamefont {R.}~\bibnamefont {Blatt}}, \bibinfo {author} {\bibfnamefont {P.}~\bibnamefont {Schindler}}, \bibinfo {author} {\bibfnamefont {T.}~\bibnamefont {Monz}}, \bibinfo {author} {\bibfnamefont {R.}~\bibnamefont {Kueng}},\ and\ \bibinfo {author} {\bibfnamefont {M.}~\bibnamefont {Ringbauer}},\ }\bibfield  {title} {\bibinfo {title} {Experimental single-setting quantum state tomography},\ }\href {https://doi.org/10.1103/PRXQuantum.3.040310} {\bibfield  {journal} {\bibinfo  {journal} {PRX Quantum}\ }\textbf {\bibinfo {volume} {3}},\ \bibinfo {pages} {040310} (\bibinfo {year} {2022})}\BibitemShut {NoStop}%
\bibitem [{\citenamefont {T\'oth}\ \emph {et~al.}(2010)\citenamefont {T\'oth}, \citenamefont {Wieczorek}, \citenamefont {Gross}, \citenamefont {Krischek}, \citenamefont {Schwemmer},\ and\ \citenamefont {Weinfurter}}]{PermutationallyInvariant_Toth2010}%
  \BibitemOpen
  \bibfield  {author} {\bibinfo {author} {\bibfnamefont {G.}~\bibnamefont {T\'oth}}, \bibinfo {author} {\bibfnamefont {W.}~\bibnamefont {Wieczorek}}, \bibinfo {author} {\bibfnamefont {D.}~\bibnamefont {Gross}}, \bibinfo {author} {\bibfnamefont {R.}~\bibnamefont {Krischek}}, \bibinfo {author} {\bibfnamefont {C.}~\bibnamefont {Schwemmer}},\ and\ \bibinfo {author} {\bibfnamefont {H.}~\bibnamefont {Weinfurter}},\ }\bibfield  {title} {\bibinfo {title} {Permutationally invariant quantum tomography},\ }\href {https://doi.org/10.1103/PhysRevLett.105.250403} {\bibfield  {journal} {\bibinfo  {journal} {Phys. Rev. Lett.}\ }\textbf {\bibinfo {volume} {105}},\ \bibinfo {pages} {250403} (\bibinfo {year} {2010})}\BibitemShut {NoStop}%
\bibitem [{\citenamefont {Flammia}\ and\ \citenamefont {Liu}(2011)}]{PhysRevLett.106.230501}%
  \BibitemOpen
  \bibfield  {author} {\bibinfo {author} {\bibfnamefont {S.~T.}\ \bibnamefont {Flammia}}\ and\ \bibinfo {author} {\bibfnamefont {Y.-K.}\ \bibnamefont {Liu}},\ }\bibfield  {title} {\bibinfo {title} {Direct fidelity estimation from few {P}auli measurements},\ }\href {https://doi.org/10.1103/PhysRevLett.106.230501} {\bibfield  {journal} {\bibinfo  {journal} {Phys. Rev. Lett.}\ }\textbf {\bibinfo {volume} {106}},\ \bibinfo {pages} {230501} (\bibinfo {year} {2011})}\BibitemShut {NoStop}%
\bibitem [{\citenamefont {Pallister}\ \emph {et~al.}(2018)\citenamefont {Pallister}, \citenamefont {Linden},\ and\ \citenamefont {Montanaro}}]{pallister2018optimal}%
  \BibitemOpen
  \bibfield  {author} {\bibinfo {author} {\bibfnamefont {S.}~\bibnamefont {Pallister}}, \bibinfo {author} {\bibfnamefont {N.}~\bibnamefont {Linden}},\ and\ \bibinfo {author} {\bibfnamefont {A.}~\bibnamefont {Montanaro}},\ }\bibfield  {title} {\bibinfo {title} {Optimal verification of entangled states with local measurements},\ }\href@noop {} {\bibfield  {journal} {\bibinfo  {journal} {Physical review letters}\ }\textbf {\bibinfo {volume} {120}},\ \bibinfo {pages} {170502} (\bibinfo {year} {2018})}\BibitemShut {NoStop}%
\bibitem [{\citenamefont {B{\u{a}}descu}\ \emph {et~al.}(2019)\citenamefont {B{\u{a}}descu}, \citenamefont {O'Donnell},\ and\ \citenamefont {Wright}}]{buadescu2019quantum}%
  \BibitemOpen
  \bibfield  {author} {\bibinfo {author} {\bibfnamefont {C.}~\bibnamefont {B{\u{a}}descu}}, \bibinfo {author} {\bibfnamefont {R.}~\bibnamefont {O'Donnell}},\ and\ \bibinfo {author} {\bibfnamefont {J.}~\bibnamefont {Wright}},\ }\bibfield  {title} {\bibinfo {title} {Quantum state certification},\ }in\ \href@noop {} {\emph {\bibinfo {booktitle} {Proceedings of the 51st Annual ACM SIGACT Symposium on Theory of Computing}}}\ (\bibinfo {year} {2019})\ pp.\ \bibinfo {pages} {503--514}\BibitemShut {NoStop}%
\bibitem [{\citenamefont {Huang}\ \emph {et~al.}(2020)\citenamefont {Huang}, \citenamefont {Kueng},\ and\ \citenamefont {Preskill}}]{Huang2020}%
  \BibitemOpen
  \bibfield  {author} {\bibinfo {author} {\bibfnamefont {H.-Y.}\ \bibnamefont {Huang}}, \bibinfo {author} {\bibfnamefont {R.}~\bibnamefont {Kueng}},\ and\ \bibinfo {author} {\bibfnamefont {J.}~\bibnamefont {Preskill}},\ }\bibfield  {title} {\bibinfo {title} {{Predicting many properties of a quantum system from very few measurements}},\ }\href {https://doi.org/10.1038/s41567-020-0932-7} {\bibfield  {journal} {\bibinfo  {journal} {Nat. Phys.}\ }\textbf {\bibinfo {volume} {16}},\ \bibinfo {pages} {1050} (\bibinfo {year} {2020})}\BibitemShut {NoStop}%
\bibitem [{\citenamefont {Dankert}\ \emph {et~al.}(2009)\citenamefont {Dankert}, \citenamefont {Cleve}, \citenamefont {Emerson},\ and\ \citenamefont {Livine}}]{RandomizedBenchmarking_Dankert2009}%
  \BibitemOpen
  \bibfield  {author} {\bibinfo {author} {\bibfnamefont {C.}~\bibnamefont {Dankert}}, \bibinfo {author} {\bibfnamefont {R.}~\bibnamefont {Cleve}}, \bibinfo {author} {\bibfnamefont {J.}~\bibnamefont {Emerson}},\ and\ \bibinfo {author} {\bibfnamefont {E.}~\bibnamefont {Livine}},\ }\bibfield  {title} {\bibinfo {title} {Exact and approximate unitary 2-designs and their application to fidelity estimation},\ }\href {https://doi.org/10.1103/PhysRevA.80.012304} {\bibfield  {journal} {\bibinfo  {journal} {Phys. Rev. A}\ }\textbf {\bibinfo {volume} {80}},\ \bibinfo {pages} {012304} (\bibinfo {year} {2009})}\BibitemShut {NoStop}%
\bibitem [{\citenamefont {Koltchinskii}\ and\ \citenamefont {Xia}(2015)}]{JMLR:v16:koltchinskii15a}%
  \BibitemOpen
  \bibfield  {author} {\bibinfo {author} {\bibfnamefont {V.}~\bibnamefont {Koltchinskii}}\ and\ \bibinfo {author} {\bibfnamefont {D.}~\bibnamefont {Xia}},\ }\bibfield  {title} {\bibinfo {title} {Optimal estimation of low rank density matrices},\ }\href {http://jmlr.org/papers/v16/koltchinskii15a.html} {\bibfield  {journal} {\bibinfo  {journal} {Journal of Machine Learning Research}\ }\textbf {\bibinfo {volume} {16}},\ \bibinfo {pages} {1757} (\bibinfo {year} {2015})}\BibitemShut {NoStop}%
\bibitem [{\citenamefont {Kueng}\ \emph {et~al.}(2017)\citenamefont {Kueng}, \citenamefont {Rauhut},\ and\ \citenamefont {Terstiege}}]{KUENG201788}%
  \BibitemOpen
  \bibfield  {author} {\bibinfo {author} {\bibfnamefont {R.}~\bibnamefont {Kueng}}, \bibinfo {author} {\bibfnamefont {H.}~\bibnamefont {Rauhut}},\ and\ \bibinfo {author} {\bibfnamefont {U.}~\bibnamefont {Terstiege}},\ }\bibfield  {title} {\bibinfo {title} {Low rank matrix recovery from rank one measurements},\ }\href {https://doi.org/https://doi.org/10.1016/j.acha.2015.07.007} {\bibfield  {journal} {\bibinfo  {journal} {Applied and Computational Harmonic Analysis}\ }\textbf {\bibinfo {volume} {42}},\ \bibinfo {pages} {88} (\bibinfo {year} {2017})}\BibitemShut {NoStop}%
\bibitem [{\citenamefont {Baumgratz}\ \emph {et~al.}(2013{\natexlab{b}})\citenamefont {Baumgratz}, \citenamefont {Nüßeler}, \citenamefont {Cramer},\ and\ \citenamefont {Plenio}}]{Baumgratz_2013}%
  \BibitemOpen
  \bibfield  {author} {\bibinfo {author} {\bibfnamefont {T.}~\bibnamefont {Baumgratz}}, \bibinfo {author} {\bibfnamefont {A.}~\bibnamefont {Nüßeler}}, \bibinfo {author} {\bibfnamefont {M.}~\bibnamefont {Cramer}},\ and\ \bibinfo {author} {\bibfnamefont {M.~B.}\ \bibnamefont {Plenio}},\ }\bibfield  {title} {\bibinfo {title} {A scalable maximum likelihood method for quantum state tomography},\ }\href {https://doi.org/10.1088/1367-2630/15/12/125004} {\bibfield  {journal} {\bibinfo  {journal} {New Journal of Physics}\ }\textbf {\bibinfo {volume} {15}},\ \bibinfo {pages} {125004} (\bibinfo {year} {2013}{\natexlab{b}})}\BibitemShut {NoStop}%
\bibitem [{\citenamefont {Lanyon}\ \emph {et~al.}(2017)\citenamefont {Lanyon}, \citenamefont {Maier}, \citenamefont {Holz{\"a}pfel}, \citenamefont {Baumgratz}, \citenamefont {Hempel}, \citenamefont {Jurcevic}, \citenamefont {Dhand}, \citenamefont {Buyskikh}, \citenamefont {Daley}, \citenamefont {Cramer} \emph {et~al.}}]{lanyon2017efficient}%
  \BibitemOpen
  \bibfield  {author} {\bibinfo {author} {\bibfnamefont {B.~P.}\ \bibnamefont {Lanyon}}, \bibinfo {author} {\bibfnamefont {C.}~\bibnamefont {Maier}}, \bibinfo {author} {\bibfnamefont {M.}~\bibnamefont {Holz{\"a}pfel}}, \bibinfo {author} {\bibfnamefont {T.}~\bibnamefont {Baumgratz}}, \bibinfo {author} {\bibfnamefont {C.}~\bibnamefont {Hempel}}, \bibinfo {author} {\bibfnamefont {P.}~\bibnamefont {Jurcevic}}, \bibinfo {author} {\bibfnamefont {I.}~\bibnamefont {Dhand}}, \bibinfo {author} {\bibfnamefont {A.}~\bibnamefont {Buyskikh}}, \bibinfo {author} {\bibfnamefont {A.~J.}\ \bibnamefont {Daley}}, \bibinfo {author} {\bibfnamefont {M.}~\bibnamefont {Cramer}}, \emph {et~al.},\ }\bibfield  {title} {\bibinfo {title} {Efficient tomography of a quantum many-body system},\ }\href@noop {} {\bibfield  {journal} {\bibinfo  {journal} {Nature Physics}\ }\textbf {\bibinfo {volume} {13}},\ \bibinfo {pages} {1158} (\bibinfo {year} {2017})}\BibitemShut {NoStop}%
\bibitem [{\citenamefont {Kurmapu}\ \emph {et~al.}(2023)\citenamefont {Kurmapu}, \citenamefont {Tiunova}, \citenamefont {Tiunov}, \citenamefont {Ringbauer}, \citenamefont {Maier}, \citenamefont {Blatt}, \citenamefont {Monz}, \citenamefont {Fedorov},\ and\ \citenamefont {Lvovsky}}]{PRXQuantum.4.040345}%
  \BibitemOpen
  \bibfield  {author} {\bibinfo {author} {\bibfnamefont {M.~K.}\ \bibnamefont {Kurmapu}}, \bibinfo {author} {\bibfnamefont {V.}~\bibnamefont {Tiunova}}, \bibinfo {author} {\bibfnamefont {E.}~\bibnamefont {Tiunov}}, \bibinfo {author} {\bibfnamefont {M.}~\bibnamefont {Ringbauer}}, \bibinfo {author} {\bibfnamefont {C.}~\bibnamefont {Maier}}, \bibinfo {author} {\bibfnamefont {R.}~\bibnamefont {Blatt}}, \bibinfo {author} {\bibfnamefont {T.}~\bibnamefont {Monz}}, \bibinfo {author} {\bibfnamefont {A.~K.}\ \bibnamefont {Fedorov}},\ and\ \bibinfo {author} {\bibfnamefont {A.}~\bibnamefont {Lvovsky}},\ }\bibfield  {title} {\bibinfo {title} {Reconstructing complex states of a $20$-qubit quantum simulator},\ }\href {https://doi.org/10.1103/PRXQuantum.4.040345} {\bibfield  {journal} {\bibinfo  {journal} {PRX Quantum}\ }\textbf {\bibinfo {volume} {4}},\ \bibinfo {pages} {040345} (\bibinfo {year} {2023})}\BibitemShut {NoStop}%
\bibitem [{\citenamefont {{\v{R}}eh{\'a}{\v{c}}ek}\ \emph {et~al.}(2001)\citenamefont {{\v{R}}eh{\'a}{\v{c}}ek}, \citenamefont {Hradil},\ and\ \citenamefont {Je{\v{z}}ek}}]{vrehavcek2001iterative}%
  \BibitemOpen
  \bibfield  {author} {\bibinfo {author} {\bibfnamefont {J.}~\bibnamefont {{\v{R}}eh{\'a}{\v{c}}ek}}, \bibinfo {author} {\bibfnamefont {Z.}~\bibnamefont {Hradil}},\ and\ \bibinfo {author} {\bibfnamefont {M.}~\bibnamefont {Je{\v{z}}ek}},\ }\bibfield  {title} {\bibinfo {title} {Iterative algorithm for reconstruction of entangled states},\ }\href@noop {} {\bibfield  {journal} {\bibinfo  {journal} {Physical Review A}\ }\textbf {\bibinfo {volume} {63}},\ \bibinfo {pages} {040303} (\bibinfo {year} {2001})}\BibitemShut {NoStop}%
\bibitem [{\citenamefont {Hradil}(1997)}]{QuantumMLE_Hradil1997}%
  \BibitemOpen
  \bibfield  {author} {\bibinfo {author} {\bibfnamefont {Z.}~\bibnamefont {Hradil}},\ }\bibfield  {title} {\bibinfo {title} {Quantum-state estimation},\ }\href {https://doi.org/10.1103/PhysRevA.55.R1561} {\bibfield  {journal} {\bibinfo  {journal} {Phys. Rev. A}\ }\textbf {\bibinfo {volume} {55}},\ \bibinfo {pages} {R1561} (\bibinfo {year} {1997})}\BibitemShut {NoStop}%
\bibitem [{\citenamefont {Burer}\ and\ \citenamefont {Monteiro}(2005)}]{burer2005local}%
  \BibitemOpen
  \bibfield  {author} {\bibinfo {author} {\bibfnamefont {S.}~\bibnamefont {Burer}}\ and\ \bibinfo {author} {\bibfnamefont {R.~D.}\ \bibnamefont {Monteiro}},\ }\bibfield  {title} {\bibinfo {title} {Local minima and convergence in low-rank semidefinite programming},\ }\href@noop {} {\bibfield  {journal} {\bibinfo  {journal} {Mathematical programming}\ }\textbf {\bibinfo {volume} {103}},\ \bibinfo {pages} {427} (\bibinfo {year} {2005})}\BibitemShut {NoStop}%
\bibitem [{\citenamefont {Becker}\ \emph {et~al.}(2013)\citenamefont {Becker}, \citenamefont {Cevher},\ and\ \citenamefont {Kyrillidis}}]{becker2013randomized}%
  \BibitemOpen
  \bibfield  {author} {\bibinfo {author} {\bibfnamefont {S.}~\bibnamefont {Becker}}, \bibinfo {author} {\bibfnamefont {V.}~\bibnamefont {Cevher}},\ and\ \bibinfo {author} {\bibfnamefont {A.}~\bibnamefont {Kyrillidis}},\ }\bibfield  {title} {\bibinfo {title} {Randomized low-memory singular value projection},\ }in\ \href@noop {} {\emph {\bibinfo {booktitle} {10th International Conference on Sampling Theory and Applications ({SampTA})}}}\ (\bibinfo {year} {2013})\ \bibinfo {note} {arXiv:1303.0167}\BibitemShut {NoStop}%
\bibitem [{\citenamefont {Qi}\ \emph {et~al.}(2013)\citenamefont {Qi}, \citenamefont {Hou}, \citenamefont {Li}, \citenamefont {Dong}, \citenamefont {Xiang},\ and\ \citenamefont {Guo}}]{Qi_2013}%
  \BibitemOpen
  \bibfield  {author} {\bibinfo {author} {\bibfnamefont {B.}~\bibnamefont {Qi}}, \bibinfo {author} {\bibfnamefont {Z.}~\bibnamefont {Hou}}, \bibinfo {author} {\bibfnamefont {L.}~\bibnamefont {Li}}, \bibinfo {author} {\bibfnamefont {D.}~\bibnamefont {Dong}}, \bibinfo {author} {\bibfnamefont {G.}~\bibnamefont {Xiang}},\ and\ \bibinfo {author} {\bibfnamefont {G.}~\bibnamefont {Guo}},\ }\bibfield  {title} {\bibinfo {title} {Quantum state tomography via linear regression estimation},\ }\bibfield  {journal} {\bibinfo  {journal} {Scientific Reports}\ }\textbf {\bibinfo {volume} {3}},\ \href {https://doi.org/10.1038/srep03496} {10.1038/srep03496} (\bibinfo {year} {2013})\BibitemShut {NoStop}%
\bibitem [{\citenamefont {Hou}\ \emph {et~al.}(2016)\citenamefont {Hou}, \citenamefont {Zhong}, \citenamefont {Tian}, \citenamefont {Dong}, \citenamefont {Qi}, \citenamefont {Li}, \citenamefont {Wang}, \citenamefont {Nori}, \citenamefont {Xiang}, \citenamefont {Li},\ and\ \citenamefont {Guo}}]{Hou_2016}%
  \BibitemOpen
  \bibfield  {author} {\bibinfo {author} {\bibfnamefont {Z.}~\bibnamefont {Hou}}, \bibinfo {author} {\bibfnamefont {H.-S.}\ \bibnamefont {Zhong}}, \bibinfo {author} {\bibfnamefont {Y.}~\bibnamefont {Tian}}, \bibinfo {author} {\bibfnamefont {D.}~\bibnamefont {Dong}}, \bibinfo {author} {\bibfnamefont {B.}~\bibnamefont {Qi}}, \bibinfo {author} {\bibfnamefont {L.}~\bibnamefont {Li}}, \bibinfo {author} {\bibfnamefont {Y.}~\bibnamefont {Wang}}, \bibinfo {author} {\bibfnamefont {F.}~\bibnamefont {Nori}}, \bibinfo {author} {\bibfnamefont {G.-Y.}\ \bibnamefont {Xiang}}, \bibinfo {author} {\bibfnamefont {C.-F.}\ \bibnamefont {Li}},\ and\ \bibinfo {author} {\bibfnamefont {G.-C.}\ \bibnamefont {Guo}},\ }\bibfield  {title} {\bibinfo {title} {Full reconstruction of a 14-qubit state within four hours},\ }\href {https://doi.org/10.1088/1367-2630/18/8/083036} {\bibfield  {journal} {\bibinfo  {journal} {New Journal of Physics}\ }\textbf {\bibinfo {volume} {18}},\ \bibinfo {pages} {083036} (\bibinfo {year} {2016})}\BibitemShut
  {NoStop}%
\bibitem [{\citenamefont {Smolin}\ \emph {et~al.}(2012)\citenamefont {Smolin}, \citenamefont {Gambetta},\ and\ \citenamefont {Smith}}]{PhysRevLett.108.070502}%
  \BibitemOpen
  \bibfield  {author} {\bibinfo {author} {\bibfnamefont {J.~A.}\ \bibnamefont {Smolin}}, \bibinfo {author} {\bibfnamefont {J.~M.}\ \bibnamefont {Gambetta}},\ and\ \bibinfo {author} {\bibfnamefont {G.}~\bibnamefont {Smith}},\ }\bibfield  {title} {\bibinfo {title} {Efficient method for computing the maximum-likelihood quantum state from measurements with additive gaussian noise},\ }\href {https://doi.org/10.1103/PhysRevLett.108.070502} {\bibfield  {journal} {\bibinfo  {journal} {Phys. Rev. Lett.}\ }\textbf {\bibinfo {volume} {108}},\ \bibinfo {pages} {070502} (\bibinfo {year} {2012})}\BibitemShut {NoStop}%
\bibitem [{\citenamefont {Liu}(2011{\natexlab{a}})}]{liu2011universal}%
  \BibitemOpen
  \bibfield  {author} {\bibinfo {author} {\bibfnamefont {Y.-K.}\ \bibnamefont {Liu}},\ }\bibfield  {title} {\bibinfo {title} {Universal low-rank matrix recovery from {P}auli measurements},\ }\href@noop {} {\bibfield  {journal} {\bibinfo  {journal} {Advances in Neural Information Processing Systems}\ }\textbf {\bibinfo {volume} {24}} (\bibinfo {year} {2011}{\natexlab{a}})}\BibitemShut {NoStop}%
\bibitem [{\citenamefont {James}\ \emph {et~al.}(2001)\citenamefont {James}, \citenamefont {Kwiat}, \citenamefont {Munro},\ and\ \citenamefont {White}}]{PhysRevA.64.052312}%
  \BibitemOpen
  \bibfield  {author} {\bibinfo {author} {\bibfnamefont {D.~F.~V.}\ \bibnamefont {James}}, \bibinfo {author} {\bibfnamefont {P.~G.}\ \bibnamefont {Kwiat}}, \bibinfo {author} {\bibfnamefont {W.~J.}\ \bibnamefont {Munro}},\ and\ \bibinfo {author} {\bibfnamefont {A.~G.}\ \bibnamefont {White}},\ }\bibfield  {title} {\bibinfo {title} {Measurement of qubits},\ }\href {https://doi.org/10.1103/PhysRevA.64.052312} {\bibfield  {journal} {\bibinfo  {journal} {Phys. Rev. A}\ }\textbf {\bibinfo {volume} {64}},\ \bibinfo {pages} {052312} (\bibinfo {year} {2001})}\BibitemShut {NoStop}%
\bibitem [{\citenamefont {Fiur\'a\ifmmode~\check{s}\else \v{s}\fi{}ek}\ and\ \citenamefont {Hradil}(2001)}]{PhysRevA.63.020101}%
  \BibitemOpen
  \bibfield  {author} {\bibinfo {author} {\bibfnamefont {J.}~\bibnamefont {Fiur\'a\ifmmode~\check{s}\else \v{s}\fi{}ek}}\ and\ \bibinfo {author} {\bibfnamefont {Z.~c.~v.}\ \bibnamefont {Hradil}},\ }\bibfield  {title} {\bibinfo {title} {Maximum-likelihood estimation of quantum processes},\ }\href {https://doi.org/10.1103/PhysRevA.63.020101} {\bibfield  {journal} {\bibinfo  {journal} {Phys. Rev. A}\ }\textbf {\bibinfo {volume} {63}},\ \bibinfo {pages} {020101} (\bibinfo {year} {2001})}\BibitemShut {NoStop}%
\bibitem [{\citenamefont {\ifmmode \check{R}\else \v{R}\fi{}eh\'a\ifmmode~\check{c}\else \v{c}\fi{}ek}\ \emph {et~al.}(2007)\citenamefont {\ifmmode \check{R}\else \v{R}\fi{}eh\'a\ifmmode~\check{c}\else \v{c}\fi{}ek}, \citenamefont {Hradil}, \citenamefont {Knill},\ and\ \citenamefont {Lvovsky}}]{PhysRevA.75.042108}%
  \BibitemOpen
  \bibfield  {author} {\bibinfo {author} {\bibfnamefont {J.}~\bibnamefont {\ifmmode \check{R}\else \v{R}\fi{}eh\'a\ifmmode~\check{c}\else \v{c}\fi{}ek}}, \bibinfo {author} {\bibfnamefont {Z.~c.~v.}\ \bibnamefont {Hradil}}, \bibinfo {author} {\bibfnamefont {E.}~\bibnamefont {Knill}},\ and\ \bibinfo {author} {\bibfnamefont {A.~I.}\ \bibnamefont {Lvovsky}},\ }\bibfield  {title} {\bibinfo {title} {Diluted maximum-likelihood algorithm for quantum tomography},\ }\href {https://doi.org/10.1103/PhysRevA.75.042108} {\bibfield  {journal} {\bibinfo  {journal} {Phys. Rev. A}\ }\textbf {\bibinfo {volume} {75}},\ \bibinfo {pages} {042108} (\bibinfo {year} {2007})}\BibitemShut {NoStop}%
\bibitem [{\citenamefont {Banaszek}\ \emph {et~al.}(1999)\citenamefont {Banaszek}, \citenamefont {D'Ariano}, \citenamefont {Paris},\ and\ \citenamefont {Sacchi}}]{PhysRevA.61.010304}%
  \BibitemOpen
  \bibfield  {author} {\bibinfo {author} {\bibfnamefont {K.}~\bibnamefont {Banaszek}}, \bibinfo {author} {\bibfnamefont {G.~M.}\ \bibnamefont {D'Ariano}}, \bibinfo {author} {\bibfnamefont {M.~G.~A.}\ \bibnamefont {Paris}},\ and\ \bibinfo {author} {\bibfnamefont {M.~F.}\ \bibnamefont {Sacchi}},\ }\bibfield  {title} {\bibinfo {title} {Maximum-likelihood estimation of the density matrix},\ }\href {https://doi.org/10.1103/PhysRevA.61.010304} {\bibfield  {journal} {\bibinfo  {journal} {Phys. Rev. A}\ }\textbf {\bibinfo {volume} {61}},\ \bibinfo {pages} {010304} (\bibinfo {year} {1999})}\BibitemShut {NoStop}%
\bibitem [{\citenamefont {Acharya}\ \emph {et~al.}(2019)\citenamefont {Acharya}, \citenamefont {Kypraios},\ and\ \citenamefont {Gu{\c{t}}{\u{a}}}}]{Acharya_2019}%
  \BibitemOpen
  \bibfield  {author} {\bibinfo {author} {\bibfnamefont {A.}~\bibnamefont {Acharya}}, \bibinfo {author} {\bibfnamefont {T.}~\bibnamefont {Kypraios}},\ and\ \bibinfo {author} {\bibfnamefont {M.}~\bibnamefont {Gu{\c{t}}{\u{a}}}},\ }\bibfield  {title} {\bibinfo {title} {A comparative study of estimation methods in quantum tomography},\ }\href {https://doi.org/10.1088/1751-8121/ab1958} {\bibfield  {journal} {\bibinfo  {journal} {Journal of Physics A: Mathematical and Theoretical}\ }\textbf {\bibinfo {volume} {52}},\ \bibinfo {pages} {234001} (\bibinfo {year} {2019})}\BibitemShut {NoStop}%
\bibitem [{\citenamefont {Gross}\ \emph {et~al.}(2010)\citenamefont {Gross}, \citenamefont {Liu}, \citenamefont {Flammia}, \citenamefont {Becker},\ and\ \citenamefont {Eisert}}]{PhysRevLett.105.150401}%
  \BibitemOpen
  \bibfield  {author} {\bibinfo {author} {\bibfnamefont {D.}~\bibnamefont {Gross}}, \bibinfo {author} {\bibfnamefont {Y.-K.}\ \bibnamefont {Liu}}, \bibinfo {author} {\bibfnamefont {S.~T.}\ \bibnamefont {Flammia}}, \bibinfo {author} {\bibfnamefont {S.}~\bibnamefont {Becker}},\ and\ \bibinfo {author} {\bibfnamefont {J.}~\bibnamefont {Eisert}},\ }\bibfield  {title} {\bibinfo {title} {Quantum state tomography via compressed sensing},\ }\href {https://doi.org/10.1103/PhysRevLett.105.150401} {\bibfield  {journal} {\bibinfo  {journal} {Phys. Rev. Lett.}\ }\textbf {\bibinfo {volume} {105}},\ \bibinfo {pages} {150401} (\bibinfo {year} {2010})}\BibitemShut {NoStop}%
\bibitem [{\citenamefont {Donoho}(2006)}]{1614066}%
  \BibitemOpen
  \bibfield  {author} {\bibinfo {author} {\bibfnamefont {D.}~\bibnamefont {Donoho}},\ }\bibfield  {title} {\bibinfo {title} {Compressed sensing},\ }\href {https://doi.org/10.1109/TIT.2006.871582} {\bibfield  {journal} {\bibinfo  {journal} {IEEE Transactions on Information Theory}\ }\textbf {\bibinfo {volume} {52}},\ \bibinfo {pages} {1289} (\bibinfo {year} {2006})}\BibitemShut {NoStop}%
\bibitem [{\citenamefont {Cai}\ \emph {et~al.}(2010)\citenamefont {Cai}, \citenamefont {Cand\`{e}s},\ and\ \citenamefont {Shen}}]{doi:10.1137/080738970}%
  \BibitemOpen
  \bibfield  {author} {\bibinfo {author} {\bibfnamefont {J.-F.}\ \bibnamefont {Cai}}, \bibinfo {author} {\bibfnamefont {E.~J.}\ \bibnamefont {Cand\`{e}s}},\ and\ \bibinfo {author} {\bibfnamefont {Z.}~\bibnamefont {Shen}},\ }\bibfield  {title} {\bibinfo {title} {A singular value thresholding algorithm for matrix completion},\ }\href {https://doi.org/10.1137/080738970} {\bibfield  {journal} {\bibinfo  {journal} {SIAM Journal on Optimization}\ }\textbf {\bibinfo {volume} {20}},\ \bibinfo {pages} {1956} (\bibinfo {year} {2010})}\BibitemShut {NoStop}%
\bibitem [{\citenamefont {Liu}(2011{\natexlab{b}})}]{31416}%
  \BibitemOpen
  \bibfield  {author} {\bibinfo {author} {\bibfnamefont {Y.-K.}\ \bibnamefont {Liu}},\ }\bibfield  {title} {\bibinfo {title} {Universal low-rank matrix recovery from {P}auli measurements},\ }in\ \href@noop {} {\emph {\bibinfo {booktitle} {Advances in Neural Information Processing Systems (NIPS)}}},\ \bibinfo {series and number} {\bibinfo {number} {24}}\ (\bibinfo {address} {La Jolla, CA},\ \bibinfo {year} {2011})\BibitemShut {NoStop}%
\bibitem [{\citenamefont {Lidiak}\ \emph {et~al.}(2022)\citenamefont {Lidiak}, \citenamefont {Jameson}, \citenamefont {Qin}, \citenamefont {Tang}, \citenamefont {Wakin}, \citenamefont {Zhu},\ and\ \citenamefont {Gong}}]{lidiak2022quantum}%
  \BibitemOpen
  \bibfield  {author} {\bibinfo {author} {\bibfnamefont {A.}~\bibnamefont {Lidiak}}, \bibinfo {author} {\bibfnamefont {C.}~\bibnamefont {Jameson}}, \bibinfo {author} {\bibfnamefont {Z.}~\bibnamefont {Qin}}, \bibinfo {author} {\bibfnamefont {G.}~\bibnamefont {Tang}}, \bibinfo {author} {\bibfnamefont {M.~B.}\ \bibnamefont {Wakin}}, \bibinfo {author} {\bibfnamefont {Z.}~\bibnamefont {Zhu}},\ and\ \bibinfo {author} {\bibfnamefont {Z.}~\bibnamefont {Gong}},\ }\bibfield  {title} {\bibinfo {title} {Quantum state tomography with tensor train cross approximation},\ }\href@noop {} {\bibfield  {journal} {\bibinfo  {journal} {arXiv preprint arXiv:2207.06397}\ } (\bibinfo {year} {2022})}\BibitemShut {NoStop}%
\bibitem [{\citenamefont {Qin}\ \emph {et~al.}(2024{\natexlab{a}})\citenamefont {Qin}, \citenamefont {Jameson}, \citenamefont {Gong}, \citenamefont {Wakin},\ and\ \citenamefont {Zhu}}]{Qin_2024}%
  \BibitemOpen
  \bibfield  {author} {\bibinfo {author} {\bibfnamefont {Z.}~\bibnamefont {Qin}}, \bibinfo {author} {\bibfnamefont {C.}~\bibnamefont {Jameson}}, \bibinfo {author} {\bibfnamefont {Z.}~\bibnamefont {Gong}}, \bibinfo {author} {\bibfnamefont {M.~B.}\ \bibnamefont {Wakin}},\ and\ \bibinfo {author} {\bibfnamefont {Z.}~\bibnamefont {Zhu}},\ }\bibfield  {title} {\bibinfo {title} {Quantum state tomography for matrix product density operators},\ }\href {https://doi.org/10.1109/tit.2024.3360951} {\bibfield  {journal} {\bibinfo  {journal} {IEEE Transactions on Information Theory}\ }\textbf {\bibinfo {volume} {70}},\ \bibinfo {pages} {5030–5056} (\bibinfo {year} {2024}{\natexlab{a}})}\BibitemShut {NoStop}%
\bibitem [{\citenamefont {Shang}\ \emph {et~al.}(2017)\citenamefont {Shang}, \citenamefont {Zhang},\ and\ \citenamefont {Ng}}]{PhysRevA.95.062336}%
  \BibitemOpen
  \bibfield  {author} {\bibinfo {author} {\bibfnamefont {J.}~\bibnamefont {Shang}}, \bibinfo {author} {\bibfnamefont {Z.}~\bibnamefont {Zhang}},\ and\ \bibinfo {author} {\bibfnamefont {H.~K.}\ \bibnamefont {Ng}},\ }\bibfield  {title} {\bibinfo {title} {Superfast maximum-likelihood reconstruction for quantum tomography},\ }\href {https://doi.org/10.1103/PhysRevA.95.062336} {\bibfield  {journal} {\bibinfo  {journal} {Phys. Rev. A}\ }\textbf {\bibinfo {volume} {95}},\ \bibinfo {pages} {062336} (\bibinfo {year} {2017})}\BibitemShut {NoStop}%
\bibitem [{\citenamefont {Kyrillidis}\ \emph {et~al.}(2018)\citenamefont {Kyrillidis}, \citenamefont {Kalev}, \citenamefont {Park}, \citenamefont {Bhojanapalli}, \citenamefont {Caramanis},\ and\ \citenamefont {Sanghavi}}]{kyrillidis2017provable}%
  \BibitemOpen
  \bibfield  {author} {\bibinfo {author} {\bibfnamefont {A.}~\bibnamefont {Kyrillidis}}, \bibinfo {author} {\bibfnamefont {A.}~\bibnamefont {Kalev}}, \bibinfo {author} {\bibfnamefont {D.}~\bibnamefont {Park}}, \bibinfo {author} {\bibfnamefont {S.}~\bibnamefont {Bhojanapalli}}, \bibinfo {author} {\bibfnamefont {C.}~\bibnamefont {Caramanis}},\ and\ \bibinfo {author} {\bibfnamefont {S.}~\bibnamefont {Sanghavi}},\ }\bibfield  {title} {\bibinfo {title} {Provable compressed sensing quantum state tomography via non-convex methods},\ }\href {https://doi.org/10.1038/s41534-018-0080-4} {\bibfield  {journal} {\bibinfo  {journal} {npj Quantum Information}\ }\textbf {\bibinfo {volume} {4}} (\bibinfo {year} {2018})}\BibitemShut {NoStop}%
\bibitem [{\citenamefont {Kim}\ \emph {et~al.}(2023)\citenamefont {Kim}, \citenamefont {Kollias}, \citenamefont {Kalev}, \citenamefont {Wei},\ and\ \citenamefont {Kyrillidis}}]{photonics10020116}%
  \BibitemOpen
  \bibfield  {author} {\bibinfo {author} {\bibfnamefont {J.~L.}\ \bibnamefont {Kim}}, \bibinfo {author} {\bibfnamefont {G.}~\bibnamefont {Kollias}}, \bibinfo {author} {\bibfnamefont {A.}~\bibnamefont {Kalev}}, \bibinfo {author} {\bibfnamefont {K.~X.}\ \bibnamefont {Wei}},\ and\ \bibinfo {author} {\bibfnamefont {A.}~\bibnamefont {Kyrillidis}},\ }\bibfield  {title} {\bibinfo {title} {Fast quantum state reconstruction via accelerated non-convex programming},\ }\bibfield  {journal} {\bibinfo  {journal} {Photonics}\ }\textbf {\bibinfo {volume} {10}},\ \href {https://doi.org/10.3390/photonics10020116} {10.3390/photonics10020116} (\bibinfo {year} {2023})\BibitemShut {NoStop}%
\bibitem [{\citenamefont {Candès}\ \emph {et~al.}(2015)\citenamefont {Candès}, \citenamefont {Li},\ and\ \citenamefont {Soltanolkotabi}}]{7029630}%
  \BibitemOpen
  \bibfield  {author} {\bibinfo {author} {\bibfnamefont {E.~J.}\ \bibnamefont {Candès}}, \bibinfo {author} {\bibfnamefont {X.}~\bibnamefont {Li}},\ and\ \bibinfo {author} {\bibfnamefont {M.}~\bibnamefont {Soltanolkotabi}},\ }\bibfield  {title} {\bibinfo {title} {Phase retrieval via {W}irtinger flow: Theory and algorithms},\ }\href {https://doi.org/10.1109/TIT.2015.2399924} {\bibfield  {journal} {\bibinfo  {journal} {IEEE Transactions on Information Theory}\ }\textbf {\bibinfo {volume} {61}},\ \bibinfo {pages} {1985} (\bibinfo {year} {2015})}\BibitemShut {NoStop}%
\bibitem [{\citenamefont {Burer}\ and\ \citenamefont {Monteiro}(2003)}]{Burer2003ANP}%
  \BibitemOpen
  \bibfield  {author} {\bibinfo {author} {\bibfnamefont {S.}~\bibnamefont {Burer}}\ and\ \bibinfo {author} {\bibfnamefont {R.~D.~C.}\ \bibnamefont {Monteiro}},\ }\bibfield  {title} {\bibinfo {title} {A nonlinear programming algorithm for solving semidefinite programs via low-rank factorization},\ }\href {https://api.semanticscholar.org/CorpusID:7691228} {\bibfield  {journal} {\bibinfo  {journal} {Mathematical Programming}\ }\textbf {\bibinfo {volume} {95}},\ \bibinfo {pages} {329} (\bibinfo {year} {2003})}\BibitemShut {NoStop}%
\bibitem [{\citenamefont {Wang}\ \emph {et~al.}(2024)\citenamefont {Wang}, \citenamefont {Liu}, \citenamefont {Cheng}, \citenamefont {Li},\ and\ \citenamefont {Chen}}]{PhysRevResearch.6.033034}%
  \BibitemOpen
  \bibfield  {author} {\bibinfo {author} {\bibfnamefont {Y.}~\bibnamefont {Wang}}, \bibinfo {author} {\bibfnamefont {L.}~\bibnamefont {Liu}}, \bibinfo {author} {\bibfnamefont {S.}~\bibnamefont {Cheng}}, \bibinfo {author} {\bibfnamefont {L.}~\bibnamefont {Li}},\ and\ \bibinfo {author} {\bibfnamefont {J.}~\bibnamefont {Chen}},\ }\bibfield  {title} {\bibinfo {title} {Efficient factored gradient descent algorithm for quantum state tomography},\ }\href {https://doi.org/10.1103/PhysRevResearch.6.033034} {\bibfield  {journal} {\bibinfo  {journal} {Phys. Rev. Res.}\ }\textbf {\bibinfo {volume} {6}},\ \bibinfo {pages} {033034} (\bibinfo {year} {2024})}\BibitemShut {NoStop}%
\bibitem [{\citenamefont {Hsu}\ \emph {et~al.}(2024)\citenamefont {Hsu}, \citenamefont {Kuo}, \citenamefont {Yu}, \citenamefont {Cai},\ and\ \citenamefont {Hsieh}}]{RiemannianPRL2024}%
  \BibitemOpen
  \bibfield  {author} {\bibinfo {author} {\bibfnamefont {M.-C.}\ \bibnamefont {Hsu}}, \bibinfo {author} {\bibfnamefont {E.-J.}\ \bibnamefont {Kuo}}, \bibinfo {author} {\bibfnamefont {W.-H.}\ \bibnamefont {Yu}}, \bibinfo {author} {\bibfnamefont {J.-F.}\ \bibnamefont {Cai}},\ and\ \bibinfo {author} {\bibfnamefont {M.-H.}\ \bibnamefont {Hsieh}},\ }\bibfield  {title} {\bibinfo {title} {Quantum state tomography via nonconvex {R}iemannian gradient descent},\ }\href {https://doi.org/10.1103/PhysRevLett.132.240804} {\bibfield  {journal} {\bibinfo  {journal} {Phys. Rev. Lett.}\ }\textbf {\bibinfo {volume} {132}},\ \bibinfo {pages} {240804} (\bibinfo {year} {2024})}\BibitemShut {NoStop}%
\bibitem [{\citenamefont {Gaikwad}\ \emph {et~al.}(2025)\citenamefont {Gaikwad}, \citenamefont {Torres}, \citenamefont {Ahmed},\ and\ \citenamefont {Kockum}}]{gaikwad2025gradient}%
  \BibitemOpen
  \bibfield  {author} {\bibinfo {author} {\bibfnamefont {A.}~\bibnamefont {Gaikwad}}, \bibinfo {author} {\bibfnamefont {M.~S.}\ \bibnamefont {Torres}}, \bibinfo {author} {\bibfnamefont {S.}~\bibnamefont {Ahmed}},\ and\ \bibinfo {author} {\bibfnamefont {A.~F.}\ \bibnamefont {Kockum}},\ }\bibfield  {title} {\bibinfo {title} {Gradient-descent methods for fast quantum state tomography},\ }\href@noop {} {\bibfield  {journal} {\bibinfo  {journal} {Quantum Science and Technology}\ } (\bibinfo {year} {2025})}\BibitemShut {NoStop}%
\bibitem [{\citenamefont {Qin}\ \emph {et~al.}(2024{\natexlab{b}})\citenamefont {Qin}, \citenamefont {Jameson}, \citenamefont {Gong}, \citenamefont {Wakin},\ and\ \citenamefont {Zhu}}]{qin2024optimal}%
  \BibitemOpen
  \bibfield  {author} {\bibinfo {author} {\bibfnamefont {Z.}~\bibnamefont {Qin}}, \bibinfo {author} {\bibfnamefont {C.}~\bibnamefont {Jameson}}, \bibinfo {author} {\bibfnamefont {Z.}~\bibnamefont {Gong}}, \bibinfo {author} {\bibfnamefont {M.~B.}\ \bibnamefont {Wakin}},\ and\ \bibinfo {author} {\bibfnamefont {Z.}~\bibnamefont {Zhu}},\ }\bibfield  {title} {\bibinfo {title} {Optimal allocation of {P}auli measurements for low-rank quantum state tomography},\ }\href@noop {} {\bibfield  {journal} {\bibinfo  {journal} {arXiv preprint arXiv:2411.04452}\ } (\bibinfo {year} {2024}{\natexlab{b}})}\BibitemShut {NoStop}%
\bibitem [{\citenamefont {\ifmmode \check{R}\else \v{R}\fi{}eh\'a\ifmmode~\check{c}\else \v{c}\fi{}ek}\ \emph {et~al.}(2004)\citenamefont {\ifmmode \check{R}\else \v{R}\fi{}eh\'a\ifmmode~\check{c}\else \v{c}\fi{}ek}, \citenamefont {Englert},\ and\ \citenamefont {Kaszlikowski}}]{PhysRevA.70.052321}%
  \BibitemOpen
  \bibfield  {author} {\bibinfo {author} {\bibfnamefont {J.}~\bibnamefont {\ifmmode \check{R}\else \v{R}\fi{}eh\'a\ifmmode~\check{c}\else \v{c}\fi{}ek}}, \bibinfo {author} {\bibfnamefont {B.-G.}\ \bibnamefont {Englert}},\ and\ \bibinfo {author} {\bibfnamefont {D.}~\bibnamefont {Kaszlikowski}},\ }\bibfield  {title} {\bibinfo {title} {Minimal qubit tomography},\ }\href {https://doi.org/10.1103/PhysRevA.70.052321} {\bibfield  {journal} {\bibinfo  {journal} {Phys. Rev. A}\ }\textbf {\bibinfo {volume} {70}},\ \bibinfo {pages} {052321} (\bibinfo {year} {2004})}\BibitemShut {NoStop}%
\bibitem [{\citenamefont {Boumal}\ \emph {et~al.}(2016)\citenamefont {Boumal}, \citenamefont {Voroninski},\ and\ \citenamefont {Bandeira}}]{boumal2016non}%
  \BibitemOpen
  \bibfield  {author} {\bibinfo {author} {\bibfnamefont {N.}~\bibnamefont {Boumal}}, \bibinfo {author} {\bibfnamefont {V.}~\bibnamefont {Voroninski}},\ and\ \bibinfo {author} {\bibfnamefont {A.}~\bibnamefont {Bandeira}},\ }\bibfield  {title} {\bibinfo {title} {The non-convex {Burer-Monteiro} approach works on smooth semidefinite programs},\ }\href@noop {} {\bibfield  {journal} {\bibinfo  {journal} {Advances in Neural Information Processing Systems}\ }\textbf {\bibinfo {volume} {29}} (\bibinfo {year} {2016})}\BibitemShut {NoStop}%
\bibitem [{\citenamefont {Waldspurger}\ and\ \citenamefont {Waters}(2020)}]{WaldspurgerBM}%
  \BibitemOpen
  \bibfield  {author} {\bibinfo {author} {\bibfnamefont {I.}~\bibnamefont {Waldspurger}}\ and\ \bibinfo {author} {\bibfnamefont {A.}~\bibnamefont {Waters}},\ }\bibfield  {title} {\bibinfo {title} {Rank optimality for the {Burer--Monteiro} factorization},\ }\href {https://doi.org/10.1137/19M1255318} {\bibfield  {journal} {\bibinfo  {journal} {SIAM Journal on Optimization}\ }\textbf {\bibinfo {volume} {30}},\ \bibinfo {pages} {2577} (\bibinfo {year} {2020})},\ \Eprint {https://arxiv.org/abs/https://doi.org/10.1137/19M1255318} {https://doi.org/10.1137/19M1255318} \BibitemShut {NoStop}%
\bibitem [{\citenamefont {Aravkin}\ \emph {et~al.}(2014)\citenamefont {Aravkin}, \citenamefont {Kumar}, \citenamefont {Mansour}, \citenamefont {Recht},\ and\ \citenamefont {Herrmann}}]{AravkinBM}%
  \BibitemOpen
  \bibfield  {author} {\bibinfo {author} {\bibfnamefont {A.}~\bibnamefont {Aravkin}}, \bibinfo {author} {\bibfnamefont {R.}~\bibnamefont {Kumar}}, \bibinfo {author} {\bibfnamefont {H.}~\bibnamefont {Mansour}}, \bibinfo {author} {\bibfnamefont {B.}~\bibnamefont {Recht}},\ and\ \bibinfo {author} {\bibfnamefont {F.~J.}\ \bibnamefont {Herrmann}},\ }\bibfield  {title} {\bibinfo {title} {Fast methods for denoising matrix completion formulations, with applications to robust seismic data interpolation},\ }\href {https://doi.org/10.1137/130919210} {\bibfield  {journal} {\bibinfo  {journal} {SIAM Journal on Scientific Computing}\ }\textbf {\bibinfo {volume} {36}},\ \bibinfo {pages} {S237} (\bibinfo {year} {2014})},\ \Eprint {https://arxiv.org/abs/https://doi.org/10.1137/130919210} {https://doi.org/10.1137/130919210} \BibitemShut {NoStop}%
\bibitem [{\citenamefont {Cifuentes}(2021)}]{cifuentes2021burer}%
  \BibitemOpen
  \bibfield  {author} {\bibinfo {author} {\bibfnamefont {D.}~\bibnamefont {Cifuentes}},\ }\bibfield  {title} {\bibinfo {title} {On the {Burer--Monteiro method} for general semidefinite programs},\ }\href@noop {} {\bibfield  {journal} {\bibinfo  {journal} {Optimization Letters}\ }\textbf {\bibinfo {volume} {15}},\ \bibinfo {pages} {2299} (\bibinfo {year} {2021})}\BibitemShut {NoStop}%
\bibitem [{\citenamefont {Levin}\ \emph {et~al.}(2025)\citenamefont {Levin}, \citenamefont {Kileel},\ and\ \citenamefont {Boumal}}]{levin2025effect}%
  \BibitemOpen
  \bibfield  {author} {\bibinfo {author} {\bibfnamefont {E.}~\bibnamefont {Levin}}, \bibinfo {author} {\bibfnamefont {J.}~\bibnamefont {Kileel}},\ and\ \bibinfo {author} {\bibfnamefont {N.}~\bibnamefont {Boumal}},\ }\bibfield  {title} {\bibinfo {title} {The effect of smooth parametrizations on nonconvex optimization landscapes},\ }\href@noop {} {\bibfield  {journal} {\bibinfo  {journal} {Mathematical Programming}\ }\textbf {\bibinfo {volume} {209}},\ \bibinfo {pages} {63} (\bibinfo {year} {2025})}\BibitemShut {NoStop}%
\bibitem [{\citenamefont {Boyd}\ and\ \citenamefont {Vandenberghe}(2004)}]{boyd2004convex}%
  \BibitemOpen
  \bibfield  {author} {\bibinfo {author} {\bibfnamefont {S.~P.}\ \bibnamefont {Boyd}}\ and\ \bibinfo {author} {\bibfnamefont {L.}~\bibnamefont {Vandenberghe}},\ }\href@noop {} {\emph {\bibinfo {title} {Convex optimization}}}\ (\bibinfo  {publisher} {Cambridge university press},\ \bibinfo {year} {2004})\BibitemShut {NoStop}%
\bibitem [{\citenamefont {Nocedal}\ and\ \citenamefont {Wright}(2006)}]{nocedal2006numerical}%
  \BibitemOpen
  \bibfield  {author} {\bibinfo {author} {\bibfnamefont {J.}~\bibnamefont {Nocedal}}\ and\ \bibinfo {author} {\bibfnamefont {S.}~\bibnamefont {Wright}},\ }\href@noop {} {\emph {\bibinfo {title} {Numerical optimization}}},\ Springer Series in Operations Research and Financial Engineering\ (\bibinfo  {publisher} {Springer},\ \bibinfo {year} {2006})\BibitemShut {NoStop}%
\bibitem [{\citenamefont {Nesterov}(1983)}]{Nesterov83}%
  \BibitemOpen
  \bibfield  {author} {\bibinfo {author} {\bibfnamefont {Y.}~\bibnamefont {Nesterov}},\ }\bibfield  {title} {\bibinfo {title} {A method for unconstrained convex minimization problem with the rate of convergence $\mathcal{O}(1/k^2)$},\ }\href@noop {} {\bibfield  {journal} {\bibinfo  {journal} {Doklady AN SSSR, {\em translated as} Soviet Math. Docl.}\ }\textbf {\bibinfo {volume} {269}},\ \bibinfo {pages} {543} (\bibinfo {year} {1983})}\BibitemShut {NoStop}%
\bibitem [{\citenamefont {Beck}\ and\ \citenamefont {Teboulle}(2009)}]{FISTA}%
  \BibitemOpen
  \bibfield  {author} {\bibinfo {author} {\bibfnamefont {A.}~\bibnamefont {Beck}}\ and\ \bibinfo {author} {\bibfnamefont {M.}~\bibnamefont {Teboulle}},\ }\bibfield  {title} {\bibinfo {title} {A fast iterative shrinkage-thresholding algorithm for linear inverse problems},\ }\href@noop {} {\bibfield  {journal} {\bibinfo  {journal} {SIAM J. on Imaging Sci.}\ }\textbf {\bibinfo {volume} {2}},\ \bibinfo {pages} {183} (\bibinfo {year} {2009})}\BibitemShut {NoStop}%
\bibitem [{\citenamefont {Beck}(2017)}]{BeckBook2017}%
  \BibitemOpen
  \bibfield  {author} {\bibinfo {author} {\bibfnamefont {A.}~\bibnamefont {Beck}},\ }\href@noop {} {\emph {\bibinfo {title} {First-Order Methods in Optimization}}}\ (\bibinfo  {publisher} {MOS-SIAM Series on Optimization},\ \bibinfo {year} {2017})\BibitemShut {NoStop}%
\bibitem [{Note1()}]{Note1}%
  \BibitemOpen
  \bibinfo {note} {\protect \url {https://github.com/stephenbeckr/convex-optimization-class/blob/main/utilities/firstOrderMethods.py}}\BibitemShut {NoStop}%
\bibitem [{Note2()}]{Note2}%
  \BibitemOpen
  \bibinfo {note} {Specifically, $u_1^\star = \protect \frac {-1}{\lambda }\left ( \protect \frac {d}{c-d} - \protect \frac {b}{b-a} \right )$ and $u_2^\star =\protect \frac {-1}{\lambda }\left ( \protect \frac {a}{b-a} - \protect \frac {c}{c-d}\right )$.}\BibitemShut {Stop}%
\bibitem [{\citenamefont {Lee}\ \emph {et~al.}(2014)\citenamefont {Lee}, \citenamefont {Sun},\ and\ \citenamefont {Saunders}}]{ProxNewton_LeeSunSaunders}%
  \BibitemOpen
  \bibfield  {author} {\bibinfo {author} {\bibfnamefont {J.~D.}\ \bibnamefont {Lee}}, \bibinfo {author} {\bibfnamefont {Y.}~\bibnamefont {Sun}},\ and\ \bibinfo {author} {\bibfnamefont {M.~A.}\ \bibnamefont {Saunders}},\ }\bibfield  {title} {\bibinfo {title} {Proximal {N}ewton-type methods for minimizing composite functions},\ }\href {https://doi.org/10.1137/130921428} {\bibfield  {journal} {\bibinfo  {journal} {SIAM Journal on Optimization}\ }\textbf {\bibinfo {volume} {24}},\ \bibinfo {pages} {1420} (\bibinfo {year} {2014})}\BibitemShut {NoStop}%
\bibitem [{Note3()}]{Note3}%
  \BibitemOpen
  \bibinfo {note} {All the gradient descent solvers need a stepsize, unlike Newton's method. For each variant of gradient descent, we separately perform a grid search over $100$ stepsizes ranging between $10^{-5}$ and $10^2$ and use the best one.}\BibitemShut {Stop}%
\bibitem [{\citenamefont {Tsai}\ \emph {et~al.}(2024)\citenamefont {Tsai}, \citenamefont {Cheng},\ and\ \citenamefont {Li}}]{pmlr-v238-tsai24a}%
  \BibitemOpen
  \bibfield  {author} {\bibinfo {author} {\bibfnamefont {C.-E.}\ \bibnamefont {Tsai}}, \bibinfo {author} {\bibfnamefont {H.-C.}\ \bibnamefont {Cheng}},\ and\ \bibinfo {author} {\bibfnamefont {Y.-H.}\ \bibnamefont {Li}},\ }\bibfield  {title} {\bibinfo {title} {Fast minimization of expected logarithmic loss via stochastic dual averaging},\ }in\ \href {https://proceedings.mlr.press/v238/tsai24a.html} {\emph {\bibinfo {booktitle} {Proceedings of The 27th International Conference on Artificial Intelligence and Statistics}}},\ \bibinfo {series} {Proceedings of Machine Learning Research}, Vol.\ \bibinfo {volume} {238},\ \bibinfo {editor} {edited by\ \bibinfo {editor} {\bibfnamefont {S.}~\bibnamefont {Dasgupta}}, \bibinfo {editor} {\bibfnamefont {S.}~\bibnamefont {Mandt}},\ and\ \bibinfo {editor} {\bibfnamefont {Y.}~\bibnamefont {Li}}}\ (\bibinfo  {publisher} {PMLR},\ \bibinfo {year} {2024})\ pp.\ \bibinfo {pages} {2908--2916}\BibitemShut {NoStop}%
\bibitem [{\citenamefont {Garber}\ and\ \citenamefont {Kaplan}(2023)}]{10.1287/moor.2022.1332}%
  \BibitemOpen
  \bibfield  {author} {\bibinfo {author} {\bibfnamefont {D.}~\bibnamefont {Garber}}\ and\ \bibinfo {author} {\bibfnamefont {A.}~\bibnamefont {Kaplan}},\ }\bibfield  {title} {\bibinfo {title} {On the efficient implementation of the matrix exponentiated gradient algorithm for low-rank matrix optimization},\ }\href {https://doi.org/10.1287/moor.2022.1332} {\bibfield  {journal} {\bibinfo  {journal} {Math. Oper. Res.}\ }\textbf {\bibinfo {volume} {48}},\ \bibinfo {pages} {2094–2128} (\bibinfo {year} {2023})}\BibitemShut {NoStop}%
\bibitem [{\citenamefont {Li}\ \emph {et~al.}(2018)\citenamefont {Li}, \citenamefont {Riofr{\'\i}o},\ and\ \citenamefont {Cevher}}]{li2018general}%
  \BibitemOpen
  \bibfield  {author} {\bibinfo {author} {\bibfnamefont {Y.-H.}\ \bibnamefont {Li}}, \bibinfo {author} {\bibfnamefont {C.~A.}\ \bibnamefont {Riofr{\'\i}o}},\ and\ \bibinfo {author} {\bibfnamefont {V.}~\bibnamefont {Cevher}},\ }\bibfield  {title} {\bibinfo {title} {A general convergence result for mirror descent with {A}rmijo line search},\ }\href@noop {} {\bibfield  {journal} {\bibinfo  {journal} {arXiv preprint arXiv:1805.12232}\ } (\bibinfo {year} {2018})}\BibitemShut {NoStop}%
\bibitem [{\citenamefont {Tran-Dinh}\ \emph {et~al.}(2015)\citenamefont {Tran-Dinh}, \citenamefont {Kyrillidis},\ and\ \citenamefont {Cevher}}]{JMLR:v16:trandihn15a}%
  \BibitemOpen
  \bibfield  {author} {\bibinfo {author} {\bibfnamefont {Q.}~\bibnamefont {Tran-Dinh}}, \bibinfo {author} {\bibfnamefont {A.}~\bibnamefont {Kyrillidis}},\ and\ \bibinfo {author} {\bibfnamefont {V.}~\bibnamefont {Cevher}},\ }\bibfield  {title} {\bibinfo {title} {Composite self-concordant minimization},\ }\href {http://jmlr.org/papers/v16/trandihn15a.html} {\bibfield  {journal} {\bibinfo  {journal} {Journal of Machine Learning Research}\ }\textbf {\bibinfo {volume} {16}},\ \bibinfo {pages} {371} (\bibinfo {year} {2015})}\BibitemShut {NoStop}%
\bibitem [{\citenamefont {Odor}\ \emph {et~al.}(2016)\citenamefont {Odor}, \citenamefont {Li}, \citenamefont {Yurtsever}, \citenamefont {Hsieh}, \citenamefont {Tran-Dinh}, \citenamefont {El~Halabi},\ and\ \citenamefont {Cevher}}]{odor2016frank}%
  \BibitemOpen
  \bibfield  {author} {\bibinfo {author} {\bibfnamefont {G.}~\bibnamefont {Odor}}, \bibinfo {author} {\bibfnamefont {Y.-H.}\ \bibnamefont {Li}}, \bibinfo {author} {\bibfnamefont {A.}~\bibnamefont {Yurtsever}}, \bibinfo {author} {\bibfnamefont {Y.-P.}\ \bibnamefont {Hsieh}}, \bibinfo {author} {\bibfnamefont {Q.}~\bibnamefont {Tran-Dinh}}, \bibinfo {author} {\bibfnamefont {M.}~\bibnamefont {El~Halabi}},\ and\ \bibinfo {author} {\bibfnamefont {V.}~\bibnamefont {Cevher}},\ }\bibfield  {title} {\bibinfo {title} {Frank-wolfe works for non-lipschitz continuous gradient objectives: Scalable {P}oisson phase retrieval},\ }in\ \href@noop {} {\emph {\bibinfo {booktitle} {2016 IEEE International Conference on Acoustics, Speech and Signal Processing (ICASSP)}}}\ (\bibinfo {organization} {Ieee},\ \bibinfo {year} {2016})\ pp.\ \bibinfo {pages} {6230--6234}\BibitemShut {NoStop}%
\bibitem [{\citenamefont {Lvovsky}(2004)}]{Lvovsky_2004}%
  \BibitemOpen
  \bibfield  {author} {\bibinfo {author} {\bibfnamefont {A.~I.}\ \bibnamefont {Lvovsky}},\ }\bibfield  {title} {\bibinfo {title} {Iterative maximum-likelihood reconstruction in quantum homodyne tomography},\ }\href {https://doi.org/10.1088/1464-4266/6/6/014} {\bibfield  {journal} {\bibinfo  {journal} {Journal of Optics B: Quantum and Semiclassical Optics}\ }\textbf {\bibinfo {volume} {6}},\ \bibinfo {pages} {S556–S559} (\bibinfo {year} {2004})}\BibitemShut {NoStop}%
\bibitem [{\citenamefont {K.}(2024)}]{RigorousMLE}%
  \BibitemOpen
  \bibfield  {author} {\bibinfo {author} {\bibfnamefont {A.}~\bibnamefont {K.}},\ }\href@noop {} {\bibinfo {title} {Rigorous-mle}},\ \bibinfo {howpublished} {\url{https://github.com/Additi-K/Rigorous-MLE}} (\bibinfo {year} {2024}),\ \bibinfo {note} {gitHub repository}\BibitemShut {NoStop}%
\bibitem [{\citenamefont {Diamond}\ and\ \citenamefont {Boyd}(2016)}]{diamond2016cvxpy}%
  \BibitemOpen
  \bibfield  {author} {\bibinfo {author} {\bibfnamefont {S.}~\bibnamefont {Diamond}}\ and\ \bibinfo {author} {\bibfnamefont {S.}~\bibnamefont {Boyd}},\ }\bibfield  {title} {\bibinfo {title} {{CVXPY}: {A} {P}ython-embedded modeling language for convex optimization},\ }\href@noop {} {\bibfield  {journal} {\bibinfo  {journal} {Journal of Machine Learning Research}\ }\textbf {\bibinfo {volume} {17}},\ \bibinfo {pages} {1} (\bibinfo {year} {2016})}\BibitemShut {NoStop}%
\bibitem [{\citenamefont {Agrawal}\ \emph {et~al.}(2018)\citenamefont {Agrawal}, \citenamefont {Verschueren}, \citenamefont {Diamond},\ and\ \citenamefont {Boyd}}]{agrawal2018rewriting}%
  \BibitemOpen
  \bibfield  {author} {\bibinfo {author} {\bibfnamefont {A.}~\bibnamefont {Agrawal}}, \bibinfo {author} {\bibfnamefont {R.}~\bibnamefont {Verschueren}}, \bibinfo {author} {\bibfnamefont {S.}~\bibnamefont {Diamond}},\ and\ \bibinfo {author} {\bibfnamefont {S.}~\bibnamefont {Boyd}},\ }\bibfield  {title} {\bibinfo {title} {A rewriting system for convex optimization problems},\ }\href@noop {} {\bibfield  {journal} {\bibinfo  {journal} {Journal of Control and Decision}\ }\textbf {\bibinfo {volume} {5}},\ \bibinfo {pages} {42} (\bibinfo {year} {2018})}\BibitemShut {NoStop}%
\bibitem [{Note4()}]{Note4}%
  \BibitemOpen
  \bibinfo {note} {Specifically, we use the \protect \texttt {qmt\protect \_torch} function from their repository, which automatically switches from GPU to CPU when $n\ge 12$ because ``torch does not support more dimensional operations,'' though the details of this limitation are perhaps due to memory limitations rather than dimension limits.}\BibitemShut {Stop}%
\bibitem [{\citenamefont {Seshadri}\ \emph {et~al.}(2024)\citenamefont {Seshadri}, \citenamefont {Ringbauer}, \citenamefont {Spainhour}, \citenamefont {Monz},\ and\ \citenamefont {Becker}}]{seshadri2024theory}%
  \BibitemOpen
  \bibfield  {author} {\bibinfo {author} {\bibfnamefont {A.}~\bibnamefont {Seshadri}}, \bibinfo {author} {\bibfnamefont {M.}~\bibnamefont {Ringbauer}}, \bibinfo {author} {\bibfnamefont {J.}~\bibnamefont {Spainhour}}, \bibinfo {author} {\bibfnamefont {T.}~\bibnamefont {Monz}},\ and\ \bibinfo {author} {\bibfnamefont {S.}~\bibnamefont {Becker}},\ }\bibfield  {title} {\bibinfo {title} {Theory of versatile fidelity estimation with confidence},\ }\href@noop {} {\bibfield  {journal} {\bibinfo  {journal} {Physical Review A}\ }\textbf {\bibinfo {volume} {110}},\ \bibinfo {pages} {012431} (\bibinfo {year} {2024})}\BibitemShut {NoStop}%
\end{thebibliography}%

\end{document}